\documentclass[11pt]{amsart}

\usepackage[utf8]{inputenc}
\usepackage{amssymb,amsmath,amsthm,amscd,slashed,mathrsfs,amsfonts}
\usepackage{graphicx,hyperref,eucal}

\numberwithin{equation}{section}

\newtheorem*{Haag's Theorem}{Haag's Theorem}  
\newtheorem{Theorem}{Theorem}[section]
\newtheorem{Proposition}[Theorem]{Proposition}
\newtheorem{Corollary}[Theorem]{Corollary}                     
\newtheorem{Claim}[Theorem]{Claim}
\newtheorem{Definition}[Theorem]{Definition}

\newcommand{\Al}{\mathcal{A}}    
\newcommand{\B}{\mathscr{B}}     
\newcommand{\C}{\mathbb{C}}     
\newcommand{\D}{\mathfrak{D}}    \newcommand{\De}{\mathscr{D}}  

\newcommand{\F}{\mathscr{F}}      \newcommand{\fl}{\overline{V}_{+}}     
\newcommand{\Hi}{\mathfrak{H}}    \newcommand{\Ha}{\mathcal{H}}  
\newcommand{\G}{\mathcal{G}}             
   
\newcommand{\Lo}{\mathscr{L}^{\uparrow}_{+}}       \newcommand{\La}{\mathcal{L}}
\newcommand{\Mi}{\mathbb{M}} 
\newcommand{\N}{\mathbb{N}}
\newcommand{\Op}{\mathcal{O}}    \newcommand{\ol}{\overline}

\newcommand{\Po}{\mathscr{P}^{\uparrow}_{+}}
\newcommand{\Pd}{\mathcal{P}}
\newcommand{\R}{\mathbb{R}}       
\newcommand{\Sw}{\mathscr{S}}

\newcommand{\id}{\mbox{id}}        
\newcommand{\supp}{\mbox{\small supp}}     \newcommand{\Sum}{\mbox{$\sum$}}
\newcommand{\te}{\otimes}      \newcommand{\Tu}{\mathcal{T}}    \newcommand{\Ti}{\mathsf{T}}  

\newcommand{\U}{\mathsf{U}}    
\newcommand{\W}{\mathcal{W}}   \newcommand{\wt}[1]{\widetilde{#1}}  

\newcommand{\hs}[1]{\hspace{#1cm} }
\newcommand{\graph}[5]{\hspace{#1cm} \raisebox{#2cm}{\includegraphics[scale=#3]{#4.pdf}} \hspace{#5cm} }
\newcommand{\la}{\langle}   \newcommand{\ra}{\rangle}

\begin{document}

\title{Haag's theorem in renormalised quantum field theories}

\author{Lutz Klaczynski}  

\address{Department of Physics, Humboldt University Berlin, 12489 Berlin, Germany}
\email{klacz@mathematik.hu-berlin.de}
\date{\today}

\begin{abstract}
We review a package of no-go results in axiomatic quantum field theory with Haag's theorem at its centre. 
Since the concept of operator-valued distributions in this framework comes very close to what we believe canonical quantum fields are about, 
these results are of consequence to quantum field theory: they suggest the seeming absurdity that this highly victorious theory is 
incapable of describing interactions. We single out unitarity of the interaction picture's intertwiner as the most salient provision 
of Haag's theorem and critique canonical perturbation theory to argue 
that renormalisation bypasses Haag's theorem by violating this very assumption. 
\end{abstract}

\maketitle

\tableofcontents

\section{Introduction}\label{sec:Intro}                                                              

Quantum field theory (QFT) is undoubtedly one of the most successful physical theories. Besides the often 
cited extraordinary precision with which the anomalous magnetic moment of the electron had been computed 
in quantum electrodynamics (QED), this framework enabled theorists to \emph{predict the existence 
of hitherto unknown particles}. 

As Dirac was trying to make sense of the negative energy solutions of the 
equation which was later named after him, he proposed the existence of a positively charged 
electron \cite{Schw94}. 
This particle, nowadays known as the positron, presents an early example of a so-called \emph{antiparticle}. 
It is fair to say that it was the formalism he was playing with that \emph{led} him to think of 
such entities. And here we have theoretical physics at its best: the formulae under investigation only make 
sense provided an entity so-and-so exists. Here are Dirac's words:

\begin{quotation}
 ''A hole, if there were one, would be a new kind of particle, unknown to experimental physics, having the same mass and opposite
 charge to the electron. We may call such a particle an anti-electron. We should not expect to find any of them in nature, on account
 of their rapid rate of recombination with electrons, but if they could be produced experimentally in high vacuum, they would be quite stable
 and amenable to observations''\cite{Dir31}.
\end{quotation}

Of course, the positron was not the only particle to be predicted by quantum field theory. W and Z bosons,
ie the carrier particles of the weak force, both bottom and top quark and probably 
also the Higgs particle are all examples of matter particles whose existence was in some sense 
necessitated by theory prior to their discovery.  

Yet canonical QFT presents itself as a stupendous and intricate jigsaw puzzle. While some massive chunks 
are for themselves coherent, we shall see that some connecting pieces are still only tenuously locked, 
though simply taken for granted by many practising physicists, both of phenomenological \emph{and}
of theoretical creed.

\subsection{Constructive and axiomatic quantum field theory}

In the light of this success, it seems ironic that so far physically realistic quantum field 
theories like the standard model (SM) and its subtheories quantum electrodynamics (QED) and quantum chromodynamics (QCD) all 
defy a mathematically rigorous description \cite{Su12}. 

Take QED. While gauge transformations are classically well-understood 
as representations of a unitary group acting on sections of a principle bundle \cite{Blee81}, it is not 
entirely clear what becomes of them once the theory is quantised \cite{StroWi74,Stro13}. However, 
Wightman and G\aa rding have shown that the quantisation of the free electromagnetic field due to Gupta and Bleuler 
is mathematically consistent in the context of Krein spaces (see \cite{Stro13} and references there, p.156).

Drawing on the review article \cite{Su12}, we make the following observations as to what the state of affairs
broadly speaking currently is. 

\begin{itemize}
 \item \textsc{First}, all approaches to construct quantum field models in a way seen as mathematically sound and rigorous employ methods from operator 
       theory and stochastic analysis, the latter only in the Euclidean case.  
\end{itemize}
This is certainly natural given the corresponding heuristically very
successful notions used in Lagrangian quantum field theory and the formalism of functional integrals. 

These endeavours are widely known under the label \emph{constructive quantum field theory}, where a common objective of those approaches was to obtain 
a theory of quantum fields with some reasonable properties. \emph{Axiomatic quantum field theory} refined these properties further to a system 
of axioms. Several more or less equivalent such axiomatic systems have been proposed, the most 
prominent of which are: 
\begin{enumerate}
 \item the so-called Wightman axioms \cite{StreatWi00,Streat75},
 \item their Euclidean counterparts Osterwalder-Schrader axioms \cite{OSchra73, Stro13} and
 \item a system of axioms due to Araki, Haag and Kastler \cite{HaKa64,Ha96,Stro13}.
\end{enumerate}
These axioms were enunciated in an attempt to clarify and discern what a 
quantum field theory should or could reasonably be. 

In contrast to this, the proponents of the somewhat 
idiosyncratic school of \emph{axiomatic S-matrix theory} tried to discard the notion of quantum fields all together by setting axioms for 
the S-matrix \cite{Sta62}. 
However, it lost traction when it was trumped by QCD in describing
the strong interaction and later merged into the toolshed of string theory \cite{Ri14}. 

\begin{itemize}
 \item \textsc{Second}, efforts were made in two directions. In the constructive approach, models were built and then proven to conform with 
 these axioms \cite{GliJaf68,GliJaf70}, whereas on the axiomatic side, the general properties of quantum fields defined in such a way were investigated 
 under the proviso that they exist.
\end{itemize}
Among the achievements of the axiomatic community are rigorous proofs of the PCT and also the spin-statistics theorem \cite{StreatWi00}.

\begin{itemize}
 \item \textsc{Third}, within the constructive framework, the first attempts started with superrenormalisable QFTs to stay clear of ultraviolet (UV) 
       divergences. 
\end{itemize}
The emerging problems with these models had been resolved immediately:
the infinite volume divergences encountered there were cured by a finite number of subtractions, 
once the appropriate counterterms had been identified \cite{GliJaf68,GliJaf70}. 

\begin{itemize}
 \item \textsc{Fourth}, however, these problems exacerbated to serious and to this day unsurmountable obstructions as soon as the realm of renormalisable 
       theories was entered.
\end{itemize}
In the case of $(\phi^{4})_{d}$, the critical dimension
turned out to be $d=4$, that is, rigorous results were attained only for the cases $d=2,3$. The issue
there is that UV divergences cannot be defeated by a finite number of substractions. To our mind, it is 
their 'prolific' nature which lets these divergences preclude any nonperturbative treatment in the spirit 
of the constructive approach. For this introduction, suffice it to assert that a nonperturbative definition
of renormalisation for renormalisable fields is clearly beyond constructive methods of the above type. 

In particular, the fact that the regularised renormalisation $Z$ factors can only be expected to have asymptotic perturbation series is 
obviously not conducive to their rigorous treatment. Although formally appearing in nonperturbative
treatments as factors, they can a priori only be defined in terms of their perturbation series \cite{Ost86}. 

However, for completeness, we mention \cite{Schra76} in which a possible path towards the (in some sense 
implicit) construction of $(\phi^{4})_{4}$ in the context of the lattice approach was discussed.
As one might expect, the remaining problem was to prove the existence of the renormalised limit to the 
continuum theory.

\subsection{Haag's theorem and other triviality results} 

Around the beginning of the 1950s, soon after QED had been successfully laid out and heuristically shown to be renormalisable by its founding fathers \cite{Dys49b}, 
there was a small group of mathematical physicists who detected inconsistencies in its formulation. 

Their prime concern turned out to be the \emph{interaction picture} of a quantum field theory \cite{vHo52,Ha55}. In particular, Haag concluded that it cannot exist unless 
it is trivial, ie only describing a free theory. Rigorous proofs for these suspicions could at the time not be given for a simple reason: in order to 
prove that a mathematical object does not exist or that it can only have certain characteristics, one has to say and clarifiy what kind of mathematical 
thing \emph{it} actually is or what it is supposed to be.   

But the situation changed when QFT was put on an axiomatic footing by Wightman and collaborators who made a number of reasonable assumptions and proved
that the arguments put forward earlier against the interaction picture and \emph{Dyson's matrix} were well-founded \cite{WiHa57}. This result was then
called \emph{Haag('s) theorem}. It entails in particular that if a quantum field purports to be unitarily equivalent to a free field, it must be free 
itself.

Other important issues were the canonical (anti)commuation relations and the ill-definedness of quantum fields at sharp spacetime points. The ensuing 
decade brought to light a number of triviality results of the form ''If X is a QFT with properties so-and-so, then it is trivial'', where 'trivial'
comes in 3 types, with increasing strength: the quantum fields are free fields, identity operators or vanishing. We shall see examples of all three types
in this work.

The alternative formalism involving path integrals, although plagued by ill-definedness from the start \cite{AlHoMa08}, proved to be viable for lower 
spacetime dimensions in a Euclidean formulation \cite{GliJaf81}. Schrader showed that a variant of Haag's triviality verdict also emerges there, albeit 
in a somewhat less devastating form \cite{Schra74}.

\subsection{What to make of it}

But in the light of the above-mentioned success of quantum field theories, the question about what to make of it is unavoidable. 
People found different answers.

On the physics side, the no-go results were widely ignored (apart from a few exceptions) or misunderstood and belittled as mathematical footnotes to the 
success story that QFT surely is (there will be quotes in the main text). 
Only confirming this, the author had several conversations with practising theoretical physicists (young and middle-aged) who had never heard of Haag's theorem and 
pertinent results\footnote{We do not claim this to be representative, but we believe it is.}. 

On the mathematical physics side, the verdict was accepted and put down to the impossibility to implement relativistic quantum interactions in Fock space.
And indeed, without much mathematical expertise, the evidence is clear: the UV divergences encountered in perturbation theory leave no doubt that something
must be utterly wrong. Some were of the opinion that renormalisation only distracted the minds away from trying to find an appropriate new 
theory, as Buchholz and Haag paraphrase Heisenberg's view in \cite{BuHa00}. 

Our philosophical stance on this is that \emph{renormalised quantum field theory}, despite being a puzzle, provides us with 
\emph{peepholes} through which we are allowed to glimpse at least some parts of that 'true' theory. 
Moreover, renormalisation follows rules which have a neat underlying algebraic structure and are not those of a random whack-a-mole game.

There are a vast number of more or less viable attempts to give QFT a sound mathematical meaning. Neither did we have the space nor the expertise to do all
of them justice and include them here in our treatment. We have therefore chosen to direct our focus on the axiomatic approach: 
Haag's theorem was not just first formulated in this context \cite{StreatWi00} but it is, as we find, conceptually closer to the canonical Lagrangian theory than any 
other competing formalism. 

\subsection{Outline} 

This work is to a large extent a review and a collection of material scattered widely over the research literature. 
Apart from some rearrangements and minor changes, it comprises the first part of the author's PhD thesis \cite{Kla15}. 
Experts in axiomatic quantum field theory will find a compendium of the bits and pieces they (for the most part) 
already know, while all other readers will learn of some interesting aspects from axiomatic quantum field theory. 
The text consists roughly of 3 parts. 

Part \ref{part1} (§§\ref{sec:RepIssH} to \ref{sec:whatodo}) covers the salient material we got hold of concerning the history of Haag's theorem. 
§\ref{sec:RepIssH} takes the reader on a journey through the history of Haag's theorem and the preceding developments that led up to it, 
starting with what we call 'the representation issue'. 
In §\ref{sec:OthVersHT} we discuss some other versions of Haag's theorem that emerged in subsequent decades, while in §\ref{sec:EuHTHM} we describe 
an intriguing Euclidean variant showing how the triviality dictum is in fact circumvented. §\ref{sec:whatodo} has an account of the reactions from 
mathematical and theoretical physicists. It documents how Haag's theorem slowly slid into oblivion. 

Part \ref{part2} (§§\ref{sec:Fock} to \ref{sec:AxG}) explores what is wrong with the canonical formalism and has an introduction to the Wightman axioms.  
As is well-known, the canonical formalism is no mathematically coherent framework.
For, if one tries to translate its notions, in particular the idea of an \emph{operator field}, into the 
language of operator theory, mathematical inconsistencies arise. The general pattern of the 
associated problems is that whenever an object is overfraught with conditions, a canonical computation 
brings about nonsensical results, while a strict mathematical treatment finds that the object must be 
trivial to maintain well-definedness. §\ref{sec:Fock} expounds the intricasies of implementing interactions in Fock space and §\ref{sec:PowTHM} reviews
widely unknown negative results about the canonical (anti)commutation relations by Powers and Baumann which suggest that interacting fields cannot be 
expected to adhere to them. §\ref{sec:WaveRC} is a critical discussion of 
the wave-function renormalisation $Z$ factor's status as a nonperturbative object. §\ref{sec:CanQuant} surveys a triviality theorem by Wightman which serves as  
motivation for the introduction of the Wightman axioms, presented and discussed in §\ref{sec:WightAx}. The reader will then be equipped with the 
necessary notions to understand the details of the proof of Haag's theorem, presented and scrutinised in §\ref{sec:ProofHTHM}. 
However, Wightman's system of axioms is hardly capable of accommodating gauge theories, as the pertinent results presented in §\ref{sec:AxG} show. 

Finally, Part \ref{part3} (§§\ref{sec:GeMaLo} to \ref{sec:RCircH}) goes through the song and dance of how the highly successful 
canonical formalism manages to describe interactions by renormalisation. §\ref{sec:GeMaLo} discusses the Gell-Mann Low formula as the very 
identity that Haag's theorem is directly at odds with. Time ordering, a crucial feature of the Gell-Mann Low formula, does not allow to answer the
question as to whether the canonical commutation relations are satisfied. We have included an elementary treatment of this question in §\ref{sec:CCRQuest}.
Another problem incurred by time ordering are the divergences encountered in perturbation theory which we discuss in §\ref{sec:DivInt}. We then 
review how the canonical formalism resolves this issue in §\ref{sec:Ren} to prepare the ground for showing in §\ref{sec:RCircH} how Haag's theorem 
is essentially evaded: by an effective mass shift embodied by the self-energy. 

The appendix contains some proofs we deemed too technical for the main text, included here merely for the readers' convenience.

\part{Historical account}\label{part1}

\section{The representation issue \& Haag's theorem}\label{sec:RepIssH}                              
Contrary to what the theorem's name suggests, it is the result of a collective effort and not of a single author. Yet it was 
Haag who put out the seminal paper in which some of the mathematical problems of the canonical formalism 
were first circumscribed, in particular those associated to the interaction picture representation of 
quantum field theory (QFT). 

\subsection{Inequivalent representations}\label{sec:IneqRep}   

Right at the outset, when quantum mechanics came into being in the 1920s, there was what one may call 
the \emph{representation problem}. At the time, no one saw that the two competing formalisms - wave 
mechanics as developed by Schr\"odinger and matrix mechanics put forward by Heisenberg, Born
and Jordan - were in fact equivalent. 

Yet their proponents hardly appreciated each other's work. 
In 1926, Einstein wrote in a letter to Schr\"odinger, that he was convinced 
``that you have made a decisive advance with your quantum condition, just as I am equally convinced
that the Heisenberg-Born route is off the track''. Soon after, Schr\"odinger remarked in a note
to a paper that their route left him ``discouraged, if not repelled, 
by what appeared ... a rather difficult method ... defying visualisation'' \cite{Rue11}. And
Heisenberg told Pauli, ``the more I think of the physical
part of the Schr\"odinger theory, the more detestable I find it. What Schr\"odinger writes 
about visualisation makes scarcely any sense, in other words I think it is shit\footnote{ ''Ich finde
es Mist``. (see \cite{Strau01})}.''\cite{Rue11}

\subsection*{Canonical commutation relations} 
However, both formalisms had something in common: they were dealing with an algebra generated by 
operators $\{ Q_{1},...,Q_{n}\}$ and $\{ P_{1}, ..., P_{n} \}$ corresponding
to the canonical position and momentum variables of Hamiltonian mechanics, which satisfy the 
\emph{canonical commutation rules} (or \emph{relations})
\begin{equation}\label{QMCCR}
[Q_{j},Q_{l}]=0 = [P_{j},P_{l}] \ ,  \hs{1}  [Q_{j},P_{l}]= i \delta_{jl}  \ \hs{2} \mbox{(CCR)}
\end{equation}
for all $j,l \in \{1,...,n\}$ on a Hilbert space $\Hi$. In Heisenberg's matrix mechanics, these objects
are matrices with infinitely many entries, whereas in Schr\"odinger's wave mechanics they are represented
by the two operators
\begin{align}\label{SchRep}
 (Q_{j}\psi)(x) = x_{j} \psi(x) \ , \hs{2} (P_{l} \psi)(x) = -i \partial_{l} \psi(x) 
\end{align}
which act on square-integrable wavefunctions $\psi \in L^{2}(\R^{n})=\Hi$ (see any textbook on quantum 
mechanics, eg \cite{Strau13}). 

\subsection*{Stone-von Neumann theorem} The dispute over which theory was the right one was settled 
when von Neumann took the cues given to him by Stone and proved in 1931
that both formulations of quantum mechanics are \emph{equivalent} in the sense that 
both are \emph{unitarily equivalent representations of the canonical commutation rules} (\ref{QMCCR}) 
if their exponentiations
\begin{align}
 U(a)=\exp(ia \cdot P) \ ,  \hs{2} V(b)=\exp(ib \cdot Q)    \hs{1} a,b \in \R^{n}
\end{align}
are so-called \emph{Weyl unitaries} \cite{vNeu31}. \index{Weyl unitaries} 
In the case of the Schr\"odinger representation, these Weyl unitaries are given by the two families of 
bounded operators defined through
\begin{align}
 (U_{S}(a)\psi)(x) = \psi(x+a) \hs{2} (V_{S}(b)\psi)(x) = e^{i b \cdot x} \psi(x) 
\end{align}
for $\psi \in L^{2}(\R^{n})$. The CCR (\ref{QMCCR}) now take what is called the \emph{Weyl form of the CCR},
\begin{equation}
  U_{S}(a)V_{S}(b)=e^{i a \cdot b} V_{S}(b)U_{S}(a) \ \hs{2} \mbox{(Weyl CCR)} .
\end{equation}
The Stone-von Neumann theorem makes the assertion that all Weyl unitaries conforming with these relations
are unitarily equivalent to a finite direct sum of Schr\"odinger representations:  

\begin{Theorem}[Stone-von Neumann]\label{SvN}
 Let $\{ U(a): a \in \R^{n} \}$ and $\{ V(b): b \in \R^{n} \}$ be irreducible Weyl unitaries 
 on a separable Hilbert space $\Hi$, ie two weakly continuous families of unitary operators such that 
 $U(a)U(b)=U(a+b)$, $V(a)V(b)=V(a+b)$ and 
\begin{align}\label{WCCR}
 U(a)V(b)=e^{i a \cdot b} V(b)U(a)  \hs{2} \mbox{(Weyl CCR)}
\end{align}
for all $a,b \in \R^{n}$. Then there is a Hilbert space isomorphism $W \colon \Hi \rightarrow L^{2}(\R^{n})$
such that
\begin{align}
 WU(a)W^{-1} = U_{S}(a) \hs{1} WV(a)W^{-1}=V_{S}(a),
\end{align}
where $U_{S}$ and $V_{S}$ are the Schr\"odinger representation Weyl unitaries. If the above Weyl unitaries in (\ref{WCCR})
are reducible, then each irreducible subrepresentation is unitarily equivalent to the Schr\"odinger 
representation.
\end{Theorem}
\begin{proof}
 See \cite{vNeu31} or any book on the mathematics of quantum mechanics, for example \cite{Em09}.
\end{proof}
The reason why this theorem had to be phrased in terms of the Weyl CCR (\ref{WCCR}) and not the CCR is, 
as von Neumann pointed out in \cite{vNeu31}, that the CCR (\ref{QMCCR}) can certainly not hold on the 
whole Hilbert space $L^{2}(\R^{n})$ since the operators in (\ref{SchRep}) are \emph{unbounded}.
This is easy to see: $\int d^{n}x \ |\psi(x)|^{2} < \infty$ does \emph{not} imply 
$\int d^{n}x \ |x_{j}\psi(x)|^{2} < \infty$. Moreover, if (\ref{QMCCR}) were valid everywhere in 
the Hilbert space, one could take the trace of both sides yielding a contradiction\footnote{I thank David Broadhurst for pointing this out to me.}.

Since both the CCR algebra of the Schr\"odinger and the Heisenberg representation of quantum mechanics
generate irreducible Weyl unitaries, the issue was indeed settled.
However, some questions remained:
\begin{itemize}
 \item are there representations of the CCR (\ref{QMCCR}) which do \emph{not}
       generate Weyl unitaries and are therefore not unitarily equivalent to the Schr\"odinger and hence
       also not to the Heisenberg representation?
 \item On what conditions do they give rise to a representation of the Weyl CCR? 
\end{itemize}
Dixmier \cite{Dix58} found one particular answer to this latter question. 

\begin{Theorem}[Dixmier]\label{Dixm}
Let $Q,P$ be two closed symmetric operators on a Hilbert space $\Hi$ with common stable domain $\D$, ie
$P \D \subset \D$ and $Q \D \subset \D $. Assume the operator 
\begin{equation}
H=P^{2} + Q^{2}
\end{equation}
is essentially self-adjoint on $\Hi$. If $Q$ and $P$ satisfy the CCR algebra $[Q,P]=i$
on $\D$ then $\Hi$ decomposes into a direct sum of subspaces on each of which their restrictions are unitarily 
equivalent to the Schr\"odinger representation.  
\end{Theorem}

That $H=P^{2} + Q^{2}$ is essentially self-adjoint seems physically reasonable as this operator corresponds
to the Hamiltonian of the harmonic oscillator, the much beloved workhorse of quantum mechanics. 
In case the assumptions of Dixmier's theorem are not given, a number
of examples of representations which are unitarily inequivalent to the Schr\"odinger representation have
been found \cite{Su01}. 

Surely, most examples of inequivalent representations are physically pathological. 
However, the interesting question is whether there is an example of physical relevance. And, yes, there is.
Reeh found one such example \cite{Ree88}: a (nonrelativistic, quantum mechanical) electron 
in the exterior of an infinitely long cylinder 
with a magnetic flux running through it. 

To arrive at the model, one has to let the cylinder become 
infinitely thin, that is, in Reeh's description, become the $z$-axis. In doing so, he 
clearly stayed within the range of acceptable habits of a theoretical physicist. 

Because the system is translationally invariant along the $z$-axis, there are only two degrees 
of freedom, ie in the above setting of (\ref{QMCCR}) we have $n=2$. The canonical momentum operators are 
\begin{equation}
 p_{x} = -i \partial_{x} + e A_{1}(x,y) , \hs{2} p_{y} = -i \partial_{y} + e A_{2}(x,y),
\end{equation}
where $A_{1}(x,y)$ and $A_{2}(x,y)$ are the components of the electromagnatic vector potential 
and $e$ is the electron's charge. 

This particular example, although it satisfies the CCR, is \emph{not 
unitarily equivalent to the Schr\"odinger representation}. Regarding Theorem \ref{Dixm}, Reeh closes his paper by pointing out 
that $p_{x}^{2}+p_{y}^{2}$ is not essentially self-adjoint, in agreement with Dixmier's result.
What we learn from this is that 
\begin{itemize}
 \item \textsc{firstly}, even when the system has a finite number of degrees of freedom, not all representations of the CCR (\ref{QMCCR}) are unitarily equivalent
      to the Schr\"odinger representation and
 \item \textsc{secondly}, this need not worry us. It rather suggests that unitary equivalence is too strong a notion for physical equivalence.
\end{itemize}
Reeh's example suggests that we abandon the view that every quantum-mechanical system should lie in the unitary equivalence class of the Schr\"odinger 
representation.

In fact, a much weaker yet still physically sensible notion of equivalence has been put 
forward by Haag and Kastler in \cite{HaKa64}. The authors essentially propose a form of weak equality 
of operators, namely that two observables $A$ and $B$ are equivalent if their matrix elements cannot 
be distinguished by measurement, that is, for a subset $\D$ of state vectors, which describe the set 
of all possible experimental setups, one has
\begin{equation}
| \la \Psi | (A-B) \Phi \ra | < \varepsilon  \hs{2} \forall \Psi, \Phi \in \D ,
\end{equation}
in which $\varepsilon>0$ is below any conceivable lower measuring limit\footnote{The mathematical notion behind this is that of a convex topology
induced by a system of \emph{seminorms}: each pair of elements in $\D$ defines a seminorm.}. 

\subsection{Fock space}\label{sec:Focks} 
It was in one of the early papers on quantum field theory (QFT) in 1929 by Heisenberg and Pauli
that the notion of what is nowadays known as \emph{Fock space} first emerged \cite{HeiPau29}. A bit later, 
this concept was explored more completely by Fock \cite{Fo32} \index{Fock space} and rephrased 
in rigorous mathematical form by Cook \cite{Co53}. 
This setting seemed to be appropriate and make sense even for relativistic particles. 
Because the Schr\"odinger representation of nonrelativistic many-particle systems can also be phrased in these terms, the Fock space 
became 'the Schr\"odinger representation of QFT' (see eg \cite{Di11}). 

Since Reeh's counterexample of a perfectly physical but nevertheless non-Schr\"odinger representation
of the CCR in quantum mechanics was discovered rather late (1988!) and was not known at the time, the 
representation issue continued to be given plenty of attention.
 
It became topical again in the 1950s when Friedrichs constructed what he called \emph{myriotic 
representations} of the CCR (\ref{QMCCR}) on Fock space, also known as 'strange representations'.
These representations are defined by the absence of any number operator and are obtained by passing to the limit of a countably infinite number of
degrees of freedom \cite{Fried53}, ie $n \rightarrow \infty$ in (\ref{QMCCR}). The Stone-von Neumann theorem is in 
this case no longer applicable. 

In 1954, Wightman and G\aa rding published results proving that for this limit, there exists, as they put it, a 
``maze of irreducible inequivalent representations'' \cite{GaWi54}. However, we shall see in §\ref{sec:PowTHM} that neither the CCR for bosons nor 
the anticommutation relations (CAR) for fermions are features that fully interacting theories can be expected to possess, at least in $d\geq 4$ spacetime
dimensions. 
Therefore, the representation problem may actually be a \emph{pseudo problem}. Of course, in the case of quantum mechanics, the CCR express the
fundamental \emph{Heisenberg uncertainty principle} which one is not willing to abandon.  
However, since there is no analogue of the position operator in QFT, it is not clear how this principle 
can be implemented through the CCR or CAR in a relativistic 
quantum theory\footnote{Although for example the energy-time uncertainty is generally presumed to be true 
and employed in the interpretation of virtual off-shell particles in Feynman diagrams, there is no obvious
connection to the CCR/CAR in QFT.}!

\subsection*{Van Hove phenomenon} 
Prior to these developments, van Hove was one of the first authors who tried to rigorously 
define a Hamiltonian of a massive interacting scalar field $\phi$. The interaction he studied consists of a 
finite number of fixed point sources \cite{vHo52}, the Hamiltonian being
\begin{equation}
 H_{g}= H_{0} + g H_{I},
\end{equation}
where $H_{0}=\sum_{\mathbf{k}}\epsilon_{\mathbf{k}} a_{\mathbf{k}}^{\dagger}a_{\mathbf{k}}$ is the free and 
$H_{I}=\sum_{s=1}^{l}\beta_{s}\phi(\mathbf{x}_{s})$ the interacting part. The point sources sit at 
positions $\mathbf{x}_{s}$ with strength $g \beta_{s}$. Introducing a momentum cutoff $\kappa > 0$ such that 
$\epsilon_{\mathbf{k}}=0$ if $|\mathbf{k}| > \kappa$, he considered the two vacuum states $\Phi_{0}$
and $\Phi_{0}(g)$ of $H_{0}$ and $H_{g}$, respectively, and found for their overlap
\begin{equation}
 \la \Phi_{0} | \Phi_{g} \ra \rightarrow 0 \hs{1} \mbox{as} \hs{0.5} \kappa \rightarrow \infty
 \hs{2} (\mbox{'van Hove phenomenon'}) ,
\end{equation}
ie when the cutoff was removed by taking the limit, the Hilbert spaces of states turned out to be 
orthogonal for $H_{0}$ and $H_{g}$. He also found this to be the case for the energy eigenstates 
\begin{equation}
 H_{g} \Phi_{n}(g) = E_{n}(g)  \Phi_{n}(g)  
\end{equation}
of energy $E_{n}(g)$ for different values of $g$, ie $\la \Phi_{n}(g') | \Phi_{m}(g) \ra \rightarrow 0$
as the cutoff was removed, for all $n,m$ and $g' \neq g$. 

He concluded that ``the stationary states of the field interacting with the sources are no 
linear combinations of the stationary states of the free theory''. A first sign that something
may be wrong with the interaction picture, as Coleman wrote in a short review of van Hove's paper
``it suggests that there is no mathematical justification for using the interaction
representation and that the occasional successes of renormalization methods are lucky flukes 
...''\footnote{Coleman's review of van Hove's paper is available at 
www.ams.org/mathscinet, keywords: author ``van Hove'', year 1952. } .

\subsection{Strange representations}\label{sec:strangeRep}
Such was the backdrop against which Haag argued in his seminal publication \cite{Ha55} that the 
interaction picture cannot exist. The salient points he made were the following. 

\textsc{First}, it is very easy to find strange representations of the CCR (\ref{QMCCR}) in the case of 
infinitely many degrees of freedom: a seemingly innocuous transformation like 
\begin{equation}\label{stratra}
 q_{\alpha} \mapsto \widetilde{q}_{\alpha} = c q_{\alpha} , \hs{2} 
 p_{\alpha} \mapsto \widetilde{p}_{\alpha} = c^{-1} p_{\alpha} 
\end{equation}
for any $c \notin \{0,1\}$ of the canonical variables $\{ q_{\alpha}, p_{\alpha} \}$ leads to a 
strange representation of the CCR, ie a representation for which 
there is no number operator and no vacuum state. 

\textsc{Second}, Dyson's matrix $V=\mathsf{U}(0,-\infty)$ \emph{cannot} exist, ie the operator that evolves 
interaction picture states from the infinitely far past at $t=-\infty$, where the particles are free, 
to the present at $t=0$, where they (may) interact. 

As regards the first point, let us see why (\ref{stratra}) produces a strange representation. 
We follow Haag and write $c = \exp(\varepsilon)$ with $\varepsilon \neq 0$. For the annihilators and 
creators,
\begin{equation}
 a_{\alpha} = \frac{1}{\sqrt{2}}(q_{\alpha}+i p_{\alpha}) , \hs{2} a^{\dagger}_{\alpha} = \frac{1}{\sqrt{2}}(q_{\alpha}-i p_{\alpha})
\end{equation}
this transformation takes the form of a 'Bogoliubov transformation':
\begin{align}\label{stratraCA}
\begin{split}
 a_{\alpha} \mapsto \widetilde{a}_{\alpha} &= \cosh \varepsilon \ a_{\alpha} + \sinh \varepsilon \ a^{\dagger}_{\alpha}  \\
 a^{\dagger}_{\alpha} \mapsto \widetilde{a}^{\dagger}_{\alpha} 
 &= \sinh \varepsilon \  a_{\alpha} + \cosh \varepsilon \ a^{\dagger}_{\alpha},
\end{split}
\end{align}
which is easy to check.  We let $\alpha=1,2,...,N$, where $N<\infty$ for the moment and $\Hi$ be the Fock space generated by applying
the elements of the CCR algebra $\la a_{\alpha}, a^{\dagger}_{\alpha} : \alpha \in \{1,...,N\} \ra_{\C}$ 
to the vacuum state. 'Strangeness' of the representation 
$\{ \widetilde{a}_{\alpha}, \widetilde{a}^{\dagger}_{\alpha} \}$ will be incurred only when the limit $N \rightarrow \infty$ is taken.
Notice that the generator of the above transformation (\ref{stratra}) is given by 
\begin{equation}\label{Bogol}
 T = \frac{i}{2}\sum_{\alpha=1}^{N}  [a^{\dagger}_{\alpha}a^{\dagger}_{\alpha}-a_{\alpha}a_{\alpha}],
\end{equation}
ie $V=\exp(i \varepsilon T)$ is the transformation such that 
$\widetilde{a}_{\alpha} = V a_{\alpha} V^{-1}$ and 
$\widetilde{a}^{\dagger}_{\alpha} = V a^{\dagger}_{\alpha} V^{-1}$
which is unitary as long as $N$ is finite. When $N \rightarrow \infty$, this operator will map any vector 
of the Fock representation to a vector with infinite norm. Let $\Psi_{0}$ denote the Fock vacuum of the
original representation. Due to
\begin{equation}
 \widetilde{a}_{\alpha} V \Psi_{0} = V a_{\alpha} V^{-1} V \Psi_{0} = V a_{\alpha} \Psi_{0} = 0  
\end{equation}
we see that the new representation does also have a vacuum. We denote it by 
$\widetilde{\Psi}_{0}:= V\Psi_{0}$.
We abbreviate $\tau_{\alpha} := a^{\dagger}_{\alpha}a^{\dagger}_{\alpha} - a_{\alpha}a_{\alpha}$ and 
compute the vacua's overlap,
\begin{equation}\label{overla}
\la \Psi_{0} | \widetilde{\Psi}_{0} \ra = \la \Psi_{0} | V\Psi_{0} \ra  
=  \la \Psi_{0} | (\prod_{\alpha=1}^{N} e^{-\frac{\varepsilon}{2} \tau_{\alpha}} ) \Psi_{0} \ra
= \prod_{\alpha=1}^{N} \la \Psi_{\alpha,0} | e^{-\frac{\varepsilon}{2} \tau_{\alpha}} \Psi_{\alpha,0} \ra,
\end{equation}
where we have used that the vacuum is a tensor product $\Psi_{0} = \Psi_{1,0} \te ... \te \Psi_{N,0}$ and 
$[\tau_{\alpha},\tau_{\beta}]=0$. 
Note that each factor in the product (\ref{overla}) yields the same value.
This value is below $1$ and it therefore vanishes in the limit $N \rightarrow \infty$, ie we find that 
the van Hove phenomenon occurs. 
Similiarly, one can show that, as Haag argues in \cite{Ha55}  
\begin{equation}\label{weakvan}
 \la \Psi | V \Phi \ra = 0   \hs{2} \mbox{for all} \hs{0.3} \Psi,\Phi \in \Hi \hs{1} (\mbox{Fock space})
\end{equation}
in the limit. Although the new CCR algebra (\ref{stratraCA}) is perfectly well-defined on $\Hi$,
its vaccum - if it exists - lies outside $\Hi$! Hence \emph{inside} $\Hi$, this algebra is a 
'strange representation', ie unitarily inequivalent to the Fock representation and has no vaccum. 

The lesson Haag took from this was that seemingly minor and prima facie innocuous changes can easily
lead to a theory which is unitarily equivalent when $N < \infty$, but ceases to be so in the limit
$N \rightarrow \infty$. 

Especially interesting is the vanishing overlap of the two vacua. If Dyson's matrix is well-defined for 
a finite number of degrees of freedom, then - as the above example shows - its existence is highly 
questionable in the case of an infinite system. A vanishing overlap of the two vacua 
directly contradicts what became known as the \emph{theorem of Gell-Mann and Low} (\cite{GeMLo51},
cf. §\ref{sec:GeMaLo}) which explicitly relies on a nonvanishing overlap: 
in the context of our example, the analogous statement is that up to a normalisation constant
\begin{equation}
\widetilde{\Psi}_{0} = \frac{V \Psi_{0}}{\la \Psi_{0} | V \Psi_{0} \ra }
\end{equation}
exists in the limit $N \rightarrow \infty$. Of course, the Bogoliubov transformation $V=e^{i \varepsilon T}$
as constructed in (\ref{Bogol}) bears no resemblance to Dyson's matrix. 
While Haag's example is therefore of no direct consequence for field theory, van Hove's indeed is: in his
model, 
\begin{equation}
\la \Phi_{n}(g') | \Phi_{m}(g) \ra = 0  \hs{2} (g \neq g')
\end{equation}
implies that $V$ vanishes weakly in $\Hi$ since $\Phi_{n}(g')=V\Phi_{n}(g)$.

\subsection{Dyson's matrix}\label{DysMa} 
However, Haag then went on to make the case that Dyson's matrix does not exist as follows (\cite{Ha55}, §4): 
let two Hermitian scalar fields $\phi_{1}(\mathbf{x})$ and $\phi_{2}(\mathbf{x})$ be  
related by a unitary map $V$ according to 
\begin{equation}\label{Unitrans}
 \phi_{1}(\mathbf{x}) = V^{-1} \phi_{2}(\mathbf{x}) V
\end{equation}
and suppose their time evolution is governed by different Hamiltonians $H_{1}\neq H_{2}$. 
Let $D(\mathbf{a})$ be a representation of the translation group in $\R^{3}$ under which both fields 
behave covariantly, ie
\begin{equation}\label{transRep}
 D(\mathbf{a}) \phi_{j}(\mathbf{x}) D(\mathbf{a})^{\dagger} = \phi_{j}(\mathbf{x}+\mathbf{a}) \hs{2} 
 (j=1,2).
\end{equation}
Then these conditions imply $[D^{\dagger}V^{-1}DV,\varphi_{1}]=0$ from which $D^{\dagger}V^{-1}DV=1$ and then also $[V,D(\mathbf{a})] =0$ follow 
(by irreducibility of the field algebra\footnote{We will explain this thoroughly in §\ref{sec:ProofHTHM}, cf.(\ref{Poin}).}). 
This means $[V,\mathbf{P}]=0$, where $\mathbf{P}$ is the three-component generator of translations in 
$\R^{3}$. Let $\Phi_{0j}$ be the vacuum of the representation of
$\phi_{j}(\mathbf{x})$ ($j=1,2$), ie $\Phi_{02} = V\Phi_{01}$. Then follows that $V$ is trivial because
\begin{equation}\label{HaagArg}
 \mathbf{P} \Phi_{02} = \mathbf{P} V \Phi_{01} = V \mathbf{P} \Phi_{01} = 0
\end{equation}
implies $\Phi_{01} = w  \Phi_{02}$ with $w \in \C$ since the vacuum is the only translation-invariant 
state. Haag now argues that this result is a contradiction to the assumption of both Hamiltonians 
having a different form, ie he means that $H_{1}\Phi_{01} = 0 = H_{2} \Phi_{01}$ is not acceptable if
both Hamiltonians are to be different. In this view, the conclusion is that $H_{1}=H_{2}$ and hence both
theories are the same. 

Mathematically, of course, this conclusion is not permissible if the two operators 
agree just on the vacuum. However, behind his statement ``In all theories considered so far 
[the statement $\Phi_{01} = w  \Phi_{02}$, author's note] is contradicted immediately by the form of the Hamiltonian.''(ibidem) he refers to the
more or less tacit assumption that one of the two operators should \emph{polarise the vacuum}.
This idea in turn originates in the fact that no one has ever seen an interacting Hamiltonian not 
constructed out of free fields, ie of annihilators and creators. All of those conceivable always have a term of creators only, a term incapable 
of annihilating the vacuum.

Yet the above argument given by Haag against Dyson's matrix is flawed. If the two theories have
different Hamiltonians, then they should in principle also have different total momenta, ie in the language of Lagrangian field theory 
\begin{equation}
 \mathbf{P}_{j} = - \int d^{3}x \ \pi_{j}(\mathbf{x}) \nabla \phi_{j}(\mathbf{x}),
\end{equation}
and we expect $\mathbf{P}_{1}=V^{-1}\mathbf{P}_{2}V$ (considering the complexity of the Poincaré algebra, the statement
$\mathbf{P}_{1}=\mathbf{P}_{2}$ is a rather strong assumption that needs discussion!). This suggest that the two fields should be covariant
with respect to different representations of the translation group which means their generators $\mathbf{P}_{1}, \mathbf{P}_{2}$ are not the same 
and consequently (\ref{transRep}) is bogus. 

\subsection{Results by Hall and Wightman}\label{sec:WightHall}
This is probably (we do not know) what Hall and Wightman had in mind when they wrote ``In the opinion of the present 
authors, Haag's proof is, at least in part, inconclusive.''(see \cite{WiHa57}, ref.10). 
But Haag did have a point there. 

Hall and Wightman 'polished and generalised' Haag's argument (as they add) to obtain a result which at first glance seems less harmful to Dyson's matrix
\cite{WiHa57}: 
because both theories are different, one should allow for two Hilbert spaces $\Hi_{1}$ and $\Hi_{2}$, 
each equipped with a representation of the Euclidean group (rotations and translations) 
$D_{j}(\mathbf{a},R)$ ($j=1,2$) in $\R^{3}$ such that 
\begin{equation}\label{EuclRep}
 D_{j}(\mathbf{a},R) \phi_{j}(\mathbf{x}) D_{j}(\mathbf{a},R)^{\dagger} = \phi_{j}(R\mathbf{x}+\mathbf{a})
 \hs{2} (j=1,2)
\end{equation}
and assume there exist invariant states $\Phi_{0j} \in \Hi_{j}$, that is, $D_{j}\Phi_{0j}=\Phi_{0j}$.
Then it follows by irreducibility of both field theories from (\ref{Unitrans}) that\footnote{We 
discuss the proof more thoroughly in §\ref{sec:ProofHTHM}. }
\begin{equation}\label{irrUV}
 D_{1}(\mathbf{a},R) = V^{-1} D_{2}(\mathbf{a},R) V \hs{1} 
 \mbox{and} \hs{1} V\Phi_{01}=\Phi_{02}.
\end{equation}
This seems to be a less devastating result because the map
$V \colon \Hi_{1} \rightarrow \Hi_{2}$ need not be trivial. Yet, as they subsequently showed, this 
cannot come to the rescue of Dyson's matrix either. If we just consider what it means for the $n$-point
functions,
\begin{equation}\label{etVEV}
 \begin{split}
 \la \Phi_{01} | \phi_{1}(\mathbf{x}_{1})...\phi_{1}(\mathbf{x}_{n}) \Phi_{01} \ra &= 
 \la \Phi_{01} | V^{-1} \phi_{2}(\mathbf{x}_{1})V... V^{-1} \phi_{2}(\mathbf{x}_{n})V \Phi_{01} \ra \\
 &= \la V \Phi_{01} | \phi_{2}(\mathbf{x}_{1})V... V^{-1} \phi_{2}(\mathbf{x}_{n})V \Phi_{01} \ra \\
 &= \la \Phi_{02} |\phi_{2}(\mathbf{x}_{1})... \phi_{2}(\mathbf{x}_{n})\Phi_{02} \ra,
 \end{split}
\end{equation}
we see that they agree. This entails for the Heisenberg fields $\phi_{j}(t,\mathbf{x})$
that their $n$-point functions coincide on the time slice $t=0$. To extend this to a larger subset, let now the condition of 
Euclidean covariance of the Schr\"odinger fields in (\ref{EuclRep})
be strengthened to \emph{relativistic covariance}, ie Poincaré covariance,
\begin{equation}\label{relCov}
 U_{j}(a,\Lambda) \phi_{j}(x) U_{j}(a,\Lambda)^{\dagger} = \phi_{j}(\Lambda x+a)
 \hs{2} (j=1,2),
\end{equation}
where $x=(t,\mathbf{x}), a \in \Mi$ are Minkowski spacetime points and 
$\Lambda$ a proper orthochronous Lorentz transformation, then for pairwise spacelike-distant 
$x_{1}, ... , x_{n} \in \Mi$ one has
\begin{equation}\label{nleq4pf}
 \la \Phi_{01} | \phi_{1}(x_{1})...\phi_{1}(x_{n}) \Phi_{01} \ra 
 = \la \Phi_{02} | \phi_{2}(x_{1})...\phi_{2}(x_{n}) \Phi_{02} \ra \hs{1} (\mbox{spacelike distances}).
\end{equation}
Hall and Wightman proved that for $n \leq 4$, this equality can be extended (in the 
sense of distribution theory) to all
spacetime points, where $x_{j} \neq x_{l}$ if $j \neq l$. 
This result is referred to as \emph{generalised Haag's theorem}. The term 'generalised' has been
used because none of the fields need be free for (\ref{nleq4pf}) to hold (we have not required any
of the fields to be free so far). 
The reason the authors could not prove this for higher $n$-point functions is that there is no
element within the Poincaré group capable of jointly transforming $m\geq 4$ Minkowski spacetime 
vectors\footnote{ Recall that by translation invariance an $n$-point function is a function 
of $m=n-1$ Minkowski spacetime points.} to $m$ arbitrary times but only to a subset which is not
large enough for a complete characterisation \cite{StreatWi00}. 

Notice that (\ref{etVEV}) holds for any quantum field, of whatever spin, it is a trivial 
consequence of (\ref{Unitrans}) and the assumed irreducibility of the field algebra. 
The only difference for fields of higher spin is that the transformation 
law (\ref{relCov}) needs a finite dimensional representation of the Lorentz group for spinor and 
vector fields and makes (\ref{nleq4pf}) less obvious. We shall discuss this point in §\ref{sec:AxG}
to see that it works out fine also in these cases.

The result (\ref{nleq4pf}) mattered and still matters because a field theory
was by then already known to be sufficiently characterised by its vacuum expectation values, as shown 
by Wightman's \emph{reconstruction theorem} put forward in \cite{Wi56}, which, in brief, says that a field theory can be (re)constructed from 
its vacuum expectation values. We shall survey this result in §\ref{sec:WightAx} and discuss the case of quantum electrodynamics (QED)
in §\ref{sec:AxG}. 

As the above arguments show, an interacting field theory cannot be unitarily equivalent to a free field theory unless we take the stance
that it makes sense for an interacting field to possess $n$-point functions that for $n\leq 4$ agree with those of a free field.  

\subsection{Contributions by Greenberg and Jost}\label{ConJoG} 
In fact, this stance is untenable since a couple of years later, Greenberg proved inductively in \cite{Gre59} that if one of the two 
fields is a free field, then the equality (\ref{nleq4pf}) holds for all $n$-point functions and all 
spacetime points. It is to this day still an open question whether (\ref{nleq4pf}) is 
true for two general (Hermitian) fields at arbitrary spacetime points \cite{Streat75}. 

Around the same time, using different arguments, another proof of Greenberg's result was obtained  
by Jost and Schroer \cite{Jo61} and is therefore known under the label
\emph{Jost-Schroer theorem}:
if a field theory has the same two-point function as a free field of mass $m>0$, then that is already
sufficient for it to be a free field of the same mass. It was also independently shown by Federbush and 
Johnson \cite{FeJo60} and the massless case was proved by Pohlmeyer \cite{Po69}. 

\subsection{Haag's theorem I} With these latter results, we arrive at what became known as Haag's theorem. Let us first state it in words, a more
thorough exposition including the proof will be given in §\ref{sec:ProofHTHM}:

\begin{Haag's Theorem}
If a scalar quantum field is unitarily equivalent to a free scalar quantum field, then, by virtue of the reconstruction theorem, it is also a free field  
because all vacuum expectation values coincide. 
\end{Haag's Theorem}

As a consequence, Dyson's matrices, which purportedly transform in a unitary fashion the incoming and the outgoing free asymptotic fields into fully 
interacting fields, cannot exist. Note that the equality of the vacuum expectation values for spacelike separations (\ref{nleq4pf}) 
needs no other provisions than 
\begin{itemize}
 \item unitary equivalence of the two fields through the intertwiner $V$,
 \item Poincaré covariance and
 \item irreducibility of their operator algebras to warrant (\ref{irrUV}).
\end{itemize} 
The remainder of the results, which say that the equality extends beyond spacelike separations into the entire Minkowski space $\Mi$,
and their provisions that complete Haag's theorem, eg the Jost-Schroer theorem, therefore have a different status!  

Because the proof of Haag's theorem is fairly technical and requires a number of assumptions, we defer it
to §\ref{sec:ProofHTHM}. As the above deliberations suggest, the spacetime dimension does not 
enter the discussion anywhere and is therefore irrelevant. 

Haag's theorem is a very deep and fundamental fact, true both for superrenormalisable and renormalisable theories. In its essence, it
is rather trivial: it is a theorem about a free field of fixed mass and its unitary equivalence class. 
In fact, due to Theorem \ref{freeH}, which we call 'Haag's theorem for free fields', we know that even two free fields 
lie in distinct equivalence classes whenever their masses differ, however \emph{infinitesimally} small that difference might be.    

\subsection{Galilean exemptions} 
Note that the step from equal-time vacuum expectation values (\ref{etVEV}) to (\ref{nleq4pf}) is not
permitted for Galilean quantum field theories as employed in solid state physics. In fact, 
\emph{Haag's theorem breaks down for Galilean quantum field theories} as the Jost-Schroer theorem does 
not hold for them:
Dresden and Kahn showed that there are nontrivial (=interacting) Galilean quantum field theories 
whose two-point Wightman functions are identical to those of free fields \cite{DresKa62}.
We therefore have reason to believe that Haag's theorem in the strict sense of the above stated theorem 
is specific to relativistic quantum field theories. 

However, \emph{Euclidean quantum field theories} are a different case for which an analogue of Haag's 
theorem does indeed hold, as we shall see in §\ref{sec:EuHTHM} in the superrenormalisable case. How
a nontrivial interacting theory is still attained despite this theorem's dictum, will also be shown there.

\section{Other versions of Haag's theorem}\label{sec:OthVersHT}                                      
The ensuing decades witnessed a number of pertinent results which we describe next. 
Especially the Streit-Emch theorem is worth being considered as it is closer to Haag's original 
formulation which we discussed in §\ref{DysMa}. 

\subsection{Work by Emch \& Streit}
Emch presents a very different variant of Haag's theorem in his monograph \cite{Em09} based on  
results proved by Streit in \cite{Strei68}. 

Emch writes about the previously published proofs that 
``...these proofs, however, rely rather heavily on the analytic properties of the Wightman functions,
which themselves reflect the locality and spectrum conditions, and tend to obscure the simple algebraic
and group-theoretical facts actually responsible for the results obtained.''(\cite{Em09}, p.247). 

Unlike the older versions, the Emch-Streit result focusses on the Weyl 
representations of the CCR generated by the field $\varphi(f)$ and its canonical momentum $\pi(g)$, smeared by test functions $f$ and $g$, ie
\begin{equation}
 U(f) = e^{i\varphi(f)} \ , \hs{1}  V(g) = e^{i\pi(g)}  
\end{equation}
which then satisfy the Weyl form of the CCR: $U(f)V(g) = e^{i(f,g)} V(g)U(f)$. 
This version of Haag's theorem is purportedly more general by assuming neither relativistic covariance 
nor causality (also known as locality, see §\ref{sec:WightAx}). 

What the authors assume instead is covariance with respect to 
a more general symmetry group that exhibits a property named '$\eta$-clustering'. 
This feature of the symmetry, defined via some averaging process, is essentially the clustering property 
known for spacelike translations in relativistic theories (cf.(\ref{clust}) in §\ref{sec:WightAx}). 
We have reason to believe, however, that for viable quantum field theories, the set of 
assumptions used in Emch's proof implies the very two conditions purportedly not needed, ie 
relativistic covariance and causality. 

Before we elaborate on this point, let us have a look at the assertion of the Emch-Streit theorem. 
The upshot there is, interestingly, very close to that originally stated by Haag which we discussed in §\ref{DysMa}, where we mentioned 
the polarisation of the vacuum: vacuum polarisation cannot occur if both Weyl representations of the CCR algebra are to be unitarily 
equivalent.

Let $V$ be the unitary transformation connecting the fields $\varphi_{1}(f)$ and $\varphi_{2}(f)$ and their canonical conjugates, then the outcome is 
\begin{equation}\label{eqHam}
 H_{2} = V H_{1} V^{-1}
\end{equation}
for the corresponding generators of time translations, ie the Hamiltonians. Let $H_{1}$ be the Hamiltonian
of the free field which exhibits no vacuum polarisation: $H_{1}\Psi_{01}=0$. Then the other vacuum is also
not polarised: 
\begin{equation}
H_{2}\Psi_{02} = H_{2}V\Psi_{01} = VH_{1}\Psi_{01} = 0. 
\end{equation}
Because the decomposition $H_{2}=H_{1}+H_{int}$ with some interaction part $H_{int}$ is not compatible with 
unitary equivalence (\ref{eqHam}) of both Hamiltonians, the authors conclude that the other theory
is also free.

We shall not present the Streit-Emch theorem in its general form here as even its provisions are rather 
technical. The interested reader is referred to \cite{Em09}. Instead, we quickly discuss a simpler 
version with a very elegant proof taken from \cite{Fred10} which nevertheless shows that vacuum 
polarisation cannot occur in Fock space.

\begin{Theorem}[No vacuum polarisation]
 Let $\Hi$ be the Fock space of a free field $\varphi_{0}$ with time translation generator $H_{0}$ 
 and $U(\mathbf{a})$ a representation of the translation subgroup with invariant state $\Omega \in \Hi$
 (the vacuum). Assume there exists a field $\varphi$ with well-defined sharp-time limits $\varphi(t,f)$ 
 and $\partial_{t}\varphi(t,f)=\dot{\varphi}(t,f)$ for all test functions $f \in \Sw(\R^{n})$ such that
 \begin{enumerate}
  \item[(i)] $\varphi(0,f)=\varphi_{0}(0,f)$ and $\dot{\varphi}(0,f)=\dot{\varphi}_{0}(0,f)$;
  \item[(ii)] $U(\mathbf{a})\varphi(t,f)U(\mathbf{a})^{\dagger}=\varphi(t,\tau_{a}f)$, where 
              $(\tau_{a}f)(\mathbf{x})=f(\mathbf{x}-\mathbf{a})$;
  \item[(iii)] there exists a self-adjoint operator $H$ on $\Hi$ with $[H,U(\mathbf{a})]=0$ for all 
  $\mathbf{a} \in \R^{n}$ and 
  \begin{equation}\label{Heis}
   \varphi(t,f) = e^{iHt} \varphi(0,f)e^{-iHt}.
  \end{equation}
Then there is a constant $\lambda \in \C$ so that $H=H_{0}+\lambda$ and thus $\varphi_{0}=\varphi$.
 \end{enumerate}
\end{Theorem}
\begin{proof}
Because the free field generates a dense subspace in $\Hi$, it suffices to show  $H=H_{0}+\lambda$ on 
a state of the form $\Psi = \varphi_{0}(0,f_{1})...\varphi_{0}(0,f_{n})\Omega$ for test 
functions $f_{1}, ..., f_{n}$.
The Heisenberg picture evolution (\ref{Heis}) imlies $\dot{\varphi}(t,f)=i[H,\varphi(t,f)]$
and the condition $[H,U]=0$ entails that the vacuum $\Omega$ is an eigenstate of $H$. Let $\lambda$ be the
corresponding eigenvalue, ie $(H-\lambda)\Omega =0$. Then, by repeated application of $H \varphi(0,f)= \varphi(0,f)H -i \dot{\varphi}(0,f)$, we get 
\begin{equation}
 \begin{split}
  (H-\lambda)& \varphi_{0}(0,f_{1})...\varphi_{0}(0,f_{n})\Omega = 
  (H-\lambda)\varphi(0,f_{1})...\varphi(0,f_{n})\Omega  \\
  &= -i\sum_{j=1}^{n} \varphi(0,f_{1})...\dot{\varphi}(0,f_{j})...\varphi(0,f_{n})\Omega 
  = - i\sum_{j=1}^{n} \varphi_{0}(0,f_{1})...\dot{\varphi}_{0}(0,f_{j})...\varphi_{0}(0,f_{n})\Omega \\
  &= H_{0}\varphi_{0}(0,f_{1})...\varphi_{0}(0,f_{n})\Omega .
  \end{split}
\end{equation}
The assertion then follows straightforwardly from (\ref{Heis}) and $H=H_0 + \lambda$, 
\begin{equation}
 \varphi(t,f) = e^{iHt} \varphi(0,f)e^{-iHt} = e^{iH_{0}t} \varphi(0,f)e^{-iH_{0}t} = e^{iH_{0}t} \varphi_{0}(0,f)e^{-iH_{0}t} =
 \varphi_{0}(t,f).
\end{equation}
\end{proof}
We come back to the point made above concerning covariance and causality: 
contrary to what Emch and Streit claim, we believe  
that both relativistic covariance and causality are implied in their assumptions.

Lurking behind clustering, ie the well-known cluster decomposition property of 
vacuum expectation values, however, is causality, namely that the fields commute at spacelike 
distances\footnote{This is discussed in §\ref{sec:WightAx}, equation (\ref{clust}).}. 
The Emch-Streit theorem can only be more general for a model
and not rely on relativistic covariance and causality, if there are symmetries other than translations 
with the same properties and only if, on top of that, lacking causality cannot obstruct the clustering feature.
In other words, because the clustering property is a consequence of both causality and translational 
covariance, we deem it highly questionable that there is any symmetry in QFT other than translations that gives rise to clusters without the 
aid of causality.

\subsection*{Vacuum polarisation in Galilean theories} 
As presented by Lévy-Leblond in \cite{LeBlo67}, there are Galilean quantum field theories 
which provide nonrelativistic examples of interacting field theories that do \emph{not} polarise the vacuum
and thereby defy Haag's argument we presented in §\ref{DysMa} regarding the polarisation of the vacuum. 
We briefly expound the author's argument. 

Let $H,\{ P_{j},J_{j},K_{j} \}$ be the generators of the Galilean group, ie the Galilean algebra consisting of the generators of time translations $H$, 
spatial translations $P_{j}$, rotations $J_{j}$ and Galilean boosts $K_{j}$. 
In contrast to the Poincaré algebra, the subset $\{ P_{j},J_{j},K_{j} \}$ is a Lie subalgebra which remains unaltered if the Hamiltonian is augmented 
by an interaction term. The crucial difference is displayed by the commutators  
\begin{equation}\label{Com}
 \mbox{Galilean case:} \hs{0.2} [K_{j},P_{l}]=i \delta_{jl} m  \ , \hs{1}
 \mbox{Poincaré case:} \hs{0.2} [K_{j},P_{l}]=i \delta_{jl} H,
\end{equation}
where $m>0$ is the Galilean particle's mass and $\{ K_{j} \}$ on the Poincaré side are the Lorentz boost generators in the Poincaré algebra. 
These commutators show clearly that altering the Poincaré algebra's Hamiltonian $H$ requires the other generators be modified accordingly. 
More precisely, the point now is this: both Galilei and Poincaré algebra share the commutator 
\begin{equation}\label{GalAl}
[K_{j},H]=i P_{j} .
\end{equation}
It is now possible in the Galilean case to find a perfectly physical interaction 
term $H_{I}$ without messing up this commutator, ie (\ref{GalAl}) is left in peace because $[K_{j},H_{I}]=0$
\cite{LeBlo67}. Even if one finds a Poincaré 
counterpart, it will interfere with the commutators in (\ref{Com}) which, as we said above, entails that
the other generators must inevitably be changed as well. Let now $H=H_{0}+H_{I}$ be a 
Galilean Hamiltonian with free and interacting part $H_{0}$ and $H_{I}$, respectively, such that 
$[K_{j},H_{0}]= [K_{j},H] = i P_{j}$. Let $\Psi_{0}$ be the vacuum with $H_{0}\Psi_{0}=0$. Then
\begin{equation}
[K_{j},H]\Psi_{0}=i P_{j}\Psi_{0}=0 \hs{1} \Rightarrow \hs{1} K_{j}H\Psi_{0}=H K_{j} \Psi_{0} =0,
\end{equation}
because the vacuum is Galilean invariant. One infers from $K_{j}H\Psi_{0}=0$ that either the vacuum is polarised,
$H\Psi_{0}=\Psi_{0}$ (up to a prefactor) or not, ie $H\Psi_{0}=0$. If the former holds, 
then $H':=(H-\id)$ will do the job and there is no vacuum polarisation. 

\subsection*{Streit-Emch theorem more general?} 
This result and the following thought calls the purported generality with respect to the symmetry group 
of the Emch-Streit theorem into question. One of the assumptions made by Emch and 
Streit is \emph{cyclicity of the vacuum}\footnote{This will be defined in 
§\ref{sec:WightAx} as part of the Wightman axioms}. 
As Fraser cogently explains in her thesis \cite{Fra06}, it is at least this property
which Galilean field theories can easily violate. The standard example is a quantised Schr\"odinger field
$\psi(t,\mathbf{x})$, defined as a solution of the second quantised Schr\"odinger equation
\begin{equation}
 i \frac{\partial}{\partial t} \psi(t,\mathbf{x}) = -\frac{1}{2m}\Delta \psi(t,\mathbf{x})
\end{equation}
which has only positive-energy solutions of the form
\begin{equation}
 \psi(t,\mathbf{x}) = \int d^{3}k \ e^{-i(\frac{\mathbf{k}^{2}}{2m} t -\mathbf{k}\cdot \mathbf{x})} c(\mathbf{k}) 
\end{equation}
and is not capable of creating a dense subspace, even when its 
adjoint is taken to join the game. Rather, they create subspaces
known as \emph{superselection sectors} for both particles with mass $m>0$ and $-m<0$ 
(see \cite{Fra06} and references there\footnote{Fraser's text is an excellent piece of work 
devoted to Haag's theorem from a philosophical point of view, see also §\ref{Fraser}.}).

So, in summary, the cyclicity of the vacuum can probably only be achieved by relativistic fields. Granted that on the face of it, the Emch-Streit 
theorem is more general, the question as to what theories other than relativistic ones satisfy its assumptions if Galilean field theories
fail to do so becomes even more pressing. Our conviction therefore is: none of physical relevance.    
   
\subsection{Algebraic version} Another version of Haag's theorem can be found in \cite{Wei11} which
we prefer to mention only briefly. The author proved this theorm in the context of \emph{algebraic 
quantum field theory} as introduced by Haag and Kastler in \cite{HaKa64} 
(see also \cite{Ha96, Bu00} and \cite{Em09}). 
For those readers acquainted with this somewhat idiosyncratic take
on quantum field theory, here is very roughly speaking the upshot, unavoidably couched in algebraic
language: the unitary equivalence class of a net
of local von Neumann algebras is completely determined on a spacelike hyperplane. This means that if the
von Neumann algebras of two field theories are related on a spacelike hypersurface through a unitary map, 
then they are unitarily equivalent and all vacuum expectation values agree.

\subsection{Noncommutative QFT} The last version of Haag's theorem we mention here has been found in the
context of \emph{noncommutative quantum field theory} and is presented in \cite{AMY12}. 
This peculiar theory purports to generalise QFT by letting the time coordinate fail to commute with the 
spatial coordinates. The result there says that if two theories are related by a unitary transformation and
the S-matrix is trivial in one theory, so it must be in the other. The interested reader wanting to 
embark on a 'noncommutative journey' is referred to \cite{AMY12} as a starting point.

\section{Superrenormalisable theories evade Haag's theorem}\label{sec:EuHTHM}                        
Of particular importance is an analogue of Haag's theorem in the Euclidean realm proven for a class of superrenormalisable theories, which we 
survey in this section. 

We will now see how the triviality dictum of Haag's theorem coexists peacefully with its very circumvention by (super)renormalisation  
and discuss a result by Schrader which facilitates the understanding of the meaning of Haag's
theorem for renormalisable QFTs profoundly. We know that Haag's theorem is 
independent of the dimension of spacetime and hence is also true for superrenormalisable QFTs. 
Because there is a well-known connection between Euclidean and relativistic quantum field 
theories \cite{OSchra73}, one has to expect there to be a Euclidean manifestation of Haag's theorem. 

And indeed, this is exactly what Schrader's result says: an analogue of Haag's theorem holds also in
the Euclidean realm. The good news is that superrenormalisable quantum theories are relatively 
well-understood courtesy of the work of constructive field theorists.

Although the protagonists did not necessarily think of it that way (they did not say), we shall now have a glimpse at their results and
see that Haag's triviality dictum has in fact been proven to be evaded in a big class of 
superrenormalisable field theories by, one may say, \emph{(super)renormalisation}!

\subsection{Euclidean realm} 
Schrader proved in \cite{Schra74} that Haag's theorem is also lurking in Euclidean 
field theories of the type $P(\varphi)_{2}$, where this symbol stands for the interaction term of the
Hamiltonian in the form of a polynomial $P$ bounded from below ('semi-bounded') with $P(0)=0$ ('normalised'). 

In this context, a Euclidean field $\varphi$ is a random variable with values in the set $\De'(\R^{2})$ of 
distributions $\omega \colon \De(\R^{2}) \rightarrow \C$ on the space of test functions of compact support. Let us
write them as
\begin{equation}
 \omega(f) = \int d^{2}x \ \omega(x)f(x) \hs{2} f \in \De(\R^{2}) .
\end{equation}
The expectation values are given by \emph{functional integrals}. For example, the expectation values 
of a free field are given by the functional integral
\begin{equation}
 \la \varphi(f) \varphi(h) \ra_{0} = \int_{\De'(\R^{2})} \omega(f) \omega(h) \ d\mu_{0}(\omega) ,
\end{equation}
where $f,h \in \De(\R^{2})$ are test functions and $\mu_{0}$ is the Gaussian measure with covariance 
\[
 (-\Delta + m^{2})^{-1} .
\]
Then, given a semi-bounded nontrivial real and normalised polynomial $P$,
the interacting part of the Euclidean action with coupling $\lambda >0$ is defined by
\begin{equation}\label{Vint}
 V_{l,\lambda}(\varphi):= \lambda \int_{-l/2}^{+l/2} \int_{-l/2}^{+l/2} d^{2}x : P(\varphi(x)): ,
\end{equation}
where $:...:$ stands for Wick ordering, the Euclidean analogue of normal ordering (see §\ref{Diver} or \cite{GliJaf81}) and $l$ is the box 
length which serves as a volume cutoff. Schrader builds on Newman's findings, namely that the family of measures given by
\begin{equation}\label{Renmeas}
 d\mu_{l,\lambda}(\omega) = \frac{e^{-V_{l,\lambda}(\omega)}}{\la e^{-V_{l,\lambda}(\varphi)} \ra_{0}} 
 \ d\mu_{0}(\omega)
\end{equation}
converges to a measure $\mu_{\infty,\lambda}$ for sufficiently small $\lambda > 0$ through the limit
\begin{equation}
\begin{split}
 \lim_{l'\rightarrow \infty} 
 \int_{\De'(\R^{2})} \omega(\chi_{l}f_{1})... \omega(\chi_{l}f_{n}) \  d\mu_{l',\lambda}(\omega) 
 = \int_{\De'(\R^{2})}  \omega(\chi_{l}f_{1})... \omega(\chi_{l}f_{n}) \ d\mu_{\infty,\lambda}(\omega)
\end{split}
\end{equation}
for all $l>0$, where $\chi_{l}$ is a test function with compact support in the box. 
The $f_{j}$'s are any test functions, ie the expectation values are well-defined for all test functions
with compact support. 
Interestingly, Schrader's result is now that $\mu_{\infty,\lambda}$ and $\mu_{0}$ have mutually disjoint 
support (for sufficiently small $\lambda >0$). Thus, although the two-point function of the cutoff theory
\begin{equation}\label{cuto}
 \la \varphi(f) \varphi(h) \ra_{l,\lambda} 
  = \int_{\De'(\R^{2})}  \omega(f) \omega(h) \ d\mu_{l,\lambda}(\omega) 
  =  \int_{\De'(\R^{2})}  \frac{e^{-V_{l,\lambda}(\omega)}}{\la e^{-V_{l,\lambda}(\varphi)} \ra_{0} } 
  \ \omega(f) \omega(h) \ d\mu_{0}(\omega) 
\end{equation}
has a measure of the same support as the free measure, this situation changes dramatically in the limit
to the full (ostensibly physical) theory when $l \rightarrow \infty$. While one may expect to be able to 
approximate (\ref{cuto}) by employing perturbation theory with respect to $\lambda$ to the last integral, 
the result of mutually disjoint measure support suggests that this method will lead to anything but an 
approximation. However, since the limit measure $\mu_{\infty,\lambda}$ exists for small enough $\lambda > 0$, 
the model was further explored with the following highly interesting results. 

\begin{enumerate}
 \item[1.] The Schwinger functions, ie the Euclidean Green's functions 
 \begin{equation}
  S_{\lambda}(x_{1}, ..., x_{n}):=\la \varphi(x_{1}) ... \varphi(x_{n}) \ra_{\lambda} 
  = \int_{\De'(\R^{2})} \omega(x_{1}) ... \omega(x_{n}) \ d\mu_{\infty,\lambda}(\omega)
 \end{equation}
 exist in the sense of distributions. Let $f \in \Sw(\R^{2n})$ be a test function. 
 Then the function $\lambda \mapsto S_{\lambda}(f)$ is smooth in $\lambda$ in an interval 
 $[0,\lambda_{0})$ for $\lambda_{0} >0$ sufficiently small and has a Borel-summable asymptotic 
 Taylor series \cite{Di73}. 
 \item[2.] Their Minkowski limit $\R^{2} \ni (x^{0},x) \rightarrow (it,x)$ yield distributions that 
 satisfy the axioms due to Wightman \cite{GliJaSp74} and give rise to a Wightman theory with nontrivial 
 S-matrix \cite{OS76}, obtained through a nonperturbative LSZ reduction 
 expansion\footnote{No concrete result for comparison with any perturbative result was computed.
 This is because one has to know the $n$-point functions to make practical use of the LSZ formula!}. 
\end{enumerate}

What we see here is that although a Euclidean variant of Haag's theorem is valid, it does not preclude
the existence of nontrivial interactions. Heuristically, it is easy to see that the measure 
$\mu_{\infty,\lambda}$ may actually be seen as a 'superrenormalised' measure: if we write
\begin{equation}
\la e^{-V_{l,\lambda}(\varphi)} \ra_{0} = e^{- E_{l,\lambda} } 
\end{equation}
with ground state energy $E_{l,\lambda} = \la V_{l,\lambda}(\varphi) \ra_{0}$, then (\ref{Renmeas}) 
becomes 
\begin{equation}\label{renm}
 d\mu_{l,\lambda}(\omega) = e^{-[V_{l,\lambda}(\omega)-E_{l,\lambda}]} d\mu_{0}(\omega).
\end{equation}
This is exactly the kind of 'renormalisation' that Glimm and Jaffe used in their $(\varphi^{4})_{2}$ model
which they treated in the operator approach \cite{Jaff69, GliJaf68, GliJaf70}. 

\subsection{Evasion of Haag's theorem}\label{EvH} 
If we combine Schrader's Euclidean result with the above points, we can draw a clear conclusion.
Even though this result is about the class $P(\varphi)_{2}$, ie a big class of superrenormalisable quantum
field models, it is of particular importance as it helps us understand what Haag's theorem may mean for 
renormalisable field theories. 

The fact that the two measures $\mu_{0}$ and 
$\mu_{\infty,\lambda}= \lim_{l \rightarrow \infty} \mu_{l,\lambda}$ 
have mutually disjoint support tells us that there exists no Radon-Nikodym density relating these
two. This means that the Radon-Nikodym density in (\ref{renm}) ceases to make sense in the limit.
In the operator approach, this probably corresponds to unitary inequivalence between 
the free and the interacting theory. 
Yet still: (super)renormalisation leads to a sensible result, ie another measure which describes a 
nontrivial theory.

This is what we opine about the meaning of Haag's theorem for renormalisable theories and try to make 
plausible in this work: 
once renormalised, these theories are nontrivial and unitary inequivalent to the very free theories 
employed to construct them. In other words, it is precisely renormalisation what allows us to stay clear
of Haag's theorem. 
                                  
\section{What to do about Haag's theorem: reactions}\label{sec:whatodo}                              
We survey the reactions that Haag's theorem stirred among a minority of the physics community despite the fact that it stands in direct contradiction 
to the canonical formalism of perturbation theory which presupposes and relies on the existence of Dyson's matrix and the interaction picture.  

Some say the price for ignoring this and other triviality theorems (to be discussed in later sections) are the UV divergences that have beset 
the canonical formalism from the very start \cite{Wi67}.
While this is certainly true, one has to appreciate the role of \emph{renormalisation} which has become
an integral part of the formalism. This subtraction scheme does not only successfully fix these problems 
by rendering individual integrals finite and well-defined but changes the theory systematically.  

One has to say that both Haag's and van Hove's papers \cite{Ha55,vHo52} came out 
relatively late, in the early 1950s, that is, \emph{after} the renormalisation problem of 
quantum electrodynamics (QED) had already been settled by Feynman, Tomonaga, 
Schwinger and Dyson (see \cite{Feyn49,Dys49a,Dys49b} and references there). 
In hindsight it was perhaps fortunate that these pioneers did not get sidetracked by 
representation issues. 

However, in Dyson's review of Haag's 1955 paper, he acknowledges that 
\begin{quotation}
''The meaning of these results is to make even clearer than
before the fact that the Hilbert space of ordinary quantum mechanics is too narrow a framework in which
to give consistent definition to the operations of quantum field theory ... attempts to build a rigorous
mathematical basis for field theory ... always stop short of any nontrivial examples. The question, what 
kind of enlarged framework would make consistent definitions possible, is the basic unsolved problem of
the subject.''
\end{quotation}
Aware of van Hove's work, he furthermore writes ''The so-called 'Haag's theorem' ... is 
essentially an old theorem of L.van Hove``\footnote{ Dyson's review of Haag's paper is 
available at www.ams.org/mathscinet, keywords: author ``Haag'', year 1955. }.

Clearly, trying to explain and tackle the puzzle posed by Haag's theorem certainly falls within the 
remit of \emph{mathematical} physics which is probably why it was largely ignored by practising physicists 
or even dismissed. K\"allen is quoted as saying that the result ``... is really of a very trivial
nature and it does not mean that the eigenvalues of a Hamiltonian never exist or anything that 
fundamental.''\cite{Lu05} 
Such attitudes inspired Wightman to write in a proceedings paper that ``... there is a 
widespread opinion that the phenomena associated with Haag's theorem are somewhat pathological and 
irrelevant for real physics. I make one more attempt to explain why this is not the case.''\cite{Wi67}

\subsection{Textbooks} 
Despite all this, nowadays' standard physics textbooks (eg \cite{PeSch95, Wein95}) 
do not mention Haag's theorem in any way. To find a discussion of it, one has to turn to 
textbooks from before 1970 which treat it rather differently: while the books by Roman \cite{Ro69} 
and Barton \cite{Bar63} devote several pages to it, Bjorken and Drell \cite{BjoDre65} say in a footnote on p.175 that 
although Haag's theorem excludes the existence of Dyson's matrix, they will (understandably) assume that 
it does exist notwithstanding. 

Sterman takes a similiar stance in \cite{Ster93} which he relegated to the appendix where he reviews the interaction picture. Her writes on p.508
\begin{quotation}
''Haag's theorem states that the unitary transformation ... is not strictly consistent with Poincarè invariance ... this fascinating point, ... however, has not been shown to 
affect practical results, ...''
\end{quotation}
This is probably the most down-to-earth standpoint possible. Ticciati's view in \cite{Ti99}, p.84 is somewhat more subtle,
\begin{quotation}
 ''[the] first few terms yield wonderfully accurate predictions. It appears then that the interaction picture provides a sound approach 
 to perturbation theory but may have no non-perturbative validity.''
\end{quotation}
This gives the unfortunate impression that the author does not take the results of renormalised perturbation theory seriously. 
However, Ticciati suggests that one may drop the assumption of unitary equivalence and, also on the plus side, his formulation of Haag's theorem corresponds
to our version in that he does not make use of the conjugate momentum fields (see §\ref{sec:ProofHTHM}).

\subsection*{Barton}
Barton discusses the technicalities of the Wightman distributions' analyticity properties and surveys
the provisions of Haag's theorem at considerable length to fathom out which may be relinquished. The one he would like to retain is the condition 
that both fields obey the CCR (see \cite{Bar63}, p.158). He includes this property in his exposition of Haag's theorem as part of the
conditions for the generalised theorem (\ref{nleq4pf}) (ibid., p.153), exactly 
where it has no business to hang around: not surprising therefore that he neither uses the CCR nor gives 
reasons why they are needed. In fact, the reason he does not employ 
them is that they are dispensible (see \cite{StreatWi00}, p.100) and, as we will see in §\ref{sec:PowTHM}, probably not fulfilled by a theory of true 
interactions. To put this in perspective, however, Barton's book \cite{Bar63} clearly shows that at the time of writing, 
its author favoured the Yang-Feldman approach to the S-matrix in the Heisenberg picture to be introduced in the context of Lopuszanski's work 
in §\ref{LopuCon} where the CCR are an integral part (Yang and Feldman's seminal paper is \cite{YaFe50}). 

\subsection*{Roman} 
While Roman's account of Haag's theorem and his observation that the CCR may be discarded are 
entirely correct (\cite{Ro69}, pp.330 and pp.392), we do not at all agree with his assertions 
as to what stand to take on this issue \cite{Ro69}: the author contends that 
\emph{renormalisation does not help}. 

We have to mention that the author seemed to have a weird understanding of the Stone-von Neumann theorem 
(Theorem \ref{SvN} of this thesis) and deemed it applicable to QFT (\cite{Ro69}, p.330).
According to his understanding, it is merely the assertion that unitary equivalence of two field algebras
implies that both algebras are irreducible, if one is. While this assertion is hardly questionable, we know that the
real Stone-von Neumann theorem does not hold for an infinite number of degrees of freedom 
and hence not for QFT (see §§\ref{sec:IneqRep},\ref{sec:strangeRep}). 
However, to make his case about renormalisation, he (wrongly) invokes his understanding of the theorem to 
argue as follows. 

Assuming that by the Stone-von Neumann theorem the CCR only permit unitarily equivalent representations, the wavefunction (=field-strength) 
renormalisation $Z$ in the CCR  
\begin{equation}\label{RenCCR}
 [\varphi(\mathbf{x}),\pi(\mathbf{y})] = iZ^{-1} \delta^{(3)}(\mathbf{x}-\mathbf{y}) 
\end{equation}
of the renormalised field $\varphi(\mathbf{x})$ does not change the fact that this 'interacting' 
renormalised field is free. Apart from the fact that the wavefunction renormalisation constant $Z$ is a dubious fellow which
forces in our opinion the rhs of (\ref{RenCCR}) to vanish, he ignores that there is much more intricate
formalism between the free and the interacting field than the canonical formalism pretends. 

This is one central assertion of this work which we shall discuss in §§\ref{sec:Ren} and \ref{sec:RCircH}: 
although the canonical apparatus speaks of a unitary map between the free interaction picture and 
the interacting Heisenberg picture field, it disproves what it purports by its own actions. 
Encountering divergences, there is a subsequent backpedalling and fiddling in of infinite factors, 
only to claim again unitary equivalence where there cannot be any (we will next-to prove it within the mindset of the canonical formalism 
in §\ref{sec:RCircH}).

However, on the plus side, Roman discusses the provisions of Haag's theorem and says that he is inclined 
to give up unitary equivalence and maybe also the CCR (ibidem, p.393). Apart from his peculiar idea of what
the Stone-von Neumann theorem is about, he makes valuable remarks on the CCR and its relation to the 
property of locality (ibidem, p.328): if locality holds, ie 
\begin{equation}
 [\varphi(x),\varphi(y)]=0 \hs{2} (x-y)^{2}<0,
\end{equation}
then the CCR must be singular because 
\begin{equation}
 \frac{1}{\varepsilon}[\varphi(t,\mathbf{x}),\varphi(t+\varepsilon,\mathbf{y})-\varphi(t,\mathbf{y})]=0 ,
\end{equation}
for all small enough $\varepsilon \neq 0$. He also mentions that we cannot be sure that this singular 
commutator is a c-number. Despite his interest in the CCR, Roman was not aware of Powers' results, which 
were published in 1967, two years prior to the publication of his book. 
He was probably spared this news\footnote{His discussion of the CCR would
have been much more extensive; historical remark: a case of disconnected research communities although Roman, as far as
one can tell from his textbook was clearly mathematically inclined.}.

\subsection{Monographs}\label{monog} 
We have found two monographs covering Haag's theorem. One of them, \cite{Stro13} by Strocchi, has already been and will be again cited numerous times
in this work, especially in §\ref{sec:AxG}. Let us first consider the other one.

\subsection*{Duncan} The notable recent extensive and mathematically-oriented monograph \cite{Dunc21} by Duncan contains a section on Haag's theorem 
entitled ''How to stop worrying about Haag's theorem''.
The author shows that a simple mass shift leads to the van Hove phenomenon, ie the vanishing overlap 
of the two vacua. 

This can actually be tracked down to the fact that free field theories with different masses 
are unitarily inequivalent, which he seemed not to be aware of when he wrote his book 
(see Theorem \ref{freeH} or Theorem X.46 in \cite{ReSi75}, p.233).
Most interestingly, he subsequently treats the mass shift as a 'mass perturbation', ie the difference
between the two masses as an interaction term and applies standard perturbation theory in the 
interaction picture to it. The upshot there is that the propagator of the mass-shifted field comes out correctly.   

This brings us right to the very contention underlying this work on which we shall elaborate in 
§\ref{sec:RCircH}: \emph{unitary inequivalence does not mean that canonical (renormalised) perturbation theory 
yields nonsensical and unphysical results}. In fact, we use Duncan's example to conclude that 
\emph{renormalisation allows us to evade Haag's triviality dictum}!
 
Duncan points out that Haag's theorem does not exclude the existence of 
the $S$-matrix maintaining that both Haag-Ruelle and LSZ scattering theories ''lead to a perfectly 
well-defined, and unitary, $S$-matrix'' on the basis of the axiomatic framework (\cite{Dunc21}, p.364). 
The author then explains his attitude towards this issue (\cite{Dunc21}, pp.369-370): 
\begin{itemize}
 \item one should in a \textsc{first step} introduce a spatial volume and a UV cutoff to work with a finite number of degrees of freedom and hence a 
       well-defined interaction picture,
\end{itemize} 
while accepting the price of sacrificing Poincarè invariance at this stage. 
Then, having kept Haag's theorem at bay so far by retaining the provisions of the Stone-von Neumann theorem, 
\begin{itemize}
 \item one redefines in a \textsc{second step} both mass and couplings by renormalisation and removes all cutoffs which finally restores Poincarè invariance.
\end{itemize} 
The upshot is, at least according to Duncan, that we have enough reason to stop worrying about Haag's theorem because we have by virtue of this procedure 
circumvented it. This is essentially the view that nowadays we believe practising physicists generally subscribe to or would subscribe to if we had told 
them this story. 

Yet in the ambition to retain the interaction picture as long as possible along the way, however, the author leaves the 
impression that he wishes to arrive at a unitarily equivalent representation of the CCR. As we have seen 
in the previous section and will extensively discuss in §\ref{sec:Ren}, 
this is certainly not where the 'circumventing' procedure leads. If the aim of a 'circumvention scheme' is
unitary equivalence to free fields, then it has to go wrong as Haag's theorem cannot be circumvented in this sense: it is a mathematical theorem in the 
truest sense of the word; it brings with it the 
'hardness of the logical must'\footnote{Wittgenstein}. 

\subsection*{Strocchi}
Strocchi puts the triviality dictum expressed through Haag's theorem down to the facts 
we will discuss and explain in §\ref{sec:Fock}: the fact that any theory within the unitary 
equivalence class of the free particle Fock space $\Hi_{0}$ must have a unitarily equivalent generator 
of time translations $H_{0}$ is irreconcilable with a splitting into a free and another nontrivial piece.     

Strocchi sees the interaction picture as instrumental in computing nontrivial results in perturbation
theory but says unitary equivalence to the Heisenberg field ceases to make sense in the no-cutoff limit,
even after renormalisation (\cite{Stro13}, pp.52). We agree with this fully and have to stress again
that Haag's theorem cannot be circumvented in this way: 
Haag's theorem is a general statement and does not take into account any feature of Dyson's matrix, 
ie the field intertwiner, other than its unitarity. 
All other provisions have a fundamentally different status, as we shall discuss in §\ref{sec:ProofHTHM} where we scrutinise the proof of Haag's theorem. 

\subsection{Papers}\label{PaRea} 
Guenin and Segal suggest that Haag's theorem can be bypassed 
if the time evolution of the interaction picture is implemented in the form of \emph{locally} unitary
automorphisms and not, as usual, by globally unitary maps \cite{Gue66, Seg67}. 

The former author introduces a modified interaction picture in which the trivial part of the Hamiltonian acts on the states and the 
nontrivial one on the observables. The Dyson series he then obtains for a specific class of Euclidian 
invariant Hamiltonians does, however, to say the least, not only look peculiar but also leads to a 
convergent perturbation series in one space dimension\footnote{Recall that $(\varphi^{4})_{2}$ has an
asymptotic perturbation series and, as is well-known, so does $(\varphi^{4})_{0}$! \cite{GliJaf81}.}.

\subsection{Fraser's thesis}\label{Fraser}
The most extensive review of Haag's theorem is Doreen Fraser's doctoral thesis \cite{Fra06}. She belongs to a small community 
of philosophers of physics whose work revolves around the interpretation of quantum theory in general and, in her case, of quantum field theory in 
particular. Her thesis is roughly composed of 3 parts. 

In the first, she expounds Haag's theorem and sketches its proof. 
In the second part, she discusses 
possible responses, eg 
\begin{itemize}
 \item introduction of a volume cutoff at the price of sacrificing translational invariance,
 \item renormalisation, discussed using the example of $(\varphi^{4})_{2}$, characterised by the cutoff limit
 \begin{equation}
  H_{\mbox{\scriptsize ren}} = \lim_{\Lambda \rightarrow \infty} \{ H_{\Lambda} - E^{0}_{\Lambda} \},
 \end{equation}
       which comes at the price of an ill-defined counterterm for the ground state energy $E^{0}_{\infty}$, 
 \item dropping the assumption of unitary equivalence and using other approaches like Haag-Ruelle scattering theory and constructive approaches.
\end{itemize}
Our view on these points is that none of them is satisfactory: 
while the constructive method is fine unfortunately only for lower-dimensional Minkowski spacetime with more or less trivial rotations, 
Haag-Ruelle theory is impractical: no predictions can be made and the connection to renormalised QFT is still unclear (see §\ref{AsyS}).    

The third part is devoted to ontological questions. Her conclusion about the ontological status of particlelike entities is based on the tenuous mode 
of existence of particles as described by interacting non-Fock QFTs: 
\begin{quotation}
 '' ... since in the real world there are always interactions, QFT does not furnish grounds for regarding particlelike entities as fundamental constituents 
 of reality''(\cite{Fra06}, p.137).
\end{quotation}
This is, from a physics point of view, too radical since \emph{elastic interactions} are also captured by QFT, albeit in a nonrelativistic limit. 
However, here is an interesting aspect: Fraser states that Haag's theorem is not associated with UV but infinite-volume divergences! 

In a way, this strikes us as plausible: Haag's theorem is \textsc{firstly} independent of the dimension of spacetime\footnote{It should be a QFT, ie spacetime should be at least two-dimensional! 
We want Lorentz boosts!}, 
which is why superrenormalisable theories are also affected by Haag's theorem and \textsc{secondly} Wightman's arguments described in §\ref{Wiag} against the 
interaction picture are clearly based on an infinite-volume divergence. 

Her statement that ''...Haag's theorem undercuts global unitary equivalence, it is compatible with local unitarity equivalence.''(\cite{EaFra06}, p.323)
supports this view, which is inspired by a modified version of Theorem \ref{freeH} ('Haag's theorem for free fields', §\ref{massdest})
with a contrary outcome (\cite{ReSi75}, p.329): if 
$\varphi_{0}$ and $\varphi$ are two free fields with masses $m_{0}$ and $m$, respectively, and $B \subset \Mi$ a bounded region, then there exists a unitary 
map $V_{B} \colon \Hi_{0} \rightarrow \Hi$ between the two corresponding Hilbert spaces such that $\varphi(f)=V_{B}\varphi_{0}(f)V_{B}^{-1}$ for all 
$f \in \De(B)$. This fits in nicely with Guenin and Segal's studies. 

But when proper interactions enter the game in physical Minkowski space $\Mi=\R^{4}$, one cannot separate the UV divergences away and declare them as 
having nothing to do with Haag's theorem. We prefer a more modest position regarding this issue: the ill-definedness of the interacting picture in any 
spacetime leads to several types of divergences. 
Haag's theorem speaks of no divergences. It has to be merely read as saying: either both theories are free or they must be unitarily inequivalent.

\part{Axiomatic quantum field theory \& Haag's theorem}\label{part2}

\section{The interaction picture in Fock space}\label{sec:Fock}                                      

This section reviews the Fock space for an infinite number of degrees of freedom and has a critique
of the interaction picture. Altogether, the arguments presented there provide compelling reasons why the 
interaction picture cannot exist in the setting of a Fock space, at least for a scalar theory with a mass gap. 
They shed some light on the connection between the Fock representation with its characteristic number operator and go back to the 
publications \cite{DeDoRu66,DeDo67,Chai68}. 
Strocchi has cast these results into a canonical form, ie purged of operator-algebraic argot, in his 
monograph \cite{Stro13}, which we shall follow in this section. 

\subsection{Fock space representations} 
We start with the usual creation and annihilation operators of a canonical free field, satisfying the CCR
\begin{equation}\label{CCRk}
 [a(\mathbf{k}),a(\mathbf{k}')] = 0 = [a^{\dagger}(\mathbf{k}),a^{\dagger}(\mathbf{k}')] \ , \hs{1}
 [a(\mathbf{k}),a^{\dagger}(\mathbf{k}')] = (2\pi)^{3} \delta^{(3)}(\mathbf{k}-\mathbf{k}').
\end{equation}
and smooth them out with test functions $f \in \Sw(\R^{3})$,
\begin{equation}\label{annsmoo}
 a(f):= \int \frac{d^{3}k}{(2\pi)^{3}} f^{*}(\mathbf{k}) a(\mathbf{k}) , \hs{1}
 a^{\dagger}(f):= \int \frac{d^{3}k}{(2\pi)^{3}} f(\mathbf{k}) a^{\dagger}(\mathbf{k}).
\end{equation}
Choosing an orthonormal Schwartz basis $f_{j} \in \Sw(\R^{3})$ with respect to the inner product
\begin{equation}
(f_{i},f_{j}):= \int \frac{d^{3}k}{(2\pi)^{3}} f_{i}^{*}(\mathbf{k})f_{j}(\mathbf{k}) = \delta_{ij},
\end{equation}
the CCR (\ref{CCRk}) take the form 
\begin{equation}\label{CCRf}
 [a(f_{i}),a(f_{j})] = 0 = [a^{\dagger}(f_{i}),a^{\dagger}(f_{j})] \ , \hs{1}
 [a(f_{i}),a^{\dagger}(f_{j})] = \delta_{ij}.
\end{equation}
Then we obtain what we shall in the following call a \emph{Heisenberg algebra}: if we set $a_{j}:=a(f_{j})$
and $a_{j}^{*}:=a^{\dagger}(f_{j})$, ie 
 \begin{equation}\label{FockCCR}
[a_{j},a_{j}] = 0 = [a^{*}_{j},a^{*}_{j}] \ , \hs{1} [a_{j},a^{*}_{l}]=\delta_{jl},
 \end{equation}
then the Heisenberg algebra is given by the polynomial algebra 
$\Al_{H}:= \la a_{j}, a_{j}^{*} : j \in \N \ra_{\C}$, generated by the creation and annihilation operators. 

Given a vector $\Psi_{0}$ in a Hilbert space $\Hi$ and a representation $\varrho$ of the Heisenberg algebra,
such that $\varrho(a_{j})\Psi_{0}=0$ for all $j \in \N$, called \emph{Fock representation}, it is
clear that the \emph{number operator} $N_{\varrho}:=\sum_{j \geq 1}\varrho(a^{*}_{j})\varrho(a_{j})$ exists 
on the domain
\begin{equation}
 \D_{0} := \varrho(\Al_{H}) \Psi_{0}.
\end{equation}
If the closure of $\D_{0}$ yields $\Hi$, ie if $\D_{0}$ is dense in $\Hi$, we say that the representation of the 
Heisenberg algebra is \emph{cyclic} with respect to the vacuum $\Psi_{0}$. If $[C,\varrho(\Al_{H})]=0$ 
implies $C=c 1$ with $c \in \C$ for an operator $C$ on $\Hi$, the representation $\varrho$ is 
called \emph{irreducible}. Fock representations are always cyclic (by definition) and 
irreducible \cite{Stro13}. We shall for convenience drop the symbol $\varrho$ for the representation 
whenever there is no potential for confusion. 

\subsection*{Vacuum \& no-particle state} 
We shall now survey a small collection of assertions which urge us to draw the above mentioned negative conclusion
about the interaction picture in Fock space. The first assertion is (tacitly) well-known among physicists.

\begin{Proposition}[Number operator and vacuum \cite{Stro13}]\label{NuVa}
Let $\{a_{j}, a_{j}^{*} : j \in \N \}$ be an irreducible representation of the Heisenberg algebra
with a dense domain $\D_{0}$ in a Hilbert space $\Hi$. Then the following two conditions are equivalent.
\begin{enumerate}
 \item[1.] The total number operator $N:=\sum_{j \geq 1}a_{j}^{*}a_{j}$ has a nonnegative spectrum
           $\sigma(N)$ and exists in the sense that the strong limit
 \begin{equation}
  s-\lim_{n \rightarrow \infty} e^{i \alpha \sum_{j=1}^{n}a_{j}^{*}a_{j}} = T(\alpha)
 \end{equation}
exists and defines a strongly continuous one-parameter group of unitary operators\footnote{Stone's theorem then guarantees that $N$ exists 
as the generator of $T(\alpha)$.}.  

 \item[2.] There exists a cyclic vector $\Psi_{0} \in \Hi$ such $a_{j}\Psi_{0}=0$ for all $j \in \N$.
\end{enumerate}
\end{Proposition}
\begin{proof}
 Let the first condition be given. First note that the CCR imply  
 $T(\alpha)a_{j}=e^{-i \alpha} a_{j}T(\alpha)$ and 
 $T(\alpha)a^{\dagger}_{j}=e^{i \alpha} a^{\dagger}_{j}T(\alpha)$. This entails $[T(2\pi),\Al_{H}]=0$ 
and thus $T(2\pi)=e^{i\theta}1$ on account of unitarity and irreducibility. Using the spectral representation of $T(\alpha)$, 
ie $T(\alpha)=\int_{\sigma(N)} e^{i\alpha \lambda} d\mathsf{E}(\lambda)$ we consider
\begin{equation}
 0 = (T(2\pi)-e^{i\theta})^{\dagger}(T(2\pi)-e^{i\theta})
 = \int_{\sigma(N)} |e^{i (2\pi \lambda-\theta)}-1|^{2} d\mathsf{E}(\lambda)
\end{equation}
which means that the spectrum $\sigma(N)$ is discrete: $2 \pi \lambda - \theta \in 2 \pi \mathbb{Z}$. Pick 
$\lambda \in \sigma(N)$ with $\lambda >0$ and let $\Psi_{\lambda}$ be its eigenstate. Then 
\begin{equation}
 0 < \lambda ||\Psi_{\lambda}||^{2} = \la \Psi_{\lambda} | N \Psi_{\lambda} \ra 
 = \sum_{j \geq 1} \la \Psi_{\lambda} | a_{j}^{*} a_{j} \Psi_{\lambda} \ra 
 = \sum_{j \geq 1} ||a_{j}\Psi_{\lambda}||^{2}
\end{equation}
entails that there is at least one $j$ such that $a_{j}\Psi_{\lambda} \neq 0$. The CCR imply
\begin{equation}
 T(\alpha)a_{j}\Psi_{\lambda}= e^{-i\alpha} a_{j} T(\alpha)\Psi_{\lambda} 
 = e^{i\alpha(\lambda-1)} a_{j}\Psi_{\lambda}
\end{equation}
and therefore $a_{j}\Psi_{\lambda} = c \Psi_{\lambda-1}$ with some $c \in \C$. Because the spectrum is
bounded from below and nonnegative, there must be a state $\Psi_{0}$ such that $a_{j}\Psi_{0}=0$. 
If the second condition is fulfilled, then it is clear that $N$ exists on the sense subspace 
$\D_{0}=\Al_{H}\Psi_{0}$ and so does the limit of the exponentiation. 
\end{proof}

\subsection*{Hamiltonian chooses Fock representation} 
The next result is very interesting and pertinent to the interaction picture question.

\begin{Proposition}[Fock Hamiltonian \cite{Stro13}]\label{HfreeFock}
Assume that for an irreducible representation $\{a_{j}, a_{j}^{*} : j \in \N \}$ of the Heisenberg 
algebra with dense domain $\D_{0}$, there exists an operator 
\begin{equation}\label{freeH}
 H_{0} = \sum_{j \geq 1} \omega_{j} a^{*}_{j}a_{j} \ , \hs{1} \forall j: \ 0 < m \leq \omega_{j} 
\end{equation}
as generator of a strongly continuous one-parameter unitary group
given by the strong limit
\begin{equation}
  s-\lim_{n \rightarrow \infty} e^{i \alpha \sum_{j=1}^{n} \omega_{j} a_{j}^{*}a_{j}} = e^{i\alpha H_{0}}
\end{equation}
such that the domain $\D_{0}$ is stable under the action of this group. Then the Fock representation is 
selected by $H_{0}$. 
\end{Proposition}
\begin{proof}
 $\sum_{j=1}^{n}\omega_{j} a^{*}_{j}a_{j} \geq m \sum_{j=1}^{n}a^{*}_{j}a_{j}$ tells us that the existence
 of $H_{0}$ implies the existence of the number operator $N$ which by Proposition \ref{NuVa} selects the 
 representation. 
\end{proof}
This result relies on the existence of the mass gap, as the proof suggests. In fact, Proposition \ref{HfreeFock} 
can be formulated for a free scalar field with (formal) Hamiltonian 
\begin{equation}\label{Hamfree}
 H_{0} = \int \frac{d^{3}p}{(2\pi)^{3}} \ \omega(\mathbf{p}) a^{\dagger}(\mathbf{p})a(\mathbf{p})
\end{equation}
with relativistic energy $\omega(\mathbf{p})=\sqrt{\mathbf{p}^{2}+m^{2}}$ of a particle with 
momentum $\mathbf{p} \in \R^{3}$ and rest mass $m>0$. 
However, the case $m=0$ is different: Proposition \ref{HfreeFock} does not hold
and there is an infinite variety of inequivalent representations of the CCR, even with nonnegative 
energy \cite{BoHaSch63}. 

We shall therefore stick to the mass gap case. The Hamiltonian (\ref{Hamfree}) can be defined as the 
operator acting on the one-particle state $a^{\dagger}(f)\Psi_{0}$ according to
\begin{equation}
 H_{0}a^{\dagger}(f)\Psi_{0} = [H_{0},a^{\dagger}(f)]\Psi_{0} = a^{\dagger}(\omega f)\Psi_{0},
\end{equation}
in which $(\omega f)(\mathbf{k})=\omega(\mathbf{k})f(\mathbf{k})$ is a perfect Schwartz function, cf.(\ref{annsmoo}). This 
definition coheres with the formal Hamiltonian (\ref{Hamfree}) and the CCR (\ref{CCRk}). Using
the condition $H_{0}\Psi_{0}=0$, one can then easily compute 
\begin{equation}
 H_{0}a^{\dagger}(f_{1})...a^{\dagger}(f_{n})\Psi_{0} = 
 [H_{0},a^{\dagger}(f_{1})...a^{\dagger}(f_{n})]\Psi_{0}
\end{equation}
by using the commutator property $[A,BC]=B[A,C]+[A,B]C$. 
The matrix element of $H_{0}$ with respect to the one-particle state is 
\begin{equation}\label{HfreeMe}
 \la a^{\dagger}(f)\Psi_{0}| H_{0}a^{\dagger}(f)\Psi_{0} \ra = 
 \la a^{\dagger}(f)\Psi_{0}| a^{\dagger}(\omega f)\Psi_{0} \ra = (f,\omega f).
\end{equation}
One can now choose $f$ such that the rhs of (\ref{HfreeMe}) cannot be distinguished by measurement 
from the correct relativistic energy of a free particle with some fixed momentum. 

In the same way, one introduces the momentum operator, 
\begin{equation}\label{spatra}
 P^{j}= \int \frac{d^{3}p}{(2\pi)^{3}} \ p^{j}  a^{\dagger}(\mathbf{p})a(\mathbf{p})
\end{equation}
which is given through $[P^{j},a^{\dagger}(f)]=a^{\dagger}(p^{j}f)$ and completes the translation 
subalgebra of the Poincaré group. The creators and annihilators then transform according to
\begin{equation}
 e^{ib\cdot P}a(f)e^{-ib\cdot P} = a(e^{ib\cdot p}f) \ , \hs{1} 
 e^{ib\cdot P}a^{\dagger}(f)e^{-ib\cdot P} = a^{\dagger}(e^{ib\cdot p}f) .
\end{equation}
For for pure time translations, this takes the form
\begin{equation}
 e^{iH_{0}t} a(f) e^{-i H_{0}t} = a(e^{i \omega t}f) \ , \hs{1} 
 e^{iH_{0}t}a^{\dagger}(f)e^{-i H_{0}t} = a^{\dagger}(e^{i\omega t}f) .
\end{equation}
The Lorentz group is implemented similiarly. Finally, the canonical free field $\varphi_{0}(t,f)$ and its conjugate momentum field are then given by 
\begin{equation}\label{Fockf}
\varphi_{0}(t,f) = \frac{1}{\sqrt{2}} [a(\omega^{-\frac{1}{2}} e^{i\omega t} f) 
+ a^{\dagger}(\omega^{-\frac{1}{2}} e^{i\omega t} f) ], \hs{0.2} 
\pi_{0}(t,f) = \frac{i}{\sqrt{2}} [a^{\dagger}(\omega^{\frac{1}{2}} e^{i\omega t} f) 
- a(\omega^{\frac{1}{2}} e^{i\omega t} f) ]
\end{equation}
and satisfy the CCR 
\begin{equation}\label{FockfCCR}
 [\varphi_{0}(t,f),\varphi_{0}(t,g)]=0=[\pi_{0}(t,f),\pi_{0}(t,g)] \ , \hs{1} 
 [\varphi_{0}(t,f),\pi_{0}(t,g)]=i(f,g).
\end{equation}
This is what we call the Fock representation of a free field. One distinguishes the translation-invariant
state called vacuum from the no-particle state. Because of $P^{\mu}\Psi_{0}=0=N\Psi_{0}$, the intuition 
is that both coincide. Now consider 

\begin{Proposition}[Fock representation \cite{Stro13}]\label{FockP}
All Fock representations are unitarily equivalent. The vacuum is unique, ie the only 
translation-invariant state, and coincides with the no-particle state. 
\end{Proposition}
\begin{proof}
The isomorphism $V:\varrho(\Al_{H})\Psi_{0} \rightarrow \varrho'(\Al_{H})\Psi_{0}'$ is densely defined
and preserves the scalar product. Let $\Psi_{0}$ be the no-particle state, ie $N\Psi_{0}=0$ and let 
$\Psi \neq \Psi_{0}$ also have this property. Then $a_{j}\Psi=0$ and hence $\la \Psi | A \Psi_{0} \ra =0$
if $A  \in \Al_{H}^{*} := \la a_{j}^{*} : j \in \N \ra_{\C}$. But because 
$\Al_{H}\Psi_{0} = \Al_{H}^{*}\Psi_{0}$ is dense, we find $\Psi=0$, ie the state $\Psi_{0}$ is the only state annihilated by the $a_{j}$'s.  
Now note that $N$ commutes with the space translation operator $U(\mathbf{b})=\exp(-i\mathbf{b} \cdot \mathbf{P})$ by definition of $P^{j}$ in 
(\ref{spatra}). Let $\Phi \neq \Psi_{0}$ be translation invariant. Then, 
\begin{equation}
 U(\mathbf{b})N\Phi = NU(\mathbf{b})\Phi = N \Phi = \sum_{n \geq 0} N \Phi_{n}  
\end{equation}
where $\Phi = \sum_{n \geq 0}\Phi_{n}$ is the decomposition into the $n$-particle subspace components. This
implies $U(\mathbf{b})\Phi_{n} = \Phi_{n}$, ie a contradiction because the $n$-particle state is not translation invariant.  
\end{proof}

\subsection{Interaction picture}\label{Wiag} 
Now, in the light of the above results, the implementation of the interaction picture poses the following 
problem. The intertwining operator for the interaction picture\footnote{See §\ref{sec:GeMaLo} for 
a review of this part of the canonical formalism.},
\begin{equation}
 V(t)=e^{iH_{0}t}e^{-iHt}
\end{equation}
demands the separate existence of a free Hamiltonian $H_{0}$ which by Propositions \ref{HfreeFock} and \ref{FockP} selects a Fock
representation with a unique Poincaré-invariant vacuum. The interacting Hamiltonian $H$ is yet another, an
additional time translation generator that needs to be implemented. The time evolution operator for 
interaction picture states, given by
\begin{equation}
 \mathsf{U}(t,s)= V(t)V(s)^{\dagger} = e^{iH_{0}t}e^{-iH(t-s)}e^{-iH_{0}s}
\end{equation}
clearly shows that $H$ and $H_{0}$ are supposed to act in the same Hilbert space, ie the Fock space. 
The problem is now that even in relatively simple models studied so far, this split into two well-defined 
self-adjoint operators $H=H_{0}+H_{int}$ has not been possible: either the sum is well-defined or only one
part of it, not both. The requirement of some form of renormalisation then always leads to non-Fock 
representations, which by Proposition \ref{FockP} cannot be unitarily equivalent to the Fock representation
of a free field (see \cite{Stro13}, pp. 40). 

Another compelling argument put forward by Wightman in \cite{Wi67} against the existence of the interaction 
picture uses translation invariance to show that the Schr\"odinger picture integral 
\begin{equation}
 H_{int} = \int d^{3}x \ \Ha_{int}(\mathbf{x}),  
\end{equation}
ie the interacting part of the Hamiltonian and its interaction picture representation 
\begin{equation}\label{HI}
 H_{I}(t) := e^{iH_{0}t}H_{int}e^{-iH_{0}t} = \int d^{3}x \ e^{iH_{0}t}\Ha_{int}(\mathbf{x})e^{-iH_{0}t} 
 =: \int d^{3}x \ \Ha_{I}(t,\mathbf{x}) 
\end{equation}
do not make sense. And here is Wightman's argument, a head-on attack on the interaction picture with straightforward reasoning: we first compute 
\begin{equation}\label{normHint}
\begin{split}
 ||H_{int} \Psi_{0}||^{2} &= \int d^{3}x \int d^{3}y \
       \la \Psi_{0} | \Ha_{int}(\mathbf{x}) \Ha_{int}(\mathbf{y}) \Psi_{0} \ra  \\
       &= \int d^{3}x \int d^{3}y \
       \la \Psi_{0} | \Ha_{int}(0) \Ha_{int}(\mathbf{y}) \Psi_{0} \ra .
\end{split}
\end{equation}
This integral over $\mathbf{x} \in \R^{3}$ diverges unless $\Ha_{int}(\mathbf{y}) \Psi_{0}=0$,
at least according to Wightman in \cite{Wi67}; although one does not have to follow him here, the integrand
must vanish in any case. And because the evolution operators in (\ref{HI}) can be 
inserted without changing the norm, we find that $||H_{I}(t)\Psi_{0}||=0$ is the only sensible and 
acceptable outcome. 

But because the Hamiltonian $\Ha_{I}(t,\mathbf{x})$ is made up of free interaction picture fields which 
always have terms with the right combination of creators and annihilators for the vacuum expectation
value not to vanish, we can see that the interaction picture Hamiltonian does 
not exist in the way the canonical formalism desires it to. However, contrary to Wightman's conclusion, 
one can interprete (\ref{normHint}) in such a way as to say 
\begin{equation}
\la \Psi_{0} | \Ha_{I}(t,0) \Ha_{I}(t,\mathbf{y}) \Psi_{0} \ra 
=\la \Psi_{0} | \Ha_{int}(0) \Ha_{int}(\mathbf{y}) \Psi_{0} \ra = 0
\end{equation}
just means that the states $\Ha_{I}(t,0) \Psi_{0}$ and $\Ha_{I}(t,\mathbf{y}) \Psi_{0}$ have no overlap.
But still, $||H_{I}(t)\Psi_{0}||=0$ and hence $H_{I}(t)\Psi_{0}=0$ is inevitable.

\section{Canonical (anti)commutation relations and no-interaction theorems}\label{sec:PowTHM}        
Next, we present the no-interaction theorems of Powers and Baumann. The central outcome here
is that field theories conforming with the CCR/CAR must be necessarily free if 
the dimension of spacetime exceeds a certain threshold ($d\geq 3$ for fermions and $d\geq 5$ or $d\geq 4$ 
for bosons). Although free fields clearly satisfy these relations, this calls into question their meaning 
in a general QFT. 

To make the case against the CCR/CAR, some authors bring in the field-strength (or wave-function) renormalsation constant.
Since we deem this issue worthy of discussion, we have included some observations about this truely dubitable object in §\ref{sec:WaveRC}. 
We argue that it is not at all understood nonperturbatively and only obstructs insight into the connection between asymptotic scattering theory and 
renormalised perturbation theory. 

It is important to note that the proof of Haag's theorem does not require any of the fields 
to obey the canonical (anti)commutation relations (CCR/CAR) \emph{explicitly} which, for the spatially 
smeared scalar field operators 
\begin{equation}\label{ScSmear}
 \phi(t,f) = \int d^{n}x \ f(\mathbf{x}) \phi(t,\mathbf{x}) \ , \hs{1} 
 \pi(t,g) = \int d^{n}x \ g(\mathbf{x}) \pi(t,\mathbf{x}) \hs{1} f,g \in \Sw(\R^{n})
\end{equation}
take the form
\begin{equation}\label{CCRSmear}
[\phi(t,f), \phi(t,g)] = 0 = [\pi(t,f),\pi(t,g)] \ , \hs{1} 
[\phi(t,f), \pi(t,g)] = i (f,g),
\end{equation}
where $n \geq 1$ is the dimension of space\footnote{The smeared fields in (\ref{Fockf}) and (\ref{ScSmear})
differ: $f$ in (\ref{ScSmear}) is the Fourier transform of $f$ in (\ref{Fockf}), but never mind.}.
This condition is, however, \emph{implied} in a trivial way for the following reason: if (\ref{ScSmear}) 
are free fields and the unitary map $V$ also transforms the conjugate momentum field, then they 
of course satify the CCR (\ref{CCRSmear}). Applying the unitary transformation that 
connects both field theories to (\ref{CCRSmear}) shows that the other fields inevitably obey the CCR, too. 
Hence both theories are unitarily equivalent representations of the CCR. 

This uncorks the question whether the CCR (\ref{CCRSmear}) are somehow related to the triviality result 
entailed by Haag's theorem. Can truely interacting fields be \emph{any} representation of the these 
commutation relations? We shall see now that here, in contrast to Haag's theorem, the dimension of 
spacetime is decisive.

\subsection{Anticommutation relations and triviality}
Under some regularity conditions, Powers answered this question for fermion fields in \cite{Pow67}. 
For a Dirac field $\psi(t,f)$ and its canonical conjugate field 
$\psi(t,f)^{\dagger}=\psi^{\dagger}(t,f^{*})$ the canonical anticommutation relations (CAR) take the form  
\begin{equation}\label{CARSmear}
\{\psi(t,f), \psi(t,g)^{\dagger} \} = 0 = \{\psi(t,f)^{\dagger},\psi(t,g)^{\dagger} \} \ , \hs{1} 
\{\psi(t,f), \psi(t,g)^{\dagger} \} = i (f,g).
\end{equation}
Powers' result is now that these relations imply triviality in space dimension $n\geq 2$. 

\begin{Theorem}[Powers' theorem]
 Let $\psi(t,f)$ be a local relativistic Fermi field in the sense of Wightman's framework\footnote{We introduce this framework in §\ref{sec:WightAx}.} 
 in $d=n+1 \geq 3$ spacetime dimensions fulfilling the CAR (\ref{CARSmear}) and acting together with its 
 adjoint $\psi(t,f)^{\dagger}$ in a Hilbert space $\Hi$ with vacuum state $\Omega_{0}$. 
 
 Assume that they form an irreducible set of operators at one fixed instant and that there is a unitary 
 transformation $U(t)$ such that 
 \begin{equation}
  \psi(t,f) = U(t)\psi(0,f)U(t)^{\dagger}
 \end{equation}
for all times $t\in \R$ and that the limits  
\begin{equation}
\begin{split}
 \lim_{t \rightarrow 0} \frac{1}{t} [ \psi(t,f) - \psi(0,f)]\Omega_{0} 
 &= \partial_{t}\psi(0,f)\Omega_{0} \ , \\
 \lim_{t \rightarrow 0} \frac{1}{t} [ \psi(t,f) - \psi(0,f)]\psi(0,f)\Omega_{0} 
 &=\partial_{t}\psi(0,f) \psi(0,f)\Omega_{0}
\end{split} 
\end{equation}
and the corresponding ones for the adjoint exist in the norm for all test functions $f \in \Sw(\R^{n})$. 
Then $\psi(t,f)$ is a free field in the sense that it satisfies a linear differential equation which is 
first order in time.
\end{Theorem}
\begin{proof}
See \cite{Pow67}. 
\end{proof}

A few comments are in order. Powers did \emph{not} prove that the fields satisfy the Dirac equation. Instead he found that the conditions imposed on the 
Fermi field are so restrictive that 
there exist operators $T_{1}$ and $T_{2}$, linear and antilinear on $\Sw(\R^{n})$, respectively, such that 
\begin{equation}\label{FerEq}
 \partial_{t}\psi(t,f) = \psi(t,T_{1}f) + \psi(t,T_{2}f)^{\dagger}.  
\end{equation}
These operators may be realised in such a way that $T_{1}$ is the spatial part of the Dirac operator 
(stripped of the zeroth $\gamma$-matrix) and $T_{2}=0$. They do not have to take this specific form, though.
This is the precise sense in which the fields are free. 

Surely, if an interacting Dirac field obeys a differential equation in whatever sense,
it should not be of the form (\ref{FerEq}) because this implies in particular that there are no other  
fields to interact with. The assumption that these fields form an irreducible set of 
operators is, up to some mathematical subleties, equivalent to their capability of generating a dense 
subspace\footnote{See §\ref{sec:WightAx} on this issue.}. In other words, Powers' theorem is an 
assertion about a quantum field theory with fermions only. 
The theories that spring to mind are of the Fermi theory type with 'four-fermion' interactions like 
\begin{equation}
 (\overline{\psi}\gamma_{\mu}\psi)(\overline{\psi}\gamma^{\mu}\psi)  \hs{1}
 \mbox{('Thirring model')}
\end{equation}
which are all together non-renormalisable (in dimension $d \geq 3$) and hence unphysical, albeit of some value as effective 
field theories. Consequently, if from a physical point of view we require renormalisability,
fermions cannot directly interact with each other and should thus be free. Therefore, Powers' theorem makes perfect sense. 

\subsection{Commutation relations and triviality} 
Baumann investigated the case of bosonic fields but could not find the exact analogue of Powers'
result \cite{Bau87}.
Yet he proved that in space dimension $n>3$ only free theories can satisfy the CCR (\ref{CCRSmear}) 
and suggests that the theories $(\phi^{4})_{4}$ and $(\phi^{6})_{3}$ may fulfill these relations, whereas 
for $(\phi^{4})_{2}$, he found 'no restrictions', in agreement with Glimm and Jaffe's 
results. 

In fact, the only interacting theories so far found to satisfy the CCR are 
superrenormalisable field theories like $P(\phi)_{2}$, the sine-Gordon model $\cos(\phi)_{2}$ and 
the exponential interaction $\exp(\phi)_{2}$ \cite{GliJaf81}. 

Baumann's provisions fill a rather long list, to be looked up in \cite{Bau87} by the interested reader. To present them here would coerce us to explain 
earlier results by Herbst and Fr\"ohlich, including some argot (see the references in \cite{Bau87}). 
We shall not embark on this here and content ourselves with the main assertion and a sketch of the proof which we have relegated to the appendix. 

\begin{Theorem}[Baumann]\label{Baumann}
Let $n \geq 4$ be the space dimension and $\varphi(t,\cdot)$ a scalar field with 
conjugate momentum field $\pi(t,\cdot)=\partial_{t}\varphi(t,\cdot)$ such that the CCR (\ref{CCRSmear}) are 
obeyed and assume furthermore that $\dot{\pi}(t,\cdot):=\partial_{t}\pi(t,\cdot)$ exists. 
Then, if $\varphi(t,\cdot)$ has a 
vanishing vacuum expectation value and the provisions listed in the introduction of \cite{Bau87} are 
satified, one has
\begin{equation}
 \dot{\pi}(t,f) - \varphi(t,\Delta f) + m^{2} \varphi(t,f) = 0
\end{equation}
for all $f \in \De(\R^{n})$ and a parameter $m^{2}>0$. 
\end{Theorem}
\begin{proof}
See Appendix §\ref{sec:AppBau}.
\end{proof}
Notice that $\pi(t,\cdot)=\partial_{t}\varphi(t,\cdot)$ is not correct if the Lagrangian has an interaction 
term with a first time derivative, as we have in the case of a renormalised field theory incurred by the counterterms. This line of argument is 
heuristic, but we have no cause to believe $\pi=\dot{\varphi}$ for renormalised scalar field theories either.  

When Baumann scrutinised the case in which both 
coexisting bosons and fermions satisfy the CCR and CAR, respectively, he again found that for $n>3$ space 
dimensions any theory with these relations must necessarily be free while for $n=3$ it was impossible to
say \cite{Bau88}. We remind the reader that no renormalisable and non-renormalisable models have been 
constructed so far (see §\ref{sec:Intro}): it was merely Baumann's working assumption that 
these interacting theories exist and that their interaction terms have no derivative coupling. 

Unfortunately, space dimension $n=3$ remained defiant. Baumann mentions that the unpublished proof for 
$n\geq 3$ by Sinha, then ostensibly a PhD student of Emch's \cite{SinEm69}, had weaker
assumptions that somehow did not appeal to him. Emch presents Sinha's version 
without proof in \cite{Em09}. From what we can tell by comparing both results, Sinha's 
provisions are weaker as he uses the Weyl form of the CCR instead of the CCR themselves. 
We shall merely quote\footnote{The only source containing the proof seems to be Sinha's PhD thesis, only existing in print at the library of the 
University of Florida, whose staff did not reply to the our email. We did not insist.} the Weyl CCR case from \cite{Em09}.

\begin{Theorem}[Sinha] 
Let $n\geq 3$ be the dimension of space and $\phi(t,f)$ a sharp-time local 
scalar Wightman field with canonical conjugate $\pi(t,f)$ that generate a representation of the Weyl CCR 
at $t=0$, ie
\begin{equation}
 U(f) = e^{i\varphi(0,f)}, \ V(g) = e^{i\pi(0,g)} , \hs{1}  U(f)V(g) = e^{i(f,g)} V(g)U(f)
\end{equation}
where $U(f)U(g)=U(f+g)$ and $V(f)V(g)=V(f+g)$ for all $f,g \in \Sw(\R^{n})$ such that the families
$\lambda \mapsto U(\lambda f)$ and $\lambda \mapsto V(\lambda f)$ are weakly continuous. 
Assume further that $\partial_{t}\pi(t,f)$ exists and that the Weyl unitaries $U(f), V(g)$
are irreducible with common dense domain $\mathfrak{D}$ which is stable under the algebra of their 
generators. Then there exists a linear operator $T \colon \Sw(\R^{n}) \rightarrow L^{2}(\R^{n})$ and a 
distribution $c \in \Sw'(\R^{n})$ such that
\begin{equation}
 \partial_{t}^{2}\phi(t,f)=\phi(t,Tf)+c(f)
\end{equation}
and $\phi(t,\cdot), \pi(t,\cdot)$ fulfill the CCR (\ref{CCRSmear}).  
\end{Theorem}

These no-interaction theorems by Powers, Baumann and Sinha suggest that renormalisable and fully fledged 
interacting field theories will most likely neither satisfy the CCR (\ref{CCRSmear}) nor the CAR 
(\ref{CARSmear}). In other words, for a quantum field theory, there is no analogous representation issue 
as in quantum mechanics: we had better not seek a unitarily equivalent representation of the Fock 
representation since that will in all likelyhood be a free theory. 

For theories of only one type of field, ie fermions or scalar bosons, these results can be viewed 
as a variant of Haag's theorem in the sense that Dyson's matrix cannot exist: 
if a quantum field and its canonical momentum field are unitarily equivalent to those of a free theory 
at one fixed instant $t=t_{0}$, then it satisfies the canonical (anti)commutation rules at that time.
By the foregoing theorems, the field can only be free in spacetime dimensions $d \geq 5$ for 
bosons and $d \geq 3$ for fermions. If one assumes the Weyl form of the CCR, boson fields are free for 
spacetime dimensions $d\geq 4$. 

The interesting issue here is the spacetime dependence: we know that superrenormalisable theories
of the type $P(\varphi)_{2}$ conform with the CCR \cite{GliJaf81}, whereas (probably) renormalisable 
and non-renormalisable ones do not. Notice that these results do not mean that interacting theories cannot
exist in higher-dimensional spacetimes. Powers' and Baumann's results merely inform us that interactions 
are incompatible with the CCR/CAR there. 

\subsection{CCR/CAR and the Heisenberg uncertainty principle} In §\ref{sec:Focks} we have already 
discussed where the (anti)commutation relations came from: the Heisenberg 
uncertainty principle in quantum mechanics. Although free fields enjoy these relations by construction, we do not need 
them for interacting theories.

One reason is philosophical in nature: there is no position operator in QFT and the relation of the 
CCR/CAR to the Heisenberg uncertainty principle is obscure, at least in our mind and to the best of 
our knowledge. 

The other is practical:
any particles measured in scattering experiments are detected \emph{after} the scattering event when they are deemed free. 
The measuring apparatus cannot be placed within the interaction 
vertex\footnote{When the wire chamber detects a \emph{Townsend avalanche} kicked off by, say, ionised argon, everything is over already: 
whatever the interaction vertex's size (far below atomic scale), what comes out of it can only be thought of as free.}. 

Besides, concrete computations leading to numbers that can be measured in experiments 
are always carried out with free fields. It is through these fields that the constant $\hbar$ enters 
the theory. Unfortunately, none of the authors, ie Baumann, Powers, Strocchi and Wightman, who proposed that we abandon the CCR/CAR for a general 
interacting QFT, touched upon this important issue. This also goes for the next author whose pertinent results we shall discuss in brief.

\subsection{Lopuszanski's contribution}\label{LopuCon}
We finally mention Lopuszanski's results because of their relevance to the above no-interaction theorems. 
The aim of his work was to extensively characterise 
free scalar fields in order to figure out what properties interacting fields can by exclusion not have.

We briefly review some of his results published in \cite{Lo61}. We start with the Yang-Feldman representation of a 
massive scalar interacting Heisenberg field
\begin{equation}\label{YFE}
 \phi(t,\mathbf{x})= \phi_{\mbox{\tiny in}}(t,\mathbf{x}) - \int d^{4}y \ \Delta_{\mbox{\tiny ret}}(x-y)j(y),
\end{equation}
where $\Delta_{\mbox{\tiny ret}}(x-y)$ is the retarded Green's function of the Klein-Gordon operator 
and $j(y)$ the interaction term from the equation of motion, possibly containing other fields. 
The field $\phi_{\mbox{\tiny in}}(t,\mathbf{x})$ describes free incoming and asymptotic bosons. Because
the retarded Green's function vanishes in the limit $t=x^{0} \rightarrow \pm \infty$, the Heisenberg field 
$\phi(t,\mathbf{x})$ converges to the asymptotic field by the way it is defined\footnote{The Yang-Feldman equation
(\ref{YFE}) says that the outgoing field is identical to the incoming one.}. 

Lopuszanski's results interest us here because they make plausible that another path to the construction 
of the S-matrix which circumvents the interaction picture and hence makes no use of Dyson's matrix, is 
equally haunted by triviality if the provisions are too strong: 
the representation of the S-matrix in terms of Heisenberg fields due to Yang and Feldman \cite{YaFe50}. 

The first assumption is that the interaction current takes the form
\begin{equation}\label{phi4int}
 j(y)=\frac{g}{3!} \phi(y)^{3} .
\end{equation}
This is already problematic because powers of fields are ill-defined and must be Wick-ordered as we will discuss in §\ref{Diver} 
(see also \cite{StreatWi00}, p.168). 

The assumption that the Heisenberg field can be representated as a Fourier mode expansion
\begin{equation}\label{YangFeldExp}
 \phi(t,\mathbf{x}) = \int \frac{d^{3}q}{(2\pi)^{3}} \frac{1}{\sqrt{2E_{q}}} \ [e^{-i q \cdot x} a(t,q)
 + e^{i q \cdot x} a^{\dagger}(t,q) ], 
\end{equation}
with $E_{q} := \sqrt{\mathbf{q}^{2}+m^{2}}$ is also a bit strong for an interacting field. These mode operators are assumed 
to satisfy\footnote{We adopt Lopuszanski's notation $|0\ra$ for the vacuum.}
\begin{equation}\label{YangFeldCCR}
 \la 0 | [a(t,q),a^{\dagger}(t,q')]|0 \ra = Z^{-1} \delta^{(3)}(\mathbf{q}-\mathbf{q}')
\end{equation}
where all other commutators vanish and $Z^{-1}$ is (by Lopuszanski's assumption) the \emph{finite} inverse wavefunction renormalisation.
This amounts to demanding that the Heisenberg field obey the CCR, albeit including a peculiar factor. 
In view of (\ref{YangFeldExp}), the only way this theory has a chance to differ from a free one lies in the time dependence of
the mode operators.  

The idea that $\phi$ converges to the free field $\phi_{in}$ as we go back to the remote past means 
that these operators converge in some sense to those of the incoming field, ie
\begin{equation}
 a(t,q) \  \rightarrow \ a_{\mbox{\tiny in}}(q)  \ , \hs{1} 
 a^{\dagger}(t,q) \ \rightarrow \ a^{\dagger}_{\mbox{\tiny in}}(q)  
 \hs{1} \mbox{as} \hs{0.2} t \rightarrow - \infty .
\end{equation}
This entails $Z=1$ because in the limit, the lhs of (\ref{YangFeldCCR}) goes over to the commutator of the
incoming mode operators which require $Z=1$, as this object is necessarily time-independent. 
Because, so he argues, the K\"allen-Lehmann spectral representation must satisfy
\begin{equation}\label{ZKL}
Z^{-1} = \int d\mu^{2} \ \rho(\mu^{2}) = \int d\mu^{2} \ [\ \delta(\mu^{2}-m^{2}) + \sigma(\mu^{2}) \ ]
= 1 + \int d\mu^{2} \ \sigma(\mu^{2}) 
\end{equation}
the conclusion is $\sigma(\mu^{2})=0$, that is, $\phi(t,\mathbf{x})$ is trivial. This is Lopuszanski's 
first no-interaction result which tells us that the CCR (\ref{YangFeldCCR}) had better not be fulfilled.  

We feel strongly obliged to critique (\ref{ZKL}), but will defer a discussion of this issue and the wave renormalisation constant to §\ref{sec:WaveRC}). 
Let us now have a look at Lopuszanski's main theorem.

\begin{Claim}[Lopuszanski]
Assume $\phi(x)|0\ra = \phi_{\mbox{\tiny in}}(x)|0\ra$. Then follows that $\phi(x)=\phi_{\mbox{\tiny in}}(x)$ and the theory is trivial. 
\end{Claim}
\begin{proof}
The proof is elementary: first of all, the assumption tells us $j(y)|0\ra =0$. This means the interaction term is 
incapable of polarising the vacuum. This seems questionable for (\ref{phi4int}), to say the least.
However, because of causality, we have
\begin{equation}
 [\phi(t,\mathbf{x}), j(t,\mathbf{y})]=0
\end{equation}
and thus $0=\phi(t,\mathbf{x})j(t,\mathbf{y})|0\ra = j(t,\mathbf{y})\phi(t,\mathbf{x})|0\ra
= j(t,\mathbf{y})\phi_{\mbox{\tiny in}}(t,\mathbf{x})|0\ra$. This is nothing but 
\begin{equation}
0= \int \frac{d^{3}q}{(2\pi)^{3}} \frac{1}{\sqrt{2E_{q}}} 
\ e^{-i x \cdot q } j(t,\mathbf{y})a^{\dagger}_{\mbox{\tiny in}}(q) |0 \ra 
\end{equation}
and therefore $j(t,\mathbf{y})|q \ra = 0$. By induction one gets $j(y)|q_{1},...,q_{n} \ra =0$. Finally,
on account of the assumption that these states span a dense subspace, the result is $j(y)=0$.  
\end{proof}

Lopuszanski's assumption $\phi(x)|0\ra = \phi_{\mbox{\tiny in}}(x)|0\ra$ was probably inspired by the idea that applying
the interacting field only once should create \emph{one} single particle. Because there is no other particle, the particle created by the free field 
cannot be distinguished from a single particle minted by the interacting field.  

But this assumption already sneaks in that there is no vacuum polarisation:  
once the particle has been created, it will be there on its own and there will be nothing for it to interact with. No cloud around it, hence no interaction.  

We are thus informed that this assumption is fallacious for an interacting field.
Lopuszanski himself concludes that a reasonable interacting field should not be of the form 
(\ref{YangFeldExp}) and should also not satisfy the CCR (\ref{YangFeldCCR}) \cite{Lo61}. 

\subsection{Conclusion about the CCR/CAR} In 1964 and hence prior to the publication of the above no-interaction theorems, Streater and
Wightman wrote in their book \cite{StreatWi00} that they do not exclude the CCR for interacting 
fields in general, but contend that the hints they have from examples leave them in no doubt that singular 
behaviour is to be expected for sharp-time fields, even after being smoothed out in space. 

Consequently, it may in such cases be difficult to give the CCR in (\ref{CCRSmear}) a meaning. ``Thus, one is reluctant to accept canonical
commutation relations as an indispensible requirement on a field theory.'' (\cite{StreatWi00}, p.101).

Although Wightman encouraged Baumann to work on the CCR/CAR question (see acknowledgements in \cite{Bau87, Bau88}), he did obviously not deem the 
results important enough to update these remarks in the latest edition \cite{StreatWi00} of the year 2000. 

So we conclude that although some superrenormalisable theories have been found to conform with the CCR,
renormalisable and non-renormalisable theories cannot be expected to have this feature. 
However, if an interacting field theory fulfills what is known as the \emph{asymptotic condition}, then
it may obey these relations at least \emph{asymptotically}. 

The idea that interacting fields obey some form of the CCR or CAR is generally not discussed and strictly
speaking not claimed to be true in physics. In fact, when asked, many practicing physicists would have
to first think of where the idea came from to produce the Heisenberg uncertainty principle as an answer. 

Yet the CCR/CAR represent a property constitutive for the quantisation of free fields. Therefore, 
it is somehow tacitly taken for granted for 'all' fields when a classical theory is 'quantised'. 
But as a property, it is, in actual fact, only used for free fields to compute the Feynman propagator which is in turn
needed for perturbation theory in the interaction picture. 

We shall see in §\ref{sec:CCRQuest} that the CCR question cannot be answered within the canonical formalism of perturbation theory in 
its current form.  

On the grounds that these relations are simply of no practical relevance for interacting fields and 
give rise to philosophical questions only, we are of the opinion that for the time being, one can 
easily get on without them and hence need not be disturbed by Powers' and Baumann's results. 

Our impression is that these triviality results are generally unknown to practising physicists\footnote{Otherwise there
should at least be a short remark about it in every lecture/textbook of QFT when the CCR/CAR are introduced.
Since free fields satisfy them, no restrictions follow and the lecture can blithely be carried on.} and given that 
\emph{those in the know about them are only mathematical physicists of a special creed}, 
namely axiomatic and algebraic quantum field theorists, this situation is likely to stay that way. 

Let us welcome Powers' and Baumann's theorems as another piece of information about interacting field 
theories that make us realise how little we know about them and how much our imagination is influenced by
being only familiar with free fields. But because a large class of superrenormalisable
theories conform with the CCR, the question as to why (or why not) still lingers on.

\section{Wave-function renormalisation constant}\label{sec:WaveRC}                                   
As alluded to above, other mathematical physicists have also expressed their doubts about the CCR 
for interacting fields.
In his monograph \cite{Stro13}, Strocchi comes to the same devastating conclusion about the CCR for 
interacting fields, he writes ``...canonical quantization cannot be used as a rigorous method for 
quantizing relativistic  interacting fields...''(ibidem, p.51). 

\subsection{CCR/CAR and wave-function renormalisation}
He mentions the no-interaction results of Powers and Baumann and expresses his opinion that the singular behaviour
of interacting sharp-time fields is the root of all evil. His main argument is to say that the CCR for the renormalised and hence interacting field
\begin{equation}\label{intCCR}
[\varphi_{r}(t,\mathbf{x}),\dot{\varphi}_{r}(t,\mathbf{x})] = i Z^{-1} \delta^{(3)}(\mathbf{x}-\mathbf{y})
\end{equation}
make no sense because - according to Strocchi - the wave renormalisation $Z$ vanishes and lets the rhs diverge. Before we comment on this, let us quickly 
review where this form of the CCR comes from and that, in fact, the canonical Lagrangian formalism protects itself from attacks like this one. 
Both (\ref{YangFeldCCR}) and (\ref{intCCR}) are the result of assuming
\begin{equation}\label{intCCR1}
[\varphi(t,\mathbf{x}),\dot{\varphi}(t,\mathbf{x})]=[\varphi(t,\mathbf{x}),\pi(t,\mathbf{x})] = i \delta^{(3)}(\mathbf{x}-\mathbf{y})
\end{equation}
for the bare fields\footnote{Notice that the bare field is not to be confused with a free field.} $\varphi(x)$ and $\pi(x)$, where the canonical 
momentum field is found by differentiating the bare Lagrangian with respect to $\dot{\varphi}=\partial_{t}\varphi$ and simply yields 
$\pi(x)=\dot{\varphi}(x)$. One then takes the renormalised field, $\varphi_{r}:=Z^{-1/2} \varphi$ and its first time derivative 
$\dot{\varphi}_{r}=Z^{-1/2} \dot{\varphi}$ to find for the renormalised field
\begin{equation}
[\varphi_{r}(t,\mathbf{x}),\dot{\varphi}_{r}(t,\mathbf{x})] = Z^{-1} [\varphi(t,\mathbf{x}),\dot{\varphi}(t,\mathbf{x})] 
= i Z^{-1} \delta^{(3)}(\mathbf{x}-\mathbf{y}).
\end{equation}
But this is \emph{not} the 'proper' CCR, ie what comes out if we strictly follow the rules of the canonical formalism: the 'correct' renormalised conjugate momentum is given by
\begin{equation}\label{renmom}
\pi_{r} = \frac{\partial \La}{\partial \dot{\varphi}_{r}} = Z \dot{\varphi}_{r} = Z Z^{-1/2} \dot{\varphi} = Z^{1/2} \pi.
\end{equation}
where one has to make use of the Lagrangian of the renormalised field, given by
\begin{equation}\label{Lren'}
\La = \frac{1}{2}Z(\partial \varphi_{r})^{2} - \frac{1}{2}m_{r}^{2}Z_{m}\varphi_{r}^{2} - \frac{g_{r}}{4!}Z_{g}\varphi_{r}^{4} ,
\end{equation}
which we will come back to in §\ref{sec:Ren}. Then, with this result, we find by (\ref{renmom})
\begin{equation}
[\varphi_{r}(t,\mathbf{x}),\pi_{r}(t,\mathbf{x})] =  [Z^{-1/2}\varphi(t,\mathbf{x}),Z^{1/2}\pi(t,\mathbf{x})] 
= i \delta^{(3)}(\mathbf{x}-\mathbf{y}),
\end{equation}
which is the 'canonically correct' CCR of the renormalised field, completely free of ailments, seemingly. 

\subsection{Lopuszanski's argument}
We resume the discussion on the CCR question in §\ref{LopuCon}, where we presented Lopuszanski's reasoning. 
To derive (\ref{ZKL}), we first consider the commutator function of the free field $\phi_{0}(t,\mathbf{x})$,
\begin{equation}\label{KaLCf'}
D(t-s,\mathbf{x}-\mathbf{y};m^{2}):=[\phi_{0}(t,\mathbf{x}),\phi_{0}(s,\mathbf{y})] 
= \int \frac{d^{4}q}{(2\pi)^{3}} \ \delta_{+}(q^{2}-m^{2}) [e^{-iq \cdot (x-y)} - e^{+iq \cdot (x-y)} ],
\end{equation}
then take the K\"allen-Lehmann representation  
\begin{equation}\label{KaLC'}
\la \Omega | [\phi(t,\mathbf{x}),\phi(s,\mathbf{y})]\Omega \ra 
= \int d\mu^{2} \ \rho(\mu^{2}) D(t-s,\mathbf{x}-\mathbf{y};\mu^{2}) 
\end{equation}
of the interacting (renormalised) field's commutator and differentiate it with respect to $s$ to get
\begin{equation}\label{KLCCR}
 \la \Omega | [\phi(t,\mathbf{x}),\dot{\phi}(s,\mathbf{y})]\Omega \ra 
= - \int d\mu^{2} \ \rho(\mu^{2}) \ \partial_{t}D(t-s,\mathbf{x}-\mathbf{y};\mu^{2})
\end{equation}
where the identity $\partial_{s}D(t-s,\cdot)=-\partial_{t}D(t-s,\cdot)$ for the integrand is obvious. On account of 
\begin{equation}
 \partial_{t}D(0,\mathbf{x}-\mathbf{y};\mu^{2}) = -i \delta^{(3)}(\mathbf{x}-\mathbf{y}) 
\end{equation}
we find that in the limit $s \rightarrow t$ the K\"allen-Lehmann representation in (\ref{KLCCR}) goes 
over to  
\begin{equation}\label{CCRLeh}
 \la \Omega | [\phi(t,\mathbf{x}),\dot{\phi}(t,\mathbf{y})]\Omega \ra 
= i \int d\mu^{2} \ \rho(\mu^{2}) \ \delta^{(3)}(\mathbf{x}-\mathbf{y}) 
= i [ \  1 +  \int d\mu^{2} \ \sigma(\mu^{2})  \ ] \delta^{(3)}(\mathbf{x}-\mathbf{y}), 
\end{equation}
and $Z^{-1} =  1 +  \int d\mu^{2} \ \sigma(\mu^{2})$ is the conclusion if the field $\varphi(x)$ obeys
the CCR (\ref{intCCR}). The canonical assertion $0\leq Z \leq 1$ and $Z=1$ for free fields is then an easy consequence.

\subsection{Wave-function renormalisation}
Strocchi contends that $Z \rightarrow 0$ upon removal of the cutoff because the spectral integral
\begin{equation}\label{spectra}
Z^{-1} =  \int d\mu^{2} \ \rho(\mu^{2}) =  1 +  \int d\mu^{2} \ \sigma(\mu^{2})
\end{equation}
diverges in four spacetime dimensions ''as a consequence of general non-perturbative arguments'' (\cite{Stro13}, p.51), at which point he cites Powers' 
and Baumann's papers \cite{Bau87, Pow67}. 

From our understanding of those nonperturbative arguments, we assume the idea behind Strocchi's statement is 
something like this: 
\begin{enumerate}
 \item free fields satisfy the CCR and besides $\int d\mu^{2} \ \rho(\mu^{2})=1$ is uncontroversial. Hence (\ref{CCRLeh}) makes total sense for free 
       fields, whereas
 \item interacting fields do not satisfy the CCR (in dimensions $d \geq 4$), ie something must go wrong. 
 \item Conclusio: (\ref{CCRLeh}) must diverge.
\end{enumerate}
Whatever the author had in mind, we do not find this convincing. Apart from the fact that Baumann's results do strictly speaking \emph{not pertain} to 
four spacetime dimensions, the problem may rather lie in the provision that a sharp-time Wightman field and its first derivative with respect to time exist.

In four spacetime dimensions, these dubious two objects possibly only exist in the trivial free case. One therefore cannot use them to conclude that the spectral integral (\ref{spectra}) diverges. 

However, on page 106 in \cite{Stro13}, Strocchi discusses the wave-function renormalisation constant $Z_{\Lambda}$ with a UV cutoff $\Lambda >0$ for 
the Dirac field $\psi$ of the \emph{derivative coupling model}
\begin{equation}
\La_{DC} = \frac{1}{2}[\partial_{\mu}\varphi\partial^{\mu}\varphi-m^{2}\varphi^{2}] + \ol{\psi}[i \gamma^{\mu}\partial_{\mu} - M]\psi 
- g (\ol{\psi} \gamma^{\mu} \psi) \partial_{\mu}\varphi
\end{equation}
and comes to the conclusion that $Z_{\Lambda}$ in $\psi_{r,\Lambda}(x) = Z_{\Lambda}^{-1/2} \psi_{\Lambda}(x)$ \emph{diverges} in the cutoff 
limit $\Lambda \rightarrow \infty$. However, in his treatment (\cite{Stro13}, Section 6.3), the wave-function renormalisation takes the form
\begin{equation}
 Z_{\Lambda}(g) = e^{-i\frac{1}{2}g^{2}\Delta^{+}_{\Lambda}(0)} ,
\end{equation}
where he denotes by $i\Delta^{+}_{\Lambda}(x)=(2\pi)^{-3} \int_{\Lambda} d^{4}p \ \delta_{+}(p^{2}-m^{2}) e^{-ip \cdot x}$ the two-point Wightman 
function of the free field with a UV cutoff. Since this obviously implies 
\begin{equation}
 \lim_{\Lambda \rightarrow \infty} Z_{\Lambda}(g) = \lim_{\Lambda \rightarrow \infty}  e^{-i\frac{1}{2}g^{2}\Delta^{+}_{\Lambda}(0)} = 0 ,
\end{equation}
we do not exactly find that $Z$ diverges, but he probably meant that is vanishes. 
However, this example is at least in accord with the divergence of the spectral integral (\ref{spectra}) and can thus be reconciled with his 
claim $Z \rightarrow 0$ made for the scalar field. 
 
To our mind, these discussions only point to the intricasies of the 'multiplicative renormalisation folklore': for the two-point function 
of the renormalised field, we find  
\begin{equation}
 \la \Omega | \varphi_{r}(x)\varphi_{r}(y) \Omega \ra = Z^{-1}\la \Omega | \varphi(x)\varphi(y) \Omega \ra
\end{equation}
which is, for example, also used to derive the Callan-Symanzik equation. 
Because the two-point function of the renormalised field $\varphi_{r}$ is finite and that of the unrenormalised field $\varphi$ diverges, $Z$ needs to
diverge and hence cannot vanish as proposed by Strocchi. 

We opine that this contradiction shows \emph{how poorly understood the connection between the K\"allen-Lehmann representation and canonical 
perturbation theory actually is}: in perturbation theory, $Z$ is given as a perturbative series with respect to the renormalised coupling $g_{r}$ 
(see §\ref{sec:Ren}) having two nasty properties: 
\begin{itemize}
 \item for finite regulator or cutoff, $Z(g_{r})$ has an asymptotic power series (clearly divergent),
 \item what is more, its coefficients diverge when the regulator (or cutoff) is removed, exacerbating things for any attempt to understand it nonperturbatively 
 through resummation schemes.
\end{itemize}
It can therefore never satisfy $0 \leq Z \leq 1$, which is a standard assertion in textbooks whenever the spectral representation (\ref{KaLC'}) is derived.
Take \cite{PeSch95} for example. On page 215, they construe $Z$ as ''... the probability for $\phi(0)$ (that is, the field at $x=0$, author's note) to
create a given state from the vacuum." This is because in their analysis, they find 
\begin{equation}
 Z = |\la \Omega | \phi(0) \lambda_{0} \ra|^{2} 
\end{equation}
where $| \lambda_{0} \ra$ is the zero-momentum state of the interacting theory and, therefore, $Z$ has to be within the unit interval. But their 
deliberations are typical and can in fact be traced back to \cite{BjoDre65}. However, before elaborating on this, we express at this point our contention  
that if we were to conduct an anonymous survey amongst theoretical physicists on this issue, the outcome would certainly be a collection of contrived 
narratives\footnote{Nowadays, video lectures on QFT abound on the internet. For example, the Perimeter Institute based in Canada sports a vast library of 
such recordings. In one lecture on the spectral representation, given by a physicist whose name we will not give away, 
a student says ''... but you get infinity!''. After some silence lasting a lengthy 8 seconds (!), the lecturer resumes by speaking of UV divergences,
and, gradually gaining back his normal speed of talking, he explains ''... if you cut the thing off, 
it really is true (the statement $0 \leq Z \leq 1$, author's note), but if you push the cutoff to infinity, then it won't be true anymore.'' }.

\subsection{Asymptotic scattering theory and wave-function renormalisation}\label{AsyS} 
In an attempt to connect two other disconnected stories of QFT, the renormalisation constant $Z$ found its way also into \emph{asymptotic scattering theory}.
From studying the literature, we glean the following brief history of developments regarding this issue.
\begin{enumerate}
 \item[1.] Lehmann, Symanzik and Zimmermann formulate a theory of scattering for quantised fields in \cite{LSZ55,LSZ57} stating explicitly in their abstract
       that \emph{''These equations contain no renormalization constants, but only experimental masses and coupling parameters. The main advantage over the 
       conventional formalism is thus the elimination of all divergent terms in the basic equations. This means no renormalization problem arises.''}
       And indeed, in these papers, no such constants appear: they introduce a spatially smeared scalar field (which they denote by $A(x)$, but never mind)
 \begin{equation}
  \varphi_{f}(t):=i\int d^{3}x \ [ f^{*}(t,\mathbf{x})\partial_{t}\varphi(t,\mathbf{x})- \partial_{t}f^{*}(t,\mathbf{x})\varphi(t,\mathbf{x})] ,
 \end{equation}
where $f(x)$ is a solution of the Klein-Gordon equation and integrable in some sense (see \cite{LSZ55}). This is essentially supposed to form a 
'quantised wave packet'. Then they state the \emph{asymptotic condition} for the existence of a free incoming field as
\begin{equation}\label{LSZ}
\la \alpha | \varphi_{f}(t) | \beta \ra \sim \la \alpha | \varphi_{f,in}(t) | \beta \ra  \hs{1.5} \mbox{as} \hs{0.3} t \rightarrow -\infty  .
\end{equation}
 \item[2.] In contrast to this, however, Bjorken and Drell cite \cite{LSZ55} in \cite{BjoDre65} (Section 16.3) and rephrase the asymptotic condition in 
 the form
\begin{equation}\label{Lsz}
 \la \alpha | \varphi_{f}(t)| \beta \ra \sim Z^{1/2} \la \alpha | \varphi_{f,in}(t) | \beta \ra  \hs{1} \mbox{as} \hs{0.3} t \rightarrow -\infty
\end{equation}
for the incoming field and for the outgoing accordingly.
They have slipped in a 'normalisation factor' $Z$ allegedly on the grounds that the matrix elements $\la \alpha |  \varphi_{in}(x) | \beta \ra$ need to be 'normalised'.
This is exactly the point where their exposition departs from \cite{LSZ55}, but to be fair, they take pedagogical care only to speak of a 
\emph{normalisation} constant within the bounds of Chapter 16. But Bjorken and Drell's departure is completed also in spirit when they subject incoming 
Dirac fields to the same procedure in Section 16.8 and identify the corresponding normalisation constant with the wave-function \emph{re}normalisation 
$Z_{2}$ in Section 19.7, where Chapter 19 is devoted to renormalisation.  
\end{enumerate}
They repeat at this point their interpretation of $Z_{2}$ as ''the probability of finding a 'bare-electron' state within the one-electron state of the interacting 
theory''. Here we see what great lengths physicists go to in order to make sense of everything. 

As far as the CCR/CAR are concerned, they pervade this 
textbook: in Section 16.3 Bjorken and Drell demand that the interacting scalar field obey the CCR with $\pi(x)=\partial_{t}\varphi(x)$ and come in 
Section 16.4 on the spectral representation to the same conclusion\footnote{The reader be warned: the statement 
$Z^{-1} =  1 +  \int d\mu^{2} \ \sigma(\mu^{2})$ takes in \cite{BjoDre65} a seemingly different form, namely $1 =  Z +  \int d\mu^{2} \ \sigma(\mu^{2})$. The resolution is that their 
spectral function is unrenormalised, ie dividing by $Z$ gives Lopuszanski's equation.} as Lopuszanski in (\ref{CCRLeh}). 

The reason why Bjorken and Drell placed the factor $Z$ \emph{before} the free incoming field $\varphi_{f,in}(t)$ is, in contrast to \cite{LSZ55}, that
they let $\varphi_{f}(t)$ in (\ref{Lsz}) be the unrenormalised field such that for the renormalised field, one gets the asymptotic identity
\begin{equation}\label{Lsz2}
 \la \alpha | \varphi_{r,f}(t) | \beta \ra \sim Z^{-1/2} \la \alpha | \varphi_{f}(t) | \beta \ra 
  \sim  \la \alpha | \varphi_{f,in}(t) | \beta \ra  \hs{1} \mbox{as} \hs{0.3} t \rightarrow -\infty 
\end{equation}
which draws the connection between the asymptotic condition of \cite{LSZ55} as given in (\ref{LSZ}) and of \cite{BjoDre65} given in (\ref{Lsz}). 
Not facilitating the comprehensibility of their arguments, however, they do not make this point clear 
anywhere in the text (which they could have in their chapter on renormalisation; great book though). 

\subsection*{Scattering theory of Haag \& Ruelle}
However, the asymptotic condition of Lehmann, Symanzik and Zimmermann in \cite{LSZ55} was of \emph{axiomatic nature}. No proof of any kind that the 
fields obey this condition was given. Small wonder given that Wightman's axioms were not yet formulated at the time. This situation changed when,
based on these axioms, the asymptotic condition was shown by Ruelle to be satisfied under some additional requirements \cite{Rue62}. The 
proof had already essentially been worked out by Haag in \cite{Ha58}, albeit at a slightly lower level of rigour (according to Ruelle). This rigorous 
form of LSZ scattering theory was hence dubbed \emph{Haag-Ruelle scattering theory} (see also §\ref{asymp}).

\section{Canonical quantum fields: too singular to be nontrivial}\label{sec:CanQuant}                
Let us start with an innocent-looking canonical free 
Hermitian scalar field $\varphi(x)$, given formally by its Fourier expansion
\begin{equation}\label{canqf}
 \varphi(x)=\int \frac{d^{4}p}{(2\pi)^{3}} \frac{1}{\sqrt{2E_{p}}} 
    \ [ e^{-ip \cdot x} a(\mathbf{p}) + e^{ip \cdot x} a^{\dagger}(\mathbf{p})],
\end{equation}
where $p_{0}=E_{p}=\sqrt{\mathbf{p}^{2}+m^{2}}$ is the energy of the scalar particle. The mode opersators
satisfy $a(\mathbf{p})\Psi_{0}=0$ and 
\begin{equation}\label{freeCCR}
 [a(\mathbf{p}),a(\mathbf{q})]=0=[a^{\dagger}(\mathbf{p}),a^{\dagger}(\mathbf{q})] \ , \hs{1} 
 [a(\mathbf{p}),a^{\dagger}(\mathbf{q})]=i \delta^{(3)}(\mathbf{p}-\mathbf{q}),
\end{equation}
as usual. The trouble starts as soon as we ask for the norm of the 'state' $\Psi =\varphi(x)\Psi_{0}$. If 
the canonical field $\varphi(x)$ is to be taken seriously as an operator at a sharp spacetime point 
$x \in \Mi$, this should be a valid question. However, applying (\ref{freeCCR}), we find 
$||\Psi || = ||\varphi(x) \Psi_{0}||=\infty$. The expedient is to smooth out the field with respect
to its spatial coordinates, as in 
\begin{equation}\label{smoothf}
 \varphi(t,f)=\int d^{3}x \ f(\mathbf{x})\varphi(t,\mathbf{x}) ,   
\end{equation}
where $f \in \Sw(\R^{3})$ is a Schwartz function. If we now compute the norm of the state vector 
$\Psi_{f}(t)=\varphi(t,f) \Psi_{0}$ we find 
\begin{equation}\label{Norm}
 ||\Psi_{f}(t)||^{2} = || \varphi(t,f) \Psi_{0} ||^{2} = 
                                     \la \Psi | \varphi(t,f^{*}) \varphi(t,f) \Psi_{0} \ra
 = \frac{1}{2} \int \frac{d^{3}p}{(2 \pi)^{3}} \frac{|\wt{f}(\mathbf{p})|^{2}}{\sqrt{\mathbf{p}^{2}+m^{2}}}.
\end{equation}
This integral is convergent on account of $\wt{f}$ being Schwartz, where $\wt{f}=\F f$ is the Fourier 
transform of the Schwartz function\footnote{An explicit and careful treatment starts with $f \in \De(\R^{3})$ of compact support to 
ensure the theorem of Fubini can be employed and then extends the result to $\Sw(\R^{3})$.} $f \in \Sw(\R^{3})$. 
Consequently, the state $\Psi_{f}(t)=\varphi(t,f) \Psi_{0}$ has finite norm and exists. 
But this does not hold for the state $\varphi(x) \Psi_{0}$.

We conclude that the canonical free field needs smearing at least in space to be well-defined on the vacuum.
Without smearing, it is merely a symbol. Nonetheless, computing the two-point function 
\begin{equation}\label{freeW}
 \la \Psi_{0} | \varphi(x)\varphi(y) \Psi_{0} \ra = 
 \frac{1}{2} \int \frac{d^{3}p}{(2 \pi)^{3}} \frac{e^{-i p \cdot (x-y)}}{\sqrt{\mathbf{p}^{2}+m^{2}}}
 =: \Delta_{+}(x-y;m^{2})
\end{equation}
yields a well-defined function for $x\neq y$, but has a pole where $x=y$, ie a short-distance singularity.

\subsection{Triviality of sharp-spacetime fields} The following pertinent theorem due to Wightman says that if one assumes a quantum field $\varphi(x)$
exists as an operator at a sharp spacetime point $x \in \Mi$ and is covariant with respect to a strongly
continuous representation of the Poincarè group, then it is trivial in the sense that it is merely
a multiple of the identity \cite{Wi64}:

\begin{Theorem}[Short distance singularities]\label{WightTHM}
 Let $\varphi(x)$ be a Poincaré-covariant Hermitian scalar field, that is,
 \begin{equation}\label{PoincCov}
  U(a,\Lambda)\varphi(x)U(a,\Lambda)^{\dagger} = \varphi(\Lambda x + a)
 \end{equation}
and suppose it is a well-defined operator with the vacuum $\Psi_{0}$ in its domain. Then the function
\begin{equation}
 F(x,y) = \la \Psi_{0} | \varphi(x)\varphi(y) \Psi_{0} \ra 
\end{equation}
is constant, call it $c$. Furthermore $\varphi(x)\Psi_{0}= \sqrt{c}\Psi_{0}$, ie $\varphi(x)$ is trivial and
thus 
\begin{equation}
\la \Psi_{0} | \varphi(x_{1})...\varphi(x_{n}) \Psi_{0} \ra = c^{n/2}.
\end{equation}
\end{Theorem}
\begin{proof} 
We follow \cite{Stro13}.
First note that Poincaré covariance (\ref{PoincCov}) implies 
\begin{equation}
F(x+a,y+a)=F(x,y) 
\end{equation}
which entails that this function depends only on $(x-y)$. We write $F(x,y)=F(x-y)$. 
$F(x)$ is continuous by the strong continuity of the Poincarè representation, ie through the covariance 
identity $\varphi(x)=U(x,1)\varphi(0)U(x,1)^{\dagger}$ inserted into the two-point function, 
\begin{equation}\label{stroCont}
F(x)=\la \Psi_{0} | \varphi(x) \varphi(0) \Psi_{0} \ra 
    = \la \varphi(0) \Psi_{0} | U(x,1)^{\dagger} \varphi(0) \Psi_{0} \ra 
    = \la \varphi(0) \Psi_{0} | U(x,1) \varphi(0) \Psi_{0} \ra^{*} \ .
\end{equation}
By virtue of this and the property 
$\int d^{4}x \int d^{4}y \ f(x)^{*} F(x-y)f(y) = || \varphi(f) \Psi_{0}||^{2} \geq 0$ with completely
smoothed-out field
$\varphi(f):= \int d^{4}x f(x)\varphi(x)$, we conclude that $F(x)$ is a continuous function of positive type.
The Bochner-Schwartz theorem tells us now that there exists a positive tempered measure $\mu$ on $\R^{4}$
such that 
\begin{equation}
 F(x) = \int e^{-i p \cdot x} \ d\mu(p),
\end{equation}
ie $F(x)$ is the Fourier transform of a positive tempered measure. If we use the spectral 
representation of the translation symmetry operator,
\begin{equation}
 U(x,1)=\int e^{i p \cdot x} d\mathsf{E}(p),
\end{equation}
plug it into (\ref{stroCont}), we see that the measure $\mu$ must be Poincaré invariant by $F(\Lambda x)=F(x)$ and
\begin{equation}\label{Spectr}
 F(x) = \la \Psi_{0} |\varphi(x)\varphi(0)\Psi_{0} \ra 
 = \int e^{-i p \cdot x} d\la \varphi(0) \Psi_{0} | \mathsf{E}(p)\varphi(0) \Psi_{0} \ra. 
\end{equation}
Then it follows that $\mu$ is of the form (\cite{ReSi75}, p.70)
\begin{equation}
d\mu(p) = c \ \delta^{(4)}(p) d^{4}p + b \ dm^{2} \ d^{3}p \ \frac{\rho(m^{2})}{\sqrt{\mathbf{p}^{2}+m^{2}}}
\ \hs{2} ( c,b \geq  0 ),  
\end{equation}
which is essentially the K\"allen-Lehmann spectral representation with spectral function 
\begin{equation}
 \rho(m^{2})\geq 0 \ , \hs{2} \supp (\rho) \subset [0,\infty).
\end{equation}
Because $F(x)$ is continuous at $x=0$, we have 
\begin{equation}\label{ContF}
 F(0) = \int d\mu(p) = c + b \int dm^{2} \int d^{3}p \ \frac{\rho(m^{2})}{\sqrt{\mathbf{p}^{2}+m^{2}}}
\end{equation}
which implies $b=0$ because of the UV divergence of the momentum integral. This means in particular 
$F(x)=F(0)=c$ and that the spectral measure $d\mathsf{E}(p)$ has 
support only at $p=0$, where $\mathsf{E}(0)=\la \Psi_{0} | \ \cdot \ \ra \Psi_{0}$, ie where it projects onto the vacuum
(the vacuum is the only state with vanishing energy). 
If we write this in terms of (\ref{Spectr}), we get
\begin{equation}
 c = F(0) = \int d\la \varphi(0) \Psi_{0} | \mathsf{E}(p)\varphi(0) \Psi_{0} \ra 
 = \la \varphi(0) \Psi_{0} | \Psi_{0} \ra \la \Psi_{0} | \varphi(0) \Psi_{0} \ra 
 = |\la \Psi_{0} | \varphi(0) \Psi_{0} \ra|^{2}.
\end{equation}
Using this, one easily computes $||(\varphi(x)-\sqrt{c})\Psi_{0}||^{2}=0$. 
\end{proof}

This result is insightful.
We know exactly which assumption cannot be true for the two-point function of the canonical free field 
in (\ref{canqf}) and are even able to put our finger on it: the function $F(x-y)=\Delta_{+}(x-y;m^{2})$ 
in (\ref{freeW}) is not continuous at $x-y=0$ as the integral diverges logarithmically in this case. 
The argument that led to this assumption can be easily traced back to the condition of strong continuity 
of the representation of the translation group, ie the requirement that the function
\begin{equation}
 x \mapsto \la \Psi | U(x,1) \Phi \ra = \la \Phi | U(x,1)^{\dagger} \Psi \ra^{*}  
\end{equation}
be a continuous function for all state vectors $\Psi, \Phi \in \Hi$. 
The erroneous assumption for our free field is therefore (as we know) that the state 
$\Psi = \Phi=\varphi(0)\Psi_{0}$ is one of these permissible state vectors. 

Interestingly enough, there is an analogy to Haag's theorem. 
\begin{itemize}
 \item \textsc{First of all}, the assumption that there exists a Poincarè-covariant sharp-spacetime field is too strong.
 \item \textsc{Secondly}, while the rigorous procedure takes well-reasoned steps and ends up with pleading triviality, the formal canonical 
       calculation leads to an infinite result. 
\end{itemize}
This is also exactly what happens in non-renormalised canonical perturbation theory which, by its very
nature, has to work with sharp-spacetime fields.

\subsection{Tempered distributions} However, Wightman's theorem is not applicable to the smoothed-out free 
field $\varphi(t,f)$ in (\ref{smoothf}). One reason is that Poincarè covariance cannot be formulated 
like in (\ref{PoincCov}) but has to be altered, in particular, time must also be smeared. 

The axiomatic approach to be introduced in the next section proposes to construe the 
two-point function (\ref{freeW}) as a symbol for a tempered distribution, ie
\begin{equation}\label{Wdistr}
 \Sw(\Mi) \times \Sw(\Mi) \ni (f,g) \mapsto W(f,g) = \int d^{4}x \int d^{4}y \ f^{*}(x) W(x-y) g(y).
\end{equation}
This amounts to defining a 'two-point' distribution $W \in (\Sw(\Mi) \times \Sw(\Mi))'$ by
\begin{equation}
W(f,g):= \la \Psi_{0} |  \varphi(f) \varphi(g) \Psi_{0} \ra 
\end{equation}
with completely smooth-out field opersators $\varphi(f)=\int d^{4}x \ f(x) \varphi(x)$ and
$\varphi(g)=\int d^{4}x \ g(x) \varphi(x)$. Because $\varphi(f)$ makes
sense as an operator and gives rise to distributions, Wightman called these objects \emph{operator-valued
distributions}. 

Note that from a conceptual and physical point of view, the smoothing operation imposes no restriction. 
Observable fields cannot be measured with arbitary precision, and smearing a field in both time and space with 
a test function of 
arbitarily small support is certainly permissible and not a big ask, as discussed by Bohr and Rosenfeld in \cite{BoRo33,BoRo50}. 

The smoothing has the nice effect that (\ref{Wdistr}) can be written in Fourier space as
\begin{equation}\label{WFour}
 W(f,g)= \int \frac{d^{4}p}{(2\pi)^{4}} \ \wt{f}^{*}(p) \wt{W}(p)  \wt{g}(p)
 =: \wt{W}(\wt{f},\wt{g}) 
\end{equation}
and that the too strong assumption of continuity of $F(x)=W(x)$ at $x=0$ in (\ref{ContF}) can now be 
replaced by the much weaker condition that $d\mu(p)=\wt{W}(p) d^{4}p$ be a well-defined Poincarè invariant 
distribution, ie
\begin{equation}\label{WdistFou}
\int h(p) \ d\mu(p) = c \ h(0) + b \int d\nu^{2} \int d^{3}p \ 
\frac{\rho(\nu^{2})}{\sqrt{\mathbf{p}^{2}+\nu^{2}}} \ h(p)
\end{equation}
for a Schwartz function $h \in \Sw(\Mi)$. 
Then, with this weaker requirement, the integral on the rhs of (\ref{WdistFou}) need not be muted, ie the 
choice $b\neq 0$ is perfectly acceptable unless the spectral function $\rho(\nu^{2})$ goes berserk and 
overpowers the factor $1/|\nu|$. 

The spectral representation of the free field's two-point distribution can be gleaned from comparing (\ref{freeW}) with (\ref{WdistFou}): we
read off $c=0$, $b=1/(2(2\pi)^{3})$ and $\rho(\nu^{2})=\delta(\nu^{2}-m^{2})$. 

Taking into account the singular nature of sharp-spacetime fields, Poincaré covariance (\ref{PoincCov}) is
reformulated for the smeared fields as 
\begin{equation}\label{PoinCovDist}
 U(a,\Lambda)\varphi(f)U(a,\Lambda)^{\dagger} = \varphi(\{ a,\Lambda \}f),
\end{equation}
where $(\{ a,\Lambda \}f)(x)=f(\Lambda^{-1}(x-a))$ is the transformed Schwartz function. 
As the reader may remember from §\ref{sec:WightHall}, Haag's theorem relies on the sharp-spacetime 
version (\ref{PoincCov}) of Poincarè covariance and cannot be applied to smeared fields. The reason is
that the time of unitary equivalence is fixed and sharp, not averaged. 

\subsection{Sharp-time fields}
We know from (\ref{Norm}) that a free scalar field need only be smeared with respect to space to become a
well-defined object. No one can tell whether this is actually the case for general (interacting) fields
and some doubt it, eg Glimm and Jaffe (\cite{GliJaf70}, p.380), Streater and Wightman (\cite{StreatWi00}, p.101). The
latter authors speak of 'examples' which suggest that smearing in space is not sufficient but do not give a reference for further reading. 

Powers and Baumann make use of 'relativistic' sharp-time fields in \cite{Pow67, Bau87,Bau88} in the following sense. 
Starting with an operator-valued distribution transforming under the Poincaré group as in (\ref{PoinCovDist}), Baummann demanded that 
for a Dirac sequence $\delta^{\epsilon}_{t} \in \Sw(\R)$ centred at time $t\in \R$ and any Schwartz function $f \in \Sw(\R^{n})$ in space dimension $n$,
the limit
\begin{equation}
\varphi(t,f) := \lim_{\epsilon \rightarrow 0} \varphi(\delta^{\epsilon}_{t} \te f)
\end{equation}
exist, where $\Mi = \R^{n+1}$.

\section{Wightman axioms and reconstruction theorem}\label{sec:WightAx}                              
Considering the issues incurred by working with sharp-spacetime and possibly also with sharp-time fields, 
it is no wonder that the following axioms due to Wightman and collaborators do not demand that general 
(interacting) quantum fields make sense as Hilbert space operators at sharp-spacetime points 
$x=(t,\mathbf{x})$ but only as operator-valued distributions. 

\subsection{Axioms for operator-valued distributions} 
Many authors quote \cite{Wi56} as a seminal paper for the Wightman axioms. This is strictly speaking not
true, as Wightman did not state them as such in this publication. He rather juggled with a few features 
that a reasonable QFT should bear without so easily falling prey to triviality results like Theorem \ref{WightTHM} 
in the previous section. 

In \cite{Wi56}, Wightman first investigates the 
consequences that relativistic covariance, local commutativity and positivity of the generator of time 
translations entail for the vacuum expectation values of scalar fields. He then discusses how these 
properties suffice to 'reconstruct' the theory (reconstruction theorem). However, he tentatively adds 
that a 'completeness requirement' should be fulfilled to recover the entire theory. 
This requirement is now part of the axioms as \emph{cyclicity of the vacuum} to be explained in due course. 

The axioms were first explicitly enunciated by Wightman and G\aa rding 
in an extensive article \cite{WiGa64} where they report on their reluctance to publish their results 
earlier. Although believing in their axioms' worth, they first wanted to make sure that nontrivial 
examples including free fields exist. 

Except for the numbering, we follow \cite{StreatWi00} in their exposition of the Wightman axioms. 
Although formulated for general quantum fields with any spin in their monograph, one has to say that 
\emph{the axioms can only be expected to hold for scalar and Dirac fields}. 
For photon fields, the axiom of Poincaré invariance turned out to 
be incompatible with the equations of motion for free photons, ie Maxwell's equations for the vacuum. This
is, of course, not the case for classical photon fields \cite{Stro13}. So in hindsight, it was
certainly a bit premature to include vector fields. 

We shall describe the issues arising for gauge theories in §\ref{sec:AxG}  
and in this section content ourselves with brief remarks. However, Wightman's axioms do not speak of any equation of motion for
the fields to satisfy. From this perspective, issues arising with Maxwell's equations can be ignored. 
Since conventional quantisation schemes for free fields always involve equations of motion,
one is reluctant to assent to this. 

But because free scalar and a vast class of superrenormalisable QFTs conform with these axioms 
\cite{GliJaf81}, we expect them to make sense at least for scalar and fermion fields.   
 
The axioms are organised in such a way that only the first one stands independently whereas the others
that follow rely increasingly on the ones stated before. Here they are.

\begin{itemize}
 \item \textbf{Axiom O (Relativistic Hilbert space)}. The states of the physical system are described by 
 (unit rays of) vectors in a separable Hilbert space $\Hi$ equipped with a strongly continuous unitary
 representation $(a,\Lambda) \mapsto U(a,\Lambda)$ of the connected Poincaré group $\Po$. 
 Moreover, there is a unique state 
 $\Psi_{0}\in \Hi$, called the \emph{vacuum}, which is invariant under this representation, ie
 \begin{equation}
  U(a,\Lambda) \Psi_{0} = \Psi_{0} \hs{2} \textrm{for all } (a,\Lambda) \in \Po .
 \end{equation}
\end{itemize}

This first axiom merely sets the stage for a relativistic quantum theory without specifying any 
operators other than those needed for the representation of the Poincaré group. 
It therefore has the number O. For photons, however, this 
is already problematic: the Hilbert space must be replaced by a complex vector space with 
a nondegenerate inner product which is a much weaker requirement (see §\ref{sec:AxG}, or \cite{Stei00}, for example). 
The next axiom ensures that the Lorentz group cannot create an 
unphysical state, eg by sending a particle on a journey back in time.

\begin{itemize}
 \item \textbf{Axiom I (Spectral condition)} The generator of the translation subgroup
 \begin{equation}
  i \frac{\partial}{\partial a^{\mu}}U(a,1)|_{a=0} = P_{\mu}
 \end{equation}
has its spectrum inside the closed forward light cone: $\sigma(P) \subset \overline{V}_{+}$ and 
$H=P_{0}\geq 0$, ie the time translation generator (=Hamiltonian) has nonnegative eigenvalues.
\end{itemize}

This axiom includes massless fields, ie fields without a mass gap.  While Axiom I seems fine at face value, 
it is in fact at odds with QED and raises serious questions for a general canonical QFT on account of the 
consequences it has in store for the vacuum expectation values. We shall come back to this point below.
The rest of the axioms introduce the concept of quantum fields and what properties they should have.

\begin{itemize}
 \item \textbf{Axiom II (Quantum fields)}. For every Schwartz function $f \in \Sw(\Mi)$ there are operators
 $\varphi_{1}(f),..., \varphi_{n}(f)$ and their adjoints 
 $\varphi_{1}(f)^{\dagger},..., \varphi_{n}(f)^{\dagger}$ on $\Hi$ such that the polynomial algebra
 \begin{equation}
 \Al(\Mi) = \la \ \varphi_{j}(f),\varphi_{j}(f)^{\dagger} : f \in \Sw(\Mi), j = 1, ..., n \ \ra_{\C}
 \end{equation}
has a stable common dense domain $\D \subset \Hi$, ie $\Al(\Mi) \D \subset \D$ which is also
Poincaré-stable, ie $U(\Po) \D \subset \D$.
The assignment $f \mapsto \varphi_{j}(f)$ is called \emph{quantum field}.
Additionally, the vacuum $\Psi_{0}$ is \emph{cyclic} for $\Al(\Mi)$ with respect to $\Hi$. This means 
$\Psi_{0} \in \D$ and the subspace
\begin{equation}\label{cycl}
 \D_{0} := \Al(\Mi) \Psi_{0} \subseteq \D
\end{equation}
is dense in $\Hi$. Furthermore, the maps 
\begin{equation}
 f \mapsto \la \Psi | \varphi_{j}(f) \Psi' \ra  \hs{2} (j = 1, ... ,n)
\end{equation}
are tempered distributions on $\Sw(\Mi)$ for all $\Psi,\Psi' \in \D$. 
\end{itemize}

As already alluded to in the previous section, it is due to this latter property, that a quantum field 
is referred to as an \emph{operator-valued distribution}. 
The canonical notion of a quantum field can be approximated by the assignment of a spacetime 
point $x \in \Mi$ to an operator $\varphi_{j}(f_{x})$ with a Schwartz function $f_{x}$ of 
compact support in a tiny (Euclidean) $\varepsilon$-ball around $x \in \Mi$. This avoids the aforementioned
ills of sharp-spacetime fields. Note that the operators in $\Al(\Mi)$ are not required to be bounded. 

Cyclicity (\ref{cycl}) expresses the condition that every (physical) state can be approximated to an arbitrarily 
high degree by applying the field variables to the vacuum. On account of the density of the so-obtained 
subspace $\D_{0}$, one can then, if necessary, reach any state in $\Hi$ after completion of $\D_{0}$ 
with respect to Cauchy sequences. Finally, nullifying all zero norm states eliminates unphysical remnants.

It is important to note that this property, together with the following remaining axioms, entail \emph{irreducibility} of the operator algebra
$\Al(\Mi)$. This means that if $C$ is an operator commuting with all field operators, then it is trivial, ie
\begin{equation}
 [\Al(\Mi),C] = 0 \hs{1} \Rightarrow \hs{1} C=c_{0}1,
\end{equation}
where $c_{0} \in \C$. The proof can be found in \cite{StreatWi00}, Theorem 4-5.

\begin{itemize}
 \item \textbf{Axiom III (Poincaré covariance)} The quantum fields transform under the (unitary 
 representation of the) Poincaré group according to 
 \begin{equation}\label{smPoin}
 U(a,\Lambda)\varphi_{j}(f)U(a,\Lambda)^{\dagger} = \sum_{l=1}^{n}S_{jl}(\Lambda^{-1})
 \varphi_{l}(\{a,\Lambda \}f),
 \end{equation}
 on the domain $\D$ where $S(\Lambda^{-1})$ is a finite-dimensional representation of the connected 
 Lorentz group $\Lo$ and 
 \begin{equation}
 ( \{a,\Lambda \}f)(x):=f(\Lambda^{-1}(x-a))
 \end{equation}
 is the Poincaré-transformed test function. 
\end{itemize}

For scalar fields, this takes the simple form $S(\Lambda^{-1})=1$, ie $S_{jl}(\Lambda^{-1})=\delta_{jl}$, as
we have seen in (\ref{PoinCovDist}). 
The next property is called \emph{locality} among proponents of the axiomatic approach 
and mostly \emph{(Einstein) causality} or \emph{microcausality} by practising physicists.

\begin{itemize}
 \item \textbf{Axiom IV (Locality, Causality)}. Let $f,g \in \Sw(\Mi)$ be of mutually spacelike-separated 
 support, ie $f(x)g(y) \neq 0$ implies $(x-y)^{2}<0$. Then,
 \begin{equation}\label{Caus}
  [\varphi_{j}(f),\varphi_{l}(g)]_{\pm} 
  = \varphi_{j}(f) \varphi_{l}(g) \pm \varphi_{l}(g)\varphi_{j}(f) = 0,
 \end{equation} 
for all indices (anticommutator '$+$' for fermions and commutator '$-$' for bosons).   
\end{itemize}
 
This last axiom accounts for the fact that signals cannot travel faster than light, ie measurements
at two different points in spacetime with spacelike separation do not interfere. With this interpretation,
however, it is questionable whether gauge fields or fermion fields, both unobservable, should be required to satisfy this axiom.
But to make sure that observables constructed from these unobservable fields conform with it, one may
retain it, although it may be one condition too many as it is possibly the case for QED 
(see §\ref{sec:AxG}).   

However, it should not be mistaken for the CCR or CAR (§\ref{sec:PowTHM}). Note that (\ref{Caus})
is an operator identity which does not say anything about the case when the supports of $f$ and $g$ are
not spacelike separated. We know that for a single free scalar field, this commutator is the distribution
\begin{equation}
 [\varphi(f),\varphi(g)]_{-}=\Delta_{+}(f,g)-\Delta_{+}(g,f) \hs{1} (\mbox{free field case}).
\end{equation}
It is interesting to see what happens if one assumes that the commutator of a generic scalar field 
yields a c-number, a case investigated by Greenberg \cite{Gre61}: one can show that 
\begin{itemize}
 \item locality is implied by Poincaré invariance of the commutator, in turn a consequence of Poincaré covariance of the field;
 \item the field can be decomposed into a positive and a negative energy piece.
\end{itemize}
A field with this property has therefore been named \emph{generalised free field} (see \cite{Stro93} for a concise treatment).

\subsection{Asymptotic fields}\label{asymp}
The Wightman axioms do not include the \emph{condition of asymptotic
completeness}. This essentially means that the field algebra $\Al(\Mi)$ contains elements which 
approach free fields in the limits $t \rightarrow \pm \infty$ and that the states these \emph{asymptotic
fields} create when applied to the vacuum fill up a dense subspace in the Hilbert space. 
Ruelle proved in \cite{Rue62} that the above axioms imply the existence of asymptotic states if the 
theory has a \emph{mass gap} and Buchholz succeeded in proving the massless case \cite{Bu75,Bu77}. 
But \emph{the existence of asymptotic states and fields does not imply asymptotic completeness} which is often written as
\begin{equation}
 \Hi_{in} = \Hi = \Hi_{out},
\end{equation}
where $\Hi_{in}$ and $\Hi_{out}$ are the Hilbert spaces of the incoming and outgoing particles. 
But \emph{if} asymptotic completeness is given, the existence of a unitary S-matrix is guaranteed.
We have already mentioned in §\ref{sec:EuHTHM} that the existence of the S-matrix has been 
proven for the superrenormalisable class $P(\varphi)_{2}$. 

Nevertheless, as there is no compelling evidence for asymptotic completeness, Streater and Wightman decided to
withdraw this condition from their list of axioms (Axiom IV in \cite{StreatWi00}, p.102). 

\subsection{Wightman distributions}
The axioms translate directly to a package of properties of the vacuum expectation values.  
Let us now for simplicity confine ourselves to a single scalar field. It is not difficult to prove 
that the Wightman distributions defined by
\begin{equation}\label{WiDist}
 W_{n}(f_{1}, ... , f_{n}):=\la \Psi_{0} |\varphi(f_{1}) ... \varphi(f_{n}) \Psi_{0} \ra
\end{equation}
have the following properties: 
\begin{enumerate}

 \item[\textbf{W1:}] \textsc{Poincaré invariance}. 
 $W_{n}(f_{1}, ... , f_{n})=W_{n}(\{a,\Lambda\}f_{1}, ... , \{a,\Lambda\}f_{n})$ for all Poincaré
 transformations $(a,\Lambda) \in \Po$. This is a simple consequence of the field's Poincaré covariance and the vacuum's Poincaré invariance.

 \item[\textbf{W2:}] \textsc{Spectral condition}. $W_{n}$ vanishes if one test 
 function's Fourier transform has its support outside the forward light cone, that is, if there is a $j$ such that
 $\wt{f}_{j}(p)=0$ for all $p \in \overline{V}_{+}$, then
 \begin{equation}\label{specprop}
  \wt{W}_{n}(\wt{f}_{1}, ... , \wt{f}_{n}) = W_{n}(f_{1}, ... , f_{n})= 0 .
 \end{equation}
 In this case one 
 says that $W_{n}$ (or better $\wt{W}_{n}$) has support inside the forward light cone $(\overline{V}_{+})^{n}$. 
 This property is a consequence of the spectral condition imposed by Axiom I. 
\end{enumerate}
 But notice what it entails. While this is all very well for free fields, it raises serious doubts in a general QFT. If we just take
 the renormalised propagator of a scalar field in momentum space, 
 \begin{equation}\label{propapicker}
  \wt{G}_{r}(p) = \frac{i}{p^{2}-m_{r}^{2}-\Sigma_{r}(p)+i0^{+}}
  = \lim_{\epsilon \downarrow 0} \frac{i}{p^{2}-m_{r}^{2}-\Sigma_{r}(p) + i \epsilon}
 \end{equation}
with physical mass $m_{r}>0$ and self-energy $\Sigma_{r}(p)$, we have to ask ourselves whether this thing can actually do us a favour and vanish for 
spacelike momenta. In the case of a free field, this is well-understood as the integration over the zeroth component picks up the on-shell particles. 
And, of course, one may assume that this mechanism also works for the distribution in (\ref{propapicker}).

But what about photons? This would mean that spacelike, ie t-channel photons effectively do not contribute to the two-point function. 
We already see here, the Wightman framework does not accommodate the Maxwell field in its edifice as straightforwardly and clearly as one might wish for! 

\begin{itemize}
 \item[\textbf{W3:}] \textsc{Hermiticity}. 
 $W_{n}(f_{1}, ... , f_{n})=W_{n}(f^{*}_{n}, ... , f^{*}_{1})^{*}$. This follows from 
 $\varphi(f)^{\dagger} = \varphi(f^{*})$. 
 
 \item[\textbf{W4:}] \textsc{Causality/Locality}. If $f_{j}$ and $f_{j+1}$ have mutually spacelike separated support, then
 \begin{equation}
  W_{n}(f_{1}, ... , f_{j}, f_{j+1}, ... , f_{n}) = W_{n}(f_{1}, ... , f_{j+1}, f_{j}, ... , f_{n}).
 \end{equation}
 
 \item[\textbf{W5:}] \textsc{Positivity}. 
  $\Psi_{f} = f_{0}\Psi_{0} + \sum_{n \geq 1}\varphi(f_{n,1}) ... \varphi(f_{n,n}) \Psi_{0}$
  is the form of a general state in $\D_{0}$. The property
 \begin{equation}
  \sum_{n\geq 0} \sum_{j+k=n} W_{n}(f^{*}_{j,j}, ... , f^{*}_{j,1}, f_{k,1}, ... , f_{k,k}) \geq 0
 \end{equation}
 is a consequence of the requirement $\la \Psi_{f} | \Psi_{f} \ra = ||\Psi_{f} ||^{2} \geq 0$, where $W_{0}=|f_{0}|^{2}\geq 0$.
 
 \item[\textbf{W6:}] \textsc{Cluster decomposition}. Let $a\in \Mi$ be spacelike. Then
 \begin{equation}\label{clust}
  \lim_{\lambda \rightarrow \infty} 
  W_{n}((f_{1}, ... , f_{j}, \{\lambda a,1\}f_{j+1}, ... , \{\lambda a,1\}f_{n})
  = W_{j}(f_{1}, ... ,f_{j})W_{n-j}(f_{j+1}, ... , f_{n}).
 \end{equation}
\end{itemize}
Distributions with these features are called \emph{Wightman distributions} because a given Wightman field
satisfying the above axioms gives rise to such distributions. 

What if one is handed a set of such distributions without any further information, is there a field theory with such distributions?  

\subsection{Reconstructing a quantum field theory}
The reconstruction theorem asserts just that, up to unitary equivalence: every given set of Wightman 
distributions is associated to an existent Wightman field theory.
A tentative first version of this result was published in \cite{Wi56} when, as we have already pointed
out, the axioms had not yet been formulated as such. It became a proper theorem once the axioms had
been formulated.   

Before we start let us have a look at an important result known as Schwartz's nuclear theorem
\cite{StreatWi00}. It states that for every multilinear tempered distribution 
$W \colon \Sw(\Mi)^{l} \rightarrow \C$ there exists a tempered distribution $S$ on $\Sw(\Mi^{l})$ such that 
\begin{equation}
 W(f_{1}, ... ,f_{l})= S(f_{1} \te ... \te f_{l}).
\end{equation}
We remind the reader that the $l$-fold tensor product of functions $f_{1}, ... , f_{l} \in \Sw(\Mi)$ is given by the function 
\begin{equation}
 (f_{1} \te ... \te f_{l})(x_{1},...,x_{l}) := f_{1}(x_{1})  ...  f_{l}(x_{l}),
\end{equation}
which is an element in $\Sw(\Mi^{l})$. This means practically that one can formally write 
\begin{equation}\label{WiDi}
 W(f_{1},...,f_{n})=\int d^{d}x_{1} \ ... \int d^{d}x_{n} \ \W(x_{1}, ..., x_{n})
 f_{1}(x_{1})  ...  f_{n}(x_{n})
\end{equation}
with $d$ being the dimension of spacetime $\Mi$.
We will identify both distributions and simply write $W(f_{1}, ... ,f_{l})= W(f_{1} \te ... \te f_{l})$ as is customary.
Here is the theorem.

\begin{Theorem}[Reconstruction theorem]\label{RecTHM}
 Let $\{ W_{n} \}$ be a family of tempered distributions adhering to the above list of properties W1-W6. 
 Then there is a scalar field theory fulfilling the Wightman axioms 0 to IV. Any other theory is unitarily 
 equivalent.
\end{Theorem}
\begin{proof}
See Appendix \ref{sec:Recon} .
\end{proof}

The proof is constructive. We sketch it briefly. The underlying concept is the so-called 
\emph{Borchers algebra} \cite{Bor62}. It is given by the vector space
\begin{equation}
 \B = \bigoplus_{n \geq 0} \Sw(\Mi^{n})
\end{equation}
of terminating sequences $(f_{0}, f_{1}, ... , f_{n}, 0, 0, ... )$ with $f_{j} \in \Sw(\Mi^{j})$ and
$\Sw(\Mi^{0}):=\C$. To make this space into an algebra, one defines the product by
\begin{equation}
 f \times h := (f_{0}h_{0}, f_{0} \te h_{1} + f_{1} \te h_{0}, ..., 
 (f \times g)_{n}, ... )
\end{equation}
in which $f_{0} \te h_{0} = f_{0}h_{0}$ is just the product in $\C$ and 
\begin{equation}
(f \times h)_{n} := \sum_{j+k=n} [f_{j} \te g_{k}] 
\end{equation}
is the $n$-th component of the product. The Wightman distributions $W_{n}$ are now represented by a so-called \emph{Wightman functional} $W$ 
on $\B$, given by
\begin{equation}
 W(f \times h) := \sum_{n \geq 0} \sum_{j+k=n} W_{n}(f^{*}_{j} \te h_{k}),
\end{equation}
where $W_{0}(f_{0}^{*} \te h_{0})=f_{0}^{*}h_{0}$. 
Using this functional, one defines an inner product on $\B$ by 
setting $\la f,h \ra := W(f \times h)$ and makes it into a Hilbert space by the standard procedures of
completion with respect to Cauchy sequences and nullifying zero norm states. A quantum field is then 
defined by the assignment of a Schwartz function $h \in \Sw(\Mi)$ to the operator $\varphi(h)$ on
the Hilbert space $\B$, given through the multiplication of a vector $g=(g_{0},g_{1}, g_{2},...) \in \B$ by the vector
$(0,h,0,...) \in \B$, ie the operator $\varphi(h)$ is the assigment 
\begin{equation}
  g \mapsto \varphi(h)g := (0,h,0,...) \times g =
  (0, h \te g_{0}, ... , h \te g_{n-1}, ... ),
\end{equation}
ie $(\varphi(h)g)_{n} = h \te g_{n-1}$ is the $n$-th component. The reader is referred to 
Appendix \ref{sec:Recon} for a complete account of the proof. 

\section{Proof of Haag's theorem}\label{sec:ProofHTHM}                                               
As we shall see in this section, Haag's theorem relies strongly on the analyticity properties of the
Wightman distributions (\ref{WiDi}). Although their Schwartz kernels $\W_{n}(x_{1}, ... , x_{n})$ are 
generalised functions, the idea is to see them as the boundary 
values of meromorphic and hence locally holomorphic functions in the sense of distribution theory. In fact,
this is the assertion of a theorem: any distribution in $\Sw(\Mi^{n})$ is the boundary value of a 
meromorphic function (see \cite{Scho08}, Section 8.5). One speaks of these meromorphic functions as
analytic continuations of the corresponding distributions. 

Before we present Haag's theorem, we have to first go through a small battery of results used in 
its proof: the edge-of-the-wedge theorem, the Reeh-Schlieder theorem, a corollary of it and the 
theorem of Jost and Schroer. 

\subsection{Analytic continuation of tempered distributions} 
A pedagogical example of the analytic continuation of a tempered distribution is the meromorphic function
\begin{equation}
 F(x+iy)=\frac{1}{x+iy} = \frac{x}{x^{2}+y^{2}} - i \pi \frac{y}{\pi(x^{2}+y^{2})}.
\end{equation}
It is perfectly holomorphic off the origin and gives rise to the tempered distribution 
\begin{equation}
 F(f) = \lim_{y \rightarrow 0} \int_{-\infty}^{+\infty} dx \ F(x+iy)f(x)
 = \Pd \int_{-\infty}^{+\infty} \frac{dx}{x}f(x) - i \pi f(0) ,
\end{equation}
where $\Pd \int$ is the Cauchy principle value integral. 

In the case of Wightman distributions, one can specify distinctly what values the imaginary parts are 
permitted to take. We roughly follow \cite{Stro13}. If we consider the two-point Wightman function
\begin{equation}
 \W(x-iy) = \int \frac{d^{4}q}{(2\pi)^{4}} \ e^{-iq \cdot (x-iy)} \ \wt{W}(q),
\end{equation}
we see that $q \cdot y > 0$ is required. Given the spectral property $\wt{W }(q)=0$ if 
$q \notin \fl$, then $y \in V_{+}$ suffices. This motivates the definition of the \emph{forward tube}
\begin{equation}
 \Tu_{n} := \Mi^{n} - i (V_{+})^{n}. 
\end{equation}
For the two-point Wightman distribution we then pick $\eta \in V_{+}$ and get the identity
\begin{equation}
W(f,h) = \lim_{t \downarrow 0}
\int d^{4}x_{1} \int d^{4}x_{2} \ f(x_{1}) \W(x_{1}-x_{2} - i t \eta) h(x_{2}).   
\end{equation}
The \emph{extended forward tube} is defined by $\Tu_{n}':=L_{+}(\C)\Tu_{n}$, where $L_{+}(\C)$ is the 
connected complex Lorentz group which is applied to each argument separately, ie
\begin{equation}
 \Tu_{n} \ni z = (z_{1}, ... , z_{n}) \mapsto \Lambda z := (\Lambda z_{1}, ... , \Lambda z_{n}) 
 \in \Tu_{n}'.  
\end{equation}
These transformations are defined by the two properties $\det (\Lambda) =1$ and the invariance property
\begin{equation}
 (\Lambda z) \cdot (\Lambda z) = z^{2} = x^{2} - y^{2} + i \ 2 x \cdot y, 
\end{equation}
where all multiplications are understood as Minkowski products and $x,y \in \Mi$. 
The kernel function can now be analytically continued to this extended tube by setting
\begin{equation}
 \W(\Lambda z):= \W(z) \hs{2} \forall z \in \Tu_{n}.
\end{equation}
Note that this is an unambiguous definition because different preimages are Lorentz-equivalent: 
let $u=\Lambda w = \Lambda' v$. Then 
\begin{equation}
 \W(u)= \W(w)= \W(\Lambda^{-1}\Lambda' v) = \W(v)
\end{equation}
because $\Lambda^{-1}\Lambda' \in L_{+}(\C)$. This is known as the \emph{theorem of Bargmann, Hall
and Wightman} (see \cite{Jo65}, Chapter 4). Note that because there is a complex Lorentz transformation 
such that $\Lambda z = -z$, the Wightman distributions are now also analytically
continued to the \emph{backward tube} $\Mi^{n} + i (V_{+})^{n}$.  

\subsection*{Jost points} 
However, the forward tube does not contain any real points because the imaginary parts of the complex 
arguments lie in $(V_{+})^{n}$, ie $z \in \Tu_{n}$ implies $\Im(z_{j}) \neq 0$ for all $j=1, ... , n$. 
An important property of the extended tube $\Tu_{n}'$ is that in contrast to $\Tu_{n}$, it \emph{does} 
contain real points, the so-called \emph{Jost points}, for which $\Im(z)=\Im((z_{1}, ..., z_{n}))=0$. Let us grant them a definition.

\begin{Definition}[Jost points]
 A point $z=(z_{1}, ..., z_{n})\in \Tu_{n}' \subset \C^{4n}$ in the extended forward tube with vanishing imaginary part, that is, $\Im(z)=0$, is called \emph{Jost point}.  
 By 
 \begin{equation}
  \emph{\mbox{co}}(z) = \left\{ \left. \mbox{$\sum_{j=1}^{n}$}\alpha_{j}z_{j} \ \right| 
  \  \mbox{$\sum_{j=1}^{n}$}\alpha_{j}=1, \forall j: \alpha_{j} \geq 0   \right\} \subset \C^{4}
 \end{equation}
we denote the convex hull of a point $z \in \C^{4n}$. 
\end{Definition}
Note that the convex hull of a point $z \in \C^{4n}$ lies in $\C^{4}$ for any $n \geq 1$.
The following theorem due to Jost characterises Jost points and their convex hull.
\begin{Theorem}[Jost]\label{JostT}
Let $z=(z_{1}, ... , z_{n}) \in \Tu_{n}'$ be a Jost point. Then the convex hull of this point, $\emph{\mbox{co}}(z) \subset \C^{4}$, is spacelike, ie consists of spacelike 
points only: 
\begin{equation}
 w \in \emph{\mbox{co}}(z) \hs{0.3} \Rightarrow \hs{0.3} w^{2}<0 .
\end{equation}
Conversely, a set of vectors $z_{1}, ... , z_{n} \in \Mi$ with spacelike convex hull comprise a Jost point in the extended forward tube $\Tu_{n}'$.
\end{Theorem}
\begin{proof}
We only prove the case $n=1$ (for general $n \in \N$ see \cite{StreatWi00}, pp.70,71).
If $z' \in \Mi$ is spacelike, then we can in a first step Lorentz-transform it to $z'=(z'_{0},z'_{1},0,0)$ with $z'_{1} > |z'_{0}|$. Next, we apply the complex Lorentz transformation
\begin{equation}\label{LorT}
 \begin{pmatrix} z_{0} \\ z_{1} \end{pmatrix}
 = \underbrace{\begin{pmatrix}  \cos \alpha & i \sin \alpha \\ i \sin \alpha & \cos \alpha  \end{pmatrix} }_{\Lambda \in L_{+}(\C)}
  \begin{pmatrix}  z_{0}' \\ z_{1}' \end{pmatrix}
\end{equation}
for $\sin \alpha > 0$. Then $\Im(z) \in V_{+}$, ie $z \in \Tu_{1}$, to be verified by the reader. This means in particular 
$z'=\Lambda^{-1}z\in \Tu_{1}'$. Hence any spacelike vector in $\Mi$ is a Jost point in the 
extended tube $\Tu_{1}'$, ie the set of Jost points is not empty and contains all spacelike points 
for $n=1$. Now let $z' \in \Tu_{1}'$ be a Jost point. Then there must be a complex Lorentz transformation 
$\Lambda \in L_{+}(\C)$ and a vector $z \in \Tu_{1}$ such that $z'=\Lambda z$ (otherwise $z'$ would not be
in $\Tu_{1}'$). We write $z=x+iy$ and
compute 
\begin{equation}
 (z')^{2} = (\Lambda z) \cdot (\Lambda z) = z^{2} =  x^{2} - y^{2} + i 2 x \cdot y.
\end{equation}
Because $\Im(z')=0$ by assumption, we have $x \cdot y = 0$ which implies $x^{2}<0$ because $y \in V_{+}$. 
Thus, $(z')^{2} = z^{2}= x^{2}-y^{2} <0$ and therefore all Jost points are spacelike. In the case $n>1$, only spacelike 
vectors whose convex hull is spacelike form a Jost point.  
\end{proof}

This has a very useful consequence for the two-point function because it depends only on one spacetime 
variable $\xi \in \Mi$: if one knows its values on the spacelike double cone, ie the set 
$\{ \xi \in \Mi: \xi^{2} < 0 \}$ which, as Jost's theorem informs us, habours all Jost points, then there 
is an open subset $\Op \subset (\Mi + i V_{+}) \cup (\Mi - i V_{+})$ in the forward tube and the backward 
tube, to which it can be analytically continued by applying the complex Lorentz group $L_{+}(\C)$. 
The so-called 'edge-of-the-wedge' theorem, to be presented next, 
then says that the two-point function is uniquely characterised there. 

For the higher-point functions, one needs more subtle arguments to show that the Jost points comprise a 
subset large enough to uniquely characterise them (see \cite{StreatWi00}, pp.70,71). 

\subsection{Edge of the wedge}
The next result, known as 'edge-of-the-wedge' theorem, is crucial for Haag's theorem because it guarantees in 
particular that the two-point function is sufficiently characterised on the spacelike points of $\Mi$, ie the Jost points of $\Tu_{1}' \subset \C^{4}$. 

\begin{Theorem}[Edge of the wedge]\label{edwed} 
Let $\Op \subset \C^{4n}$ be an open subset which contains a real open subset $E \subset \Op$. 
Suppose $F_{\pm}$ is holomorphic in $D_{\pm}:= [\Mi^{n} \pm i (V_{+})^{n}] \cap \Op$ and for any $y \in V_{+}$
one finds
\begin{equation}
\lim_{t \downarrow 0} F_{+}(x+it y) = \lim_{t\downarrow 0} F_{-}(x-it y)
\hs{1} \forall x \in E
\end{equation}
in the sense of distributions. Then there is a function $G$ holomorphic in an open complex neighbourhood 
$N$ of $E$ such that $G = F_{\pm}$ on $D_{\pm}$. 
\end{Theorem}
\begin{proof}
 See \cite{StreatWi00}, Theorem 2-16. 
\end{proof}

In the case of the two-point function $\W(\xi)$ one starts with $E=\{ \xi \in \Mi: \xi^{2}<0 \}$, 
ie the set of all spacelike vectors. Any complex open neighbourhood $N \subset \C^{4}$ of $E$ must contain 
both lightlike and timelike vectors in $\fl$ and $-\fl$. From there, it is clear that by following the 
Lorentz group's orbits, $\W(\xi)$ is given everywhere in $\fl \cup (-\fl)$. 

\subsection{Local operator algebras} The following result, called Reeh-Schlieder theorem, is  
interesting from a physical point of view despite its mathematical fancyness. It decribes the remarkable fact that for any open $E \subset \Mi$, the 
local operator algebra
\begin{equation}
 \Al(E) = \la \ \varphi(f),\varphi(f)^{\dagger} : f \in \De(E) \ \ra_{\C}
\end{equation}
is cyclic for the vacuum, ie $\Al(E)\Psi_{0} \subset \Hi$ is dense, no matter how small $E \subset \Mi$ is! 

Let us say we choose 
the open compact subset $E \subset \Mi$ to 
be very tiny. The theorem tells us that every state can be generated from the corresponding local field 
operators, ie the operators $\varphi(f)$ smeared with $f \in \De(E)$. 
This means physically that if we know what happens in the tiny region $E$ of spacetime $\Mi$, we know what may 
happen anywhere. 

In high energy physics, this makes sense: to find out what happens in a tiny subset $E$, we need to use a gargantuan amount of energy. 
Then, as we know from experiments, more particles (or particle species for more field types) will show up.  This suffices to 
be informed about what might occur anywhere in the universe under the same circumstances. Therefore, this result is not unphysical.

\begin{Theorem}[Reeh-Schlieder]\label{ReehSchl}
Let $E \subset \Mi$ be open. If $\Psi_{0}$ is cyclic for $\Al(\Mi)$, then it is also cyclic for $\Al(E)$.
This means in particular that vectors of the form
\begin{equation}
\Psi_{f} = f_{0}\Psi_{0} +\sum_{j=1}^{n} \varphi(f_{j,1}) ... \varphi(f_{j,j})\Psi_{0}   \hs{2} \mbox{('localised states')}
\end{equation}
with $f_{0} \in \C$ and $f_{j,1}, ... , f_{j,j} \in \De(E)$ for all $j=1, ... , n$ are dense in $\Hi$.
\end{Theorem}
\begin{proof}
We follow \cite{StreatWi00}. Let $\Psi \in \Hi$ be orthogonal to all vectors in $\Al(E)\Psi_{0}$, then 
\begin{equation}
 F(f_{1}, ... , f_{n}) = \la \Psi | \varphi(f_{1}) ... \varphi(f_{n}) \Psi_{0} \ra 
\end{equation}
is a translation-invariant distribution vanishing for all $f_{1}, ... , f_{n} \in \De(E)$. By the nuclear theorem, we may write
\begin{equation}
 F(f_{1}, ... , f_{n}) = F(f_{1} \te ... \te f_{n}).
\end{equation}
Then, for any $\eta \in (V_{+})^{n-1}$ there exists a meromorphic function $\mathcal{F}_{n}$ such that 
\begin{equation}\label{ReehSch}
F(f_{1} \te ... \te f_{n}) = \lim_{t \downarrow 0} 
        \int d^{d}x_{1} \ ... \int d^{d}x_{n} \ \mathcal{F}_{n}(\xi - i t \eta)f_{1}(x_{1}) ... f_{n}(x_{n}),
\end{equation}
where $\xi=(\xi_{1}, ... , \xi_{n-1})$ are the variables $\xi_{j}:=x_{j}-x_{j+1}$, $j=1, ... ,n-1$. The 
limit exists by virtue of the spectral property (Axiom I).
By assumption we know, however, that (\ref{ReehSch}) vanishes for $\xi \in \Mi^{n-1}$ with the property $(x_{1},...,x_{n}) \in E^{n}$, ie
\begin{equation}
\lim_{t \downarrow 0} \mathcal{F}_{n}(\xi - i t \eta) = 0 
\end{equation}
in the sense of distribution theory for all $f_{j} \in \De(E)$. By Theorem \ref{edwed} (edge-of-the-wedge) 
we conclude that $F$ vanishes on all elements $f_{1} \te ... \te f_{n} \in \Sw(\Mi)^{\te n}$: the argument in that theorem can be iterated to reach any 
point in $\Mi$. But this means $\la \Psi | \Al(\Mi) \Psi_{0} \ra=0$ and thus $\Psi=0$, on account of $\D_{0} = \Al(\Mi)\Psi_{0}$ being dense, guaranteed 
by Axiom II. This entails that $\Al(E)\Psi_{0}$ is dense in $\Hi$ because we have found that a vector $\Psi \in \Hi$ orthogonal to this set must vanish. 
\end{proof}

This theorem has an important consequence. Let $E':=\{ x \in \Mi | (x-y)^{2}<0 \ \forall y \in E \}$ be
the spacelike complement of $E$. Then there exists no annihilating operator in $\Al(E')$.

\begin{Corollary}\label{Comu}
Let $E \subset \Mi$ be open and $T \in \Al(E')$ such that $T \Psi_{0}=0$. Then $T=0$ weakly, ie there exists
no annihilator in $\Al(E')$. 
\end{Corollary}
\begin{proof}
Pick any $\Psi \in \Al(E)\Psi_{0}$ and let $P \in \Al(E)$ be such that $\Psi = P \Psi_{0}$. For any
$\Phi \in \D$ in the domain of the field operators we consider 
\begin{equation}
\la \Psi | T^{\dagger} \Phi \ra = \la T \Psi | \Phi \ra = \la T P \Psi_{0} | \Phi \ra 
= \la P T \Psi_{0} | \Phi \ra = 0
\end{equation}
since, by virtue of the field's locality, one has $[P,T]=0$ due to $P \in \Al(E)$ and $T \in \Al(E')$.  
Because $\Al(E)\Psi_{0}$ is dense, $T^{\dagger}\Phi = 0$ follows for any $\Phi \in \D$. 
Let now $\Psi \in \D$, then  
\begin{equation}
 \la T \Psi | \Phi \ra = \la \Psi | T^{\dagger} \Phi \ra = 0  \hs{1} \forall \Phi \in \D
\end{equation}
entails $T \Psi =0$ because $\D$ is dense. As this holds for all $\Psi \in \D$, the assertion follows because $T$ vanishes
weakly, ie $\la T \Psi | \Phi \ra = 0$ is true for all $\Psi,\Phi \in \D$.    
\end{proof}

\subsection{Haag's theorem II} 
The final piece we need for the proof of Haag's theorem is the so-called 
\emph{Jost-Schroer theorem} which says that a theory whose two-point function coincides with that of
a free field of mass $m>0$ is itself such a theory. 

\begin{Theorem}[Jost-Schroer Theorem]\label{JoSchro}
Let $\varphi$ be a scalar field whose two-point distribution coincides with that of a free field of mass $m>0$, ie
\begin{equation}\label{JostSch}
 \la \Psi_{0} |\varphi(f) \varphi(h) \Psi_{0} \ra
 =  \int \frac{d^{4}p}{(2\pi)^{4}} \  \wt{f}^{*}(p) \wt{\Delta}_{+}(p;m^{2}) \wt{h}(p),
\end{equation}
where $\wt{\Delta}_{+}(p;m^{2})=2 \pi \theta(p_{0})  \delta(p^{2}-m^{2})$. Then $\varphi$ is itself a 
free field of mass $m>0$. 
\end{Theorem}
\begin{proof}
If we define a free scalar field through the condition $\varphi([\Box + m^{2}]f)=0$ for all Schwartz 
functions $f \in \Sw(\Mi)$,
then the proof is short and simple. First, note that (\ref{JostSch}) implies 
\begin{equation}
 \la \Psi_{0} | \varphi([\Box + m^{2}]f) \varphi([\Box + m^{2}]h) \Psi_{0} \ra = 0
\end{equation}
and thus $||\varphi([\Box + m^{2}]f) \Psi_{0}||=0$, that is, $\varphi([\Box + m^{2}]f) \Psi_{0}=0$ for
any Schwartz function $f$. Then, by locality of the field (Axiom IV) and Corollary \ref{Comu}, 
$\varphi([\Box + m^{2}]f)=0$ and thus $\varphi$
is a free field in this sense. The full proof shows that this property entails that all Wightman 
distributions coincide with those of a free field, see Appendix \ref{sec:AppJoSch}.
\end{proof}

Since, as we have already mentioned in §\ref{ConJoG}, Pohlmeyer has proved this assertion for the massless case,
the Jost-Schroer theorem holds for fields with mass $m \geq 0$ \cite{Po69}. 
Now, finally, we shall see that the proof of Haag's theorem is implicitly contained in the package of 
the statements above together with the equality of the two-point functions (\ref{etVEV}) at equal times. 

\begin{Theorem}[Haag's Theorem]\label{Hath}
Let $\varphi$ and $\varphi_{0}$ be two Hermitian scalar fields of mass $m\geq 0$ in the sense of 
the Wightman framework. Suppose the sharp-time limits $\varphi(t,f)$ and $\varphi_{0}(t,f)$ exist and 
that at time $t=0$ these two sharp-time fields form an irreducible set in their respective Hilbert 
spaces $\Hi$ and $\Hi_{0}$. Furthermore let there be an isomorphism $V \colon \Hi_{0} \rightarrow \Hi$ such 
that at time $t$
\begin{equation}\label{intertw}
 \varphi(t,f) = V \varphi_{0}(t,f)V^{-1}.
\end{equation}
Then $\varphi$ is also a free field of mass $m \geq 0$. 
\end{Theorem}
\begin{proof}
We start by showing that $U(\mathbf{a},R)=VU_{0}(\mathbf{a},R)V^{-1}$ for the two representations of the 
Euclidean subgroups in $\Hi$ and $\Hi_{0}$, respectively. We use covariance with respect to the Euclidean 
subgroup (Axiom III) and unitary equivalence (\ref{intertw}):
\begin{equation}\label{Poin}  \begin{split}
& \varphi(t,f)U(\mathbf{a},R)^{\dagger}VU_{0}(\mathbf{a},R)V^{-1} =
U(\mathbf{a},R)^{\dagger}\varphi(t,\{\mathbf{a},R\}f)VU_{0}(\mathbf{a},R)V^{-1} \\
&= U(\mathbf{a},R)^{\dagger}V\varphi_{0}(t,\{\mathbf{a},R\}f)U_{0}(\mathbf{a},R)V^{-1} 
=  U(\mathbf{a},R)^{\dagger}VU_{0}(\mathbf{a},R)\varphi_{0}(t,f)V^{-1} \\
&= U(\mathbf{a},R)^{\dagger}V U_{0}(\mathbf{a},R) V^{-1}\varphi(t,f).  \end{split}
\end{equation}
On account of the irreducibility of $\varphi$'s field algebra, warranted by Axiom II, we have
\begin{equation}\label{Irr}
U(\mathbf{a},R)^{\dagger}VU_{0}(\mathbf{a},R)V^{-1} = c(\mathbf{a},R)1
\end{equation}
for some $c(\mathbf{a},R) \in \C$ which must be constant and, in fact, equal to $1$ due to the group 
property and unitarity of the representation (Axiom O). Thus, $U(\mathbf{a},R)V=VU_{0}(\mathbf{a},R)$. 
Let $\Omega$ and $\Omega_{0}$ denote the two vacua. Then 
\begin{equation}
U(\mathbf{a},R)V\Omega_{0} = VU_{0}(\mathbf{a},R)\Omega_{0} = V\Omega_{0},
\end{equation}
which means $V\Omega_{0}=a \Omega$. By unitarity of $V$ we have $|a|=1$. We are allowed to set $a=1$ 
(or absorb $a$ into the definition of $V$). The consequence of this is 
\begin{equation}
 \la \Omega_{0} | \varphi_{0}(t,f) \varphi_{0}(t,h) \Omega_{0} \ra 
 = \la \Omega_{0} |V^{-1} \varphi(t,f)VV^{-1} \varphi(t,h)V \Omega_{0} \ra
 = \la \Omega |\varphi(t,f)\varphi(t,h)\Omega \ra
\end{equation}
and hence 
$ \la \Omega_{0} | \varphi_{0}(t,\mathbf{x}) \varphi_{0}(t,\mathbf{y}) \Omega_{0} \ra  
= \la \Omega |\varphi(t,\mathbf{x})\varphi(t,\mathbf{y})\Omega \ra $ in the sense of distributions. 
The only singular point is where $\mathbf{x}=\mathbf{y}$, as we know because both expressions are the 
well-known two-point function of a free scalar field of mass $m \geq 0$ at equal times, ie  
\begin{equation}
\Delta_{+}(0, \mathbf{x}-\mathbf{y};m^{2}) 
=\la \Omega_{0} | \varphi_{0}(t,\mathbf{x}) \varphi_{0}(t,\mathbf{y}) \Omega_{0} \ra  
= \la \Omega |\varphi(t,\mathbf{x})\varphi(t,\mathbf{y})\Omega \ra. 
\end{equation}
For any spacelike point $z=x-y=(t,\mathbf{\xi})$, one can find a Lorentz transformation $\Lambda \in \Lo$ 
that takes it into the zero-time slice: $\Lambda z=z'=(t, \mathbf{\xi}')$. Therefore, we have by 
Lorentz invariance (implied by Axiom III) 
\begin{equation}
\Delta_{+}(x-y; m^{2}) = \la \Omega_{0} | \varphi_{0}(x) \varphi_{0}(y) \Omega_{0} \ra
= \la \Omega |\varphi(x)\varphi(y)\Omega \ra. 
\end{equation}
for $(x-y)^{2}<0$. By the edge-of-the-wedge Theorem \ref{edwed}, this means the two-point Wightman 
distribution of the field $\varphi$ agrees with that of the free field. 
Thus, on account of the Jost-Schroer Theorem \ref{JoSchro} in case $m>0$ and Pohlmeyer's version for $m=0$, 
$\varphi$ is a free field of mass $m \geq 0$. 
\end{proof}

\subsection{Summary: provisos of the proof} 
As the exposition shows, the proof of Haag's theorem makes use of all Wightman axioms and makes two 
additional strong assumptions: 
\begin{enumerate}
 \item the existence of sharp-time fields $\varphi_{0}(t,f), \varphi(t,f)$ forming an irreducible operator algebra at a fixed time $t$ and 
 \item unitary equivalence between both fields.
\end{enumerate}
Here is a summary of the salient stages of the proof, the provisos they rely on and their consequences.
\begin{itemize}
 \item Poincaré-invariance of both vacua (Axiom O), Poincaré covariance of both 
       fields (Axiom III) and their unitary equivalence (\ref{intertw}) jointly entail that the 
       two-point functions of both theories agree for spacelike separated points,
\begin{equation}\label{tw}
\la \Omega_{0} | \varphi_{0}(x) \varphi_{0}(y) \Omega_{0} \ra = \la \Omega |\varphi(x)\varphi(y)\Omega \ra
\hs{1} \mbox{if} \hs{0.3} (x-y) \in E,
\end{equation}
where $E=\{\xi \in \Mi : \xi^{2} <0\}$ is the spacelike double cone. 
\item  Jost's theorem (Theorem \ref{JostT}) tells us that $E$ is comprised of Jost points which are 
       the real points of the extended forward tube $\Tu_{1}'$. 
       This set is in turn the image of the complex Lorentz group when applied to the forward tube $\Tu_{1}=\Mi - iV_{+}$. 
       This means that starting from the spacelike double cone $E$ lying in $\Tu_{1}'$, the
       two-point functions can be analytically continued into the forward tube $\Tu_{1}$ where they 
       constitute an open set $\Op \subset \Tu_{1} \subset \C^{4}$. Axioms O - III are needed
       since these objects must be Poincaré-invariant distributions that satisfy the 
       spectral property (guaranteed by Axiom I). The edge-of-the-wedge theorem makes sure that the 
       analytic continuations of both two-point functions really agree, ie 
\begin{equation}\label{tw2}
\la \Omega_{0} | \varphi_{0}(x) \varphi_{0}(y) \Omega_{0} \ra = \la \Omega |\varphi(x)\varphi(y)\Omega \ra
\hs{1} \mbox{for all} \hs{0.3} (x-y) \in \Mi.
\end{equation}
 \item Given this result, the Jost-Schroer theorem can now be put in place. It relies on the 
       Reeh-Schlieder Theorem \ref{ReehSchl} and its Corollary \ref{Comu}. For these theorems to be 
       applicable, the spectral property of the Wightman distributions (Axioms O-III) and locality  
       (Axiom IV) must be fulfilled.   
\end{itemize}
So we see that \emph{all Wightman axioms must be fulfilled} for the proof of Haag's theorem. No use is 
made of the conjugate momentum field $\pi(t,f)$. Many proofs including the one in \cite{StreatWi00} use 
this field in addition and impose the intertwining relation
\begin{equation}\label{intertwPi}
 \pi(t,f) = V \pi_{0}(t,f)V^{-1}.
\end{equation}
However, it is not employed in any of the proof's steps except implicitly when irreducibility is used to 
obtain the relation between the two Poincaré representations in (\ref{Irr}). 

The fact is, the free field $\varphi_{0}$ itself is a fully fledged irreducible Wightman field without its conjugate field:
it generates a dense subspace in $\Hi_{0}$ all by itself even at one fixed time $t$! Irreduciblity of $\varphi_{0}(t,f)$ then follows 
from cyclicity, as already mentioned (see Axiom II). 
So the difference between our proof, which agrees with Roman's version in Section 8.4 of his 
textbook \cite{Ro69}, and Wightman's lies in the assumption of what constitutes an 
irreducible field algebra. 

Because the Wightman axioms do not include the CCR and the vacuum is cyclic for the free field which then,
as a consequence, is irreducible, the conjugate momentum field need not join the game. It only brings in 
an additional strong assumption about the existence of the time derivative of the field $\varphi$.

\section{Haag's theorem for fermion and gauge fields}\label{sec:AxG}                                 
At first glance, Haag's theorem in the above form holds only for scalar fields. It can,
with some notational inconvenience, also be formulated for a collection $\{ \varphi_{j} \}$ of scalar fields. 
But for fermion and gauge fields whose Poincaré covariance is nontrivial, things seem to look slightly different. Especially 
the Jost-Schroer theorem does not appear to carry over directly. We shall briefly see in this section that
while all arguments used in the proof of Haag's theorem go also through for fermion fields, this turns out
to be naive for gauge theories.

So assume that the Wightman framework with Axioms O - IV is true for a field with spin $s > 0$, 
then the equality of both theories' $n$-point functions on spacelike separated points 
is still given for $n \leq 4$. This is the assertion labelled 'generalised Haag's theorem' by Streater 
and Wightman in \cite{StreatWi00}. 

We shall briefly go through its proof and show that because the Jost-Schroer theorem also holds for 
anticommuting fields, Haag's theorem affects fields of nontrivial spin ($s>0$) as well. 
Note that it is irrelevant how many different collections of fields we put into both Hilbert spaces
as long as they are connected via the same intertwiner $V$. 

\subsection{Dirac fields} 
Let $\psi_{j}(t,\mathbf{x})$ and $\psi^{0}_{j}(t,\mathbf{x})$ be two Dirac fields in Wightman's sense, 
the latter field being free. We denote their respective vacua as $\Omega$ and $\Omega_{0}$. We start with
\begin{equation}
 \la \Omega | \psi_{j}(0,\mathbf{x}) \overline{\psi}_{l}(0,\mathbf{y}) \Omega \ra 
 =  \la \Omega_{0} | \psi^{0}_{j}(0,\mathbf{x}) \overline{\psi}^{0}_{l}(0,\mathbf{y}) \Omega_{0} \ra ,
\end{equation}
in which $\overline{\psi}=\psi^{\dagger}\gamma^{0}$. If we perform a Lorentz transformation $\Lambda$ such
that $\Lambda (0,\mathbf{x})=x$, $\Lambda (0,\mathbf{y})=y$ and $(x-y)^{2}<0$ on both sides, we get
\begin{equation}\label{fermHaa}
 \sum_{a,b} S_{ja}(\Lambda^{-1})S_{bl}(\Lambda^{-1})^{*} 
 \la \Omega | \psi_{a}(x) \overline{\psi}_{b}(y) \Omega \ra  
  =  \sum_{a,b} S^{0}_{ja}(\Lambda^{-1})S^{0}_{bl}(\Lambda^{-1})^{*}
 \la \Omega_{0} | \psi^{0}_{a}(x) \overline{\psi}^{0}_{b}(y) \Omega_{0} \ra ,
\end{equation}
with the corresponding spinor representations of the Lorentz group. If we strip the lhs of (\ref{fermHaa})
of its spinor representation matrices, we obtain
\begin{equation} 
 \la \Omega | \psi_{j}(x) \overline{\psi}_{l}(y) \Omega \ra  
  =  \sum_{a,b} \sum_{a',b'} S_{ja'}(\Lambda) S^{0}_{a'a}(\Lambda^{-1})
  S_{b'l}(\Lambda)^{*}S^{0}_{bb'}(\Lambda^{-1})^{*}
 \la \Omega_{0} | \psi^{0}_{a}(x) \overline{\psi}^{0}_{b}(y) \Omega_{0} \ra .
\end{equation}
For the relation between the vacua $V\Omega_{0}=\Omega$, derived along the same lines as in (\ref{Poin})
to obtain (\ref{Irr}), one chooses the same spinor representation for both fields. This is sensible and 
not a loss of generality because these representations are unitarily equivalent. 
Then follows for spacelike $\xi = x-y$
\begin{equation}\label{fermtwpt}
 \la \Omega | \psi_{j}(x) \overline{\psi}_{l}(y) \Omega \ra 
 =  \la \Omega_{0} | \psi^{0}_{j}(x) \overline{\psi}^{0}_{l}(y) \Omega_{0} \ra .
\end{equation}
It is equally valid for the $n$-point functions if $n \leq 4$. Analytic continuation and edge-of-the-wedge
theorem tell us that (\ref{fermtwpt}) holds everywhere in spacetime $\Mi$. This concludes the proof of Streater and Wightman's
generalised Haag's theorem. 

Because the rhs is the two-point function of the free propagator, one finds that for 
\begin{equation}
 j_{a}(x):=\sum_{b}(i \gamma^{\mu}\partial_{\mu} - m)_{ab}\psi_{b}(x)
\end{equation}
one has 
\begin{equation}
 \la \Omega | j_{a}(x)\overline{j}_{a}(y) \Omega \ra = 0 
\end{equation}
and thus $||j^{\dagger}_{a}(f)\Omega ||=0$ for every spinor index $a$ and Schwartz function $f$. 
Since Corollary \ref{Comu} is also true for anticommuting fields (\cite{StreatWi00}, remark on p.139) it 
follows $j^{\dagger}_{a}(f)=0$. Therefore, we see that the Jost-Schroer theorem and consequently 
Haag's theorem hold also true for fermion fields.  

\subsection{Axioms and Haag's theorem for gauge theories}
In contrast to what the previous paragraph suggests, gauge fields are different and prove a 
recusant species. 
As every physics student learns in a basic QFT course, the quantisation of
the photon field is not as straightforward as for fermions and scalar bosons. If we take the Wightman
axioms as an alternative quantisation programme for the photon field $A_{\mu}(x)$ and hence its field
strength tensor $F_{\mu \nu}(x)$, then, as Strocchi has found, we end up with a very bad form of triviality
\cite{Stro67,Stro70}. 

\begin{Theorem}[Strocchi]\label{Stroch}
Let $A_{\mu}(x)$ be an operator-valued distribution adhering to Axioms 0,I and II. Assume furthermore that
that the field is Poincaré covariant (Axiom III) according to
\begin{equation}
 U(a,\Lambda)A_{\mu}(x)U(a,\Lambda)^{\dagger} = 
 \Lambda^{\sigma}_{\ \mu} A_{\sigma}(\Lambda x+a)
\end{equation}
or assume that Axiom IV is satisfied, ie locality $[A_{\mu}(x),A_{\mu}(x)]=0$ for $(x-y)^{2}<0$. If the 
operator-valued distribution 
\begin{equation}
 F_{\mu \nu}(x) := \partial_{\mu}A_{\nu}(x)-\partial_{\nu}A_{\mu}(x) 
\end{equation}
satisfies the free Maxwell equation $\partial_{\mu}F^{\mu \nu}=0$, then  
\begin{equation}\label{twptF}
\la \Omega | F_{\mu \nu}(x) F_{\rho \sigma}(y) \Omega \ra = 0 \hs{1} \forall x,y \in \Mi,
\end{equation}
where $\Omega$ is the vacuum. If both axioms III and IV are fulfilled, that is, all Wightman axioms,
one finds $F_{\mu \nu}(x)=0$ which means $A_{\mu}(x)=\partial_{\mu}\varphi(x)$, ie the photon field is 
a gradient field. 
\end{Theorem}
\begin{proof}
 We shall only sketch the proof for the case of Poincaré covariance (Axiom III). For a thorough exposition,
 the reader is referred to \cite{Stro67} and \cite{Stro70}.  
 
 The first thing to show is that $D_{\mu \nu}(x,y):=\la \Omega | A_{\mu}(x)A_{\nu}(y) \Omega \ra$ is of the form
 \begin{equation}
  D_{\mu \nu}(\xi) = g_{\mu \nu} D_{1}(\xi) + \partial_{\mu} \partial_{\nu} D_{2}(\xi),
 \end{equation}
 where $\xi=x-y$ and $D_{j}(\xi)=D_{j}(\Lambda \xi)$ for $j=1,2$. This is familiar to physicists in momentum space. 
 Strocchi invokes a theorem by Araki and Hepp \cite{He63} to prove it.
 Because the gauge field obeys $[\delta_{\nu}^{\mu} \square - \partial^{\mu}\partial_{\nu} ]A_{\mu}(x)=0$, one has 
 \begin{equation}
  [\delta_{\nu}^{\mu} \square - \partial^{\mu}\partial_{\nu}] D_{\mu \rho}(\xi) = 0.
 \end{equation}
and hence $[g_{\mu \nu} \square - \partial_{\mu}\partial_{\nu} ]D_{1}(\xi)=0$. A lemma asserting
that a Lorentz-invariant function $F(x)$ with 
\begin{equation}
 [g_{\mu \nu} \square + \alpha \partial_{\mu}\partial_{\nu} ]F(x)=0 \ , \hs{1} \alpha \notin \{ 0,4 \}
\end{equation}
must be constant incurs $D_{1}(\xi)=\mbox{const.}$ from which (\ref{twptF}) follows. If locality comes
on top, $F_{\mu \nu}(x)=0$ is proven from $|| F_{\mu \nu}(f)\Omega|| = 0$ by invoking the Jost-Schroer 
theorem.  
\end{proof}

This result tells us that the photon field cannot reasonably be quantised by postulating the 
Wightman axioms in the form given in §\ref{sec:WightAx}. However, the situation can be remedied by
abandoning the Maxwell equation $\partial_{\mu}F^{\mu \nu}=0$ in favour of 
\begin{equation}
\partial_{\mu}F^{\mu \nu} = - \partial^{\nu} (\partial_{\mu}A^{\mu})
\end{equation}
with the extra condition for the rhs' annihalators that $(\partial_{\mu}A^{\mu})^{+}\Hi'=0$ on a subspace
$\Hi' \subset \Hi$, as explained in \cite{StroWi74}) which we shall draw on in the following. 

When part of a quantisation scheme, this is known as the \emph{Gupta-Bleuler condition}, as introduced 
in \cite{Gu50} and \cite{Bleu50}. 
The problem is that the inner product needed to define the Wightman distributions of QED is no longer 
definite, even on $\Hi'$. Let us denote this inner product by $(\cdot,\cdot)$. On $\Hi'$, one finds at least
that it is positive semidefinite, ie
\begin{equation}\label{Hpri}
 (\Psi,\Psi) \geq 0 \hs{2} \forall \Psi \in \Hi',
\end{equation}
but $\Phi \in \Hi'$ with $(\Phi,\Phi)=0$ does not mean $\Phi$ vanishes, it may well be a nonvanishing zero-norm state ($\Phi \neq 0$).
We denote the subspace of such vectors by $\Hi''$. Then, as a consequence of the Gupta-Bleuler condition, Maxwell's equations 
\begin{equation}\label{weakG}
(\Psi, \partial_{\mu}F^{\mu \nu}\Phi ) = 0 \hs{2} \forall \Psi, \Phi \in \Hi'
\end{equation}
are fulfilled on $\Hi'$. As a last step one obtains a physical Hilbert space $\Hi_{\mbox{\tiny phys}}$
by taking the quotient $\Hi_{\mbox{\tiny phys}}:=\Hi'/\Hi''$ through 
nullifying all zero-norm states. In fact, the original
vector space of states $\Hi$ can be recovered as a Hilbert space provided one condition is met: there
exists a sequence of seminorms $(\rho_{n})$ such that $\rho_{n}$ is a seminorm on $\Sw(\Mi^{n})$ and the
Wightman distributions satisfy 
\begin{equation}\label{YngCond}
 |W_{n+m}(f_{n}^{*} \te g_{m})| \leq \rho_{n}(f_{n}) \rho_{m}(g_{m})
\end{equation}
for any $f_{n} \in \Sw(\Mi^{n})$ and $g_{m} \in \Sw(\Mi^{m})$. Then there exist a bounded Hermitian 
operator $\eta$ on $\Hi$, called \emph{metric operator}, and a (genuine) scalar 
product $\la \cdot , \cdot \ra$ such that 
\begin{equation}
 (\Phi,\Psi)=\la \Phi, \eta \Psi \ra 
\end{equation}
for all $\Phi, \Psi \in \D$, where $\D$ is the dense subspace generated by the field algebra.
In words, the indefinite inner product $(\cdot,\cdot)$ of the Gupta-Bleuler formalism can be embedded 
into a Hilbert space through the employment of a sesquilinear form known as metric operator. 
This is in fact part of a reconstruction theorem of some sort for QED discussed by Yngvason in \cite{Yng77}
to which we refer the interested reader for details. 

Note that the Wightman distributions of the gauge field $A_{\mu}$ and its field strength tensor 
$F_{\mu \nu}$ are given by the vaccum expectation values with respect to the indefinite inner product.
On account of these complications, the Wightman axioms must be altered for gauge theories. As we have
already explained in §\ref{sec:WightAx}, it is Axiom O which needs to be modified: the Hilbert space
is replaced by an inner product space $\Hi$ whose inner product $(\cdot, \cdot)$ is nondegenerate.
Because one can construct a Hilbert space structure from it, one has a so-called \emph{Krein space}
\cite{Stro93}. So here is Axiom O for gauge theories \cite{Stei00}:

\begin{itemize}
 \item \textbf{Axiom O' (Relativistic complex vector space with nondegenerate form)}. 
 The states of the physical system are described by (unit rays of) vectors in a separable complex vector space $\Hi$  
 equipped with a nondegenerate form and a strongly continuous unitary representation $(a,\Lambda) \mapsto U(a,\Lambda)$ of the connected 
 Poincaré group $\Po$. Moreover, there is a unique state 
 $\Psi_{0}\in \Hi$, called the \emph{vacuum}, which is invariant under this representation, ie
 \begin{equation}
  U(a,\Lambda) \Psi_{0} = \Psi_{0} \hs{2} \textrm{for all } (a,\Lambda) \in \Po .
 \end{equation}
\end{itemize}
To guarantee that a physical Hilbert space can be constructed, Axiom II must be augmented by Yngvason's
condition (\ref{YngCond}), as proposed by Strocchi \cite{Stro93}. One may then be content with these modified axioms for photon fields and prove 
Haag's theorem invoking the same arguments as we have in the case of Dirac fields above. 

However, the following deliberations show that this contention cannot be maintained. The presented results
reveal that any nonfree Maxwell theory is fundamentally at loggerheads with Wightman's framework,
probably the reason why Strocchi believes QED and nonabelian gauge theories are not afflicted by 
Haag's theorem \cite{Stro13}. 
 
\subsection*{Maxwell's equations} 
To depart from free Maxwell theory and implement interactions with matter, one needs to introduce \emph{charged fields}.
We shall see now that in the presence of such fields, the last step to construct $\Hi_{\mbox{\tiny phys}}$ 
leads to a triviality result. Therefore, as yet, \emph{the last word on the Wightman axioms for QED has not 
been spoken.} 

Before we jump to a conclusion about QED, let us survey the pertinent results obtained
by Ferrari, Picasso and Strocchi \cite{FePStro74}. A curious aspect of their exposition is that they make
no reference to the gauge field $A_{\mu}$ and the corresponding Lagrangian formalism. Instead, they only
use the field strength tensor $F_{\mu \nu}$, the charge current $j^{\nu}$ and a charged field $\phi$
which we will introduce now.

First of all, $F_{\mu \nu}$ and $j_{\mu}$ are operator-valued distributions, both local and relativistic in
the sense of tensor fields, ie transforming under $\Po$ according to 
\begin{equation}\label{PoiVeTe} 
 U(a,\Lambda)j_{\mu}(x)U(a,\Lambda)^{\dagger} = \Lambda^{\sigma}_{\ \mu} j_{\sigma}(\Lambda x+a) \ , 
 \hs{0.3}
 U(a,\Lambda)F_{\mu \nu}(x)U(a,\Lambda)^{\dagger} 
 = \Lambda^{\sigma}_{\ \mu} \Lambda^{\rho}_{\ \nu} F_{\sigma \rho}(\Lambda x+a).
\end{equation}
The field strength is supposed to be an antisymmetric tensor: $F_{\mu \nu}=-F_{\nu \mu}$. 
The charge current $j^{\mu}$ gives rise to a charge through its zeroth component
\begin{equation}\label{char}
 Q_{R} := j^{0}(f_{u} \te f_{R})= \int d^{4}x \ f_{u}(x^{0})f_{R}(\mathbf{x})j^{0}(x)
\end{equation}
where $f_{R} \in \De(\R^{3})$ is such that $f_{R}(\mathbf{x})=1$ inside a ball of radius $R$, while vanishing rapidly outside the ball and
$f_{u}$ has compact support in $(-u,u) \subset \R$ with $\int_{\R} dx^{0}f_{u}(x)=1$. 
A scalar field $\phi$ is called \emph{local relative to} $j^{\mu}$ in case it satisfies
\begin{equation}
 [\phi(x),j^{\mu}(y)]=0 
\end{equation}
for spacelike distance, ie $(x-y)^{2}<0$. A field $\phi$ of this type is said to have charge $q$ if 
\begin{equation}\label{gaugetrans}
\lim_{R \rightarrow \infty}[Q_{R},\phi(f)]=-q \phi(f)
\end{equation}
for any $f \in \De(\Mi)$. This introduces a global gauge transformation ('gauge transformation of the first kind'). Then, one has

\begin{Theorem}[No Maxwell equations]\label{stroMax}
 Let $j_{\mu}=\partial^{\alpha}F_{\alpha \mu}$ and $\phi$ be local relative to $j_{\mu}$. Then 
 \begin{equation}
 \lim_{R \rightarrow \infty}[Q_{R},\phi(f)] = 0,
 \end{equation}
 ie the Maxwell equations cannot hold if the charge current is supposed to generate a nontrivial gauge transformation of the first kind (\ref{gaugetrans}).
\end{Theorem}
\begin{proof} Sketch, for details see \cite{FePStro74}. The main argument is that 
\begin{equation}\label{noMaxE}
[Q_{R},\phi(f)] = [j^{0}(f_{u} \te f_{R}),\phi(f)]=-[F^{\mu 0}(\partial_{\mu}(f_{u} \te f_{R})),\phi(f)] = -[F^{j 0}(f_{u} \te \partial_{j}f_{R}),\phi(f)] 
\end{equation}
vanishes in the limit because $f_{R}$ is constant in the region where the commutator receives 
contributions from the charge current integral (\ref{char}), ie $\partial_{j}f_{R}=0$ inside the ball. 
\end{proof}
To cure this pathology, we abandon the strong form of Maxwell's equations in Theorem \ref{stroMax} and
replace it by their weaker Gupta-Bleuler form 
\begin{equation}\label{GuB}
 (\Psi , [\partial_{\mu}F^{\mu \nu}(f)-j^{\mu}(f)]\Phi)  = 0  \hs{2}  \forall \Psi,\Phi \in \Hi'
\end{equation}
and for all test functions $f \in \De(\Mi)$, where, as in the free case, $\Hi'\subset \Hi$ is 
the subspace on which the indefinite form $(\cdot, \cdot )$ is positive 
semidefinite\footnote{$(\Phi,\Phi)\geq 0$ for all $\Phi \in \Hi'$ but $(\Phi,\Phi)= 0$ does not imply $\Phi =0$.}.   

The next result tells us that the Gupta-Bleuler strategy comes at an unacceptable price \cite{FePStro74}.

\begin{Theorem}[No charged fields]\label{NoCf}
Assume the common domain $\D \subset \Hi'$ of $j_{\mu}, F_{\mu \nu}$ and $\phi$ is dense in $\Hi'$ and stable 
under both fields' action. Let the Gupta-Bleuler condition (\ref{GuB}) be satisfied. Then either there are
no charged fields $\phi$ or 
\begin{equation}
 (\Phi,\phi(f)\Psi)=0 \hs{2} \forall \Psi,\Phi \in \Hi',
\end{equation}
ie $\phi$ vanishes weakly on $\Hi'$.
\end{Theorem}
\begin{proof}
We define $B^{\nu}:=\partial_{\mu}F^{\mu \nu}-j^{\nu}$ and note that on account of Theorem \ref{stroMax},
we obtain the limit $\lim_{R \rightarrow \infty}[B^{0}(f_{u} \te f_{R}),\phi(f)]=q\phi(f)$ because 
$\lim_{R \rightarrow \infty}[\partial_{\mu}F^{\mu 0}(f_{u} \te f_{R}),\phi(f)]=0$ from (\ref{noMaxE}) 
and the definition of a charged field in (\ref{gaugetrans}). This entails that 
\begin{equation}
\lim_{R \rightarrow \infty} ( \Phi, [B^{0}(f_{u} \te f_{R}),\phi(f)] \Psi )
= q ( \Phi, \phi(f) \Psi ). 
\end{equation}
But because the lhs vanishes for $\Phi, \Psi \in \D$ due to the Gupta-Bleuler condition $B^{\nu}(f)\D=0$
and $\D \subset \Hi'$ is dense in $\Hi'$, we have either $q=0$ or $( \Phi, \phi(f) \Psi ) =0$ for 
$\Phi, \Psi \in \Hi'$.
\end{proof}
This tells us that neither the strong nor the weak Gupta-Bleuler form of Maxwell's equations with nontrivial currents seem to be compatible
with the axiomatic framework. 

\subsection{Haag's theorem in QED} 
In the light of the results discussed in this section, which undoubtedly pertain to QED, even though charged scalar field were used, 
one may suggest two extreme stances regarding Haag's theorem for QED.
\begin{enumerate}
\item[(St1)] The source-free Maxwell's equations in the sense of the covariant formalism due to Gupta and Bleuler,
\begin{equation}\label{Maxw}
 (\Hi_{\mbox{\tiny phys}},F^{\mu \nu}(\partial_{\mu}f)\Hi_{\mbox{\tiny phys}}) = 0  \hs{2} \forall f \in \Sw(\Mi)
\end{equation}
      cohere with Wightman's axiomatic framework in the adapted form. Consequently, Haag's theorem applies fully to QED.
      No one knows whether and in what sense differential equations for the interacting operator fields are fulfilled. 
      We should therefore not impose such equations and require only (\ref{Maxw}) for the quantisation of the free Maxwell field.
\item[(St2)] Maxwell's equations with nontrivial electromagnetic currents 
\begin{equation}
 (\Hi_{\mbox{\tiny phys}},[F^{\mu \nu}(\partial_{\mu}f)+j^{\nu}(f)]\Hi_{\mbox{\tiny phys}}) = 0  \hs{2} \forall f \in \Sw(\Mi)
\end{equation}      
      are essential for QED and should not be abandoned. These equations are not compatible with Wightman's axioms even in the gauge-adapted version
      because they forbid charged fields. Haag's theorem does therefore not hold in the strict sense of coinciding vacuum expectation values in 
      QED\footnote{Haag's theorem relies heavily on the validity of the Wightman framework.}. 
\end{enumerate}
We are inclined to a position that partially embraces the first stance while leaving room for the second. This is
to be understood as follows. 

Given the arguments presented in §\ref{sec:Fock} against the existence of the interaction picture in its 
canonical form, we are in no doubt that the interaction picture cannot exist in QED. The UV divergences
one encounters in perturbation theory of QED show what happens if one assumes otherwise. Yet the 
canonical theory does not stop there but changes the rules of the game drastically by renormalisation.

While unrenormalised QED almost surely falls prey to Haag's theorem, we contend that \emph{renormalised}, ie
physical QED, which has to be clearly distinguished from its unrenormalised cousin, is safe from it:

\begin{enumerate}
\item[1.] Renormalised QED yields nontrivial results in perturbation theory that agree nicely with
         experiment.
\item[2.] As will become clear in §\ref{sec:RCircH},
         we contend that Haag's theorem cannot be applied to renormalised QED for the very
         reason that it is \emph{almost surely not unitary equivalent to a free theory}, where we remind the reader that this form of
         equivalence is one of Haag's theorem's core provisions. 
\item[3.]The validity of Wightman's axiomatic framework is dubitable given the results by Strocchi and
         collaborators discussed in this section, especially Theorem \ref{NoCf}.  
\end{enumerate}
Besides, and this opens flanking fire against one of Wightman's axioms in QED which (to the ken of the author) seems to have been so far been overlooked: 
the spectral condition for photons (Axiom I). Because the t-channel of fermion interactions contributes to the four-point scattering amplitude in QED, one simply 
needs and routinely utilises the concept of \emph{virtual, that is, spacelike photons}!
If we consider what the spectral condition (Axiom I) entails for the vacuum expectation values of a general QFT and take into account the results 
that perturbative approaches have brought to light so far, \emph{the contrast could hardly be starker: the spectral condition for photons is probably 
never satisfied}, at least to the best of our knowledge!

Regarding Maxwell's equations, a possibly existent reconstructed renormalised QED may have observables 
which in some subspace of $\Hi_{\mbox{\tiny phys}}$ satisfy Maxwell's equations, albeit within a form of the Wightman framework that, to this day, is 
still inconceivable.

\part{Renormalisation wrecks unitary equivalence}\label{part3}

\section{The theorem of Gell-Mann and Low}\label{sec:GeMaLo}                                         
The operations performed in the following derivation of the Gell-Mann-Low formula are purely 
formal and not well-defined in quantum field theory (QFT). Gell-Mann and Low simple assume that both 
free and interacting Hamiltonians are given as bona fide operators on a common Hilbert space $\Hi$ with their individual 
ground states, ie the vacua. 

\subsection{Review of the interaction picture}  
We first remind ourselves of the 3 pictures in quantum theory, namely Schr\"odinger, Heisenberg and 
interaction picture. The latter is also known as Dirac picture. Let 
\begin{equation}
 \varphi(t,\mathbf{x}) = e^{iHt} \varphi(\mathbf{x}) e^{-iHt}  \hs{2} (\mbox{Heisenberg picture})
\end{equation}
be the Heisenberg picture field and
\begin{equation}
 \varphi_{0}(t,\mathbf{x}) = e^{iH_{0}t} \varphi(\mathbf{x}) e^{-iH_{0}t}   \hs{1.6} (\mbox{interaction picture})
\end{equation}
the interaction picture field both at time $t$, where $\varphi(\mathbf{x})$ is the time-independent Schr\"odinger picture field,
$H$ the Hamiltonian of the full interacting theory and $H_{0}$ that of the free theory. Both pictures are consequently intertwined 
according to 
\begin{equation}\label{HeiInt}
 \varphi(t,\mathbf{x}) = e^{iHt} \underbrace{e^{-iH_{0}t}  \varphi_{0}(t,\mathbf{x}) e^{iH_{0}t}}_{\varphi(\mathbf{x})} e^{-iHt} 
 = V(t)^{\dagger}\varphi_{0}(t,\mathbf{x}) V(t)
\end{equation}
where the operator fields coincide at $t=0$. The idea is borrowed from classical mechanics: from 
looking at a particle system on a time slice one cannot infer whether its constituents interact. This is
only possible by watching how things change in the course of time, ie how the system \emph{evolves} in time. 
Some authors, especially in older textbooks like \cite{ItZu80}, replace $V(t)$ in (\ref{HeiInt}) by the
time-ordered exponential
\begin{equation}
 \U(t,-\infty) = \Ti e^{-i\int_{-\infty}^{t} d\tau \ H_{I}(\tau)}
\end{equation}
where $H_{I}(t)$ is the interacting part of the Hamiltonian in terms of the incoming free 
field $\varphi_{in}$ which takes the role of $\varphi_{0}$ in their treatment. 
This incoming free field then agrees with the Heisenberg field in the remote past $t \rightarrow -\infty$ and 
not, as in our case, on a time slice. 

Notwithstanding this detail, both formulations purport to employ a unitary map that relates the interacting 
Heisenberg field $\varphi$ to the free interaction picture field $\varphi_{0}$ such that (\ref{HeiInt}) 
holds for any time. 
Notice that Haag's theorem asks for less, namely that the unitary relation is given at one fixed instant.

To recall how the interaction picture states are defined and evolved in time, we consider the expectation 
value 
\begin{equation}
 \la \Psi | \varphi(t,\mathbf{x}) \Psi \ra = \la \Psi | V(t)^{\dagger}\varphi_{0}(t,\mathbf{x}) V(t) \Psi \ra 
 = \la V(t) \Psi | \varphi_{0}(t,\mathbf{x}) V(t) \Psi \ra,
\end{equation}
where $\Psi$ is a stationary reference Heisenberg state. This expression suggests that 
$\Psi(t) = V(t) \Psi$ is an interaction picture state at time $t$, evolved in this picture from $\Psi$.
A transition from one interaction picture state $\Psi(t)$ into another $\Psi(s)$ at time $s$, is thus 
governed by the evolution operator $\mathsf{U}(t,s)$ given by
\begin{equation}
\Psi(t) = V(t) \Psi = V(t) V(s)^{\dagger} \Psi(s) =: \mathsf{U}(t,s) \Psi(s) .
\end{equation}
The interaction picture state $\Psi(s)$ at time $s$ is thus time-evolved to time $t$ by the operator 
\begin{equation}
 \mathsf{U}(t,s) =  V(t) V(s)^{\dagger} = e^{iH_{0}t}  e^{-iH(t-s)} e^{-iH_{0}s}. 
\end{equation}
Let $\Omega_{0}$ be the vacuum of $H_{0}$, ie $H_{0}\Omega_{0} = 0$, $\{ \Psi_{n} \}$ an eigenbasis of the Hamiltonian $H$ and $E_{n}$ the 
corresponding eigenvalues, ie $H\Psi_{n} = E_{n} \Psi_{n}$. 
The identity operator $\id_{\Hi}$ is assumed to have a spectral decomposition which we write as 
$\id_{\Hi} = \sum_{n \geq 0}\mathsf{E}_{n}$, in which $\mathsf{E}_{n} = \la \Psi_{n} | \ \cdot \ \ra \Psi_{n}$
are the projectors of the presumed energy eigenbasis. $E_{0}=0$ is the ground state energy and $\Psi_{0}$
the vacuum of $H$.  

Now here is a crucial identity for the Gell-Mann-Low formula: the two vacua $\Psi_{0}$ and $\Omega_{0}$  
are mapped into each other by
\begin{equation}\label{vac}
 \mathsf{U}(t,\pm \infty)\Omega_{0} = c_{0} V(t)\Psi_{0},
\end{equation}
with $c_{0}:= \la \Psi_{0} | \Omega_{0} \ra$ being the overlap between the two vacua. This is made plausible
by considering the following computation
\begin{equation}\label{Rie}
 \mathsf{U}(t,s)\Omega_{0} =V(t)e^{iHs}\Omega_{0} = \sum_{n \geq 0} V(t) e^{iHs} \mathsf{E}_{n}\Omega_{0} 
 = c_{0} V(t) \Psi_{0} + \sum_{n \geq 1} V(t) e^{iE_{n}s} \la \Psi_{n} | \Omega_{0} \ra \Psi_{n}  ,
\end{equation}
where we have used $H_{0}\Omega_{0}=0$, slipped in the spectral decomposition of the identity operator and 
utilised $H \Psi_{0}=0$. The limit $s \rightarrow \pm \infty$ then forces the remainder of the 
sum to vanish on account of the Riemann-Lebesgue lemma from complex analysis\footnote{We prefer this 
argument to the one often used when letting $s\rightarrow \infty$ in $s (1+i\varepsilon)$ in (\ref{Rie}), cf.\cite{PeSch95} } 
(by 'analogy' because $E_{n} \neq 0$ for $n>0$). Then (\ref{vac}) follows.

\subsection{Gell-Mann-Low formula}
For the two-point function, this entails the following. First consider\footnote{The conventions of the axiomatic approach 
in the exposition of Haag's theorem in \cite{StreatWi00}, which we have also used, and the Gell-Mann-Low formalism differ slightly:
$V$ corresponds to $V(t)^{\dagger}$ and hence $V^{-1}$ to $V(t)$. We apologise for this notational inconvenience.}
\begin{equation}\label{GeMann}
\begin{split}
 \la \Psi_{0} | \varphi(x_{1})\varphi(x_{2}) \Psi_{0} \ra
 &= \la \Psi_{0} | V(t_{1})^{\dagger} \varphi_{0}(x_{1}) V(t_{1}) V(t_{2})^{\dagger} \varphi_{0}(x_{2})
 V(t_{2}) \Psi_{0} \ra \\
 &= \la V(t_{1}) \Psi_{0} | \varphi_{0}(x_{1}) V(t_{1}) V(t_{2})^{\dagger} \varphi_{0}(x_{2}) V(t_{2}) 
 \Psi_{0} \ra \\
 &= |c_{0}|^{-2} 
 \la \mathsf{U}(t_{1},+\infty) \Omega_{0} | \varphi_{0}(x_{1}) \mathsf{U}(t_{1},t_{2})\varphi_{0}(x_{2})
 \mathsf{U}(t_{2},-\infty) \Omega_{0} \ra \\
&=  |c_{0}|^{-2}  
\la  \Omega_{0} |\mathsf{U}(+\infty,t_{1}) \varphi_{0}(x_{1}) \mathsf{U}(t_{1},t_{2})\varphi_{0}(x_{2})
 \mathsf{U}(t_{2},-\infty) \Omega_{0} \ra .
\end{split}
\end{equation}
Next, note that no time-ordering is necessary so far and that the constant $c_{0}= \la \Psi_{0} | \Omega_{0} \ra$ can be expressed by using (\ref{vac}) 
and, applying the group law $\mathsf{U}(t,s)\mathsf{U}(s,t') = \mathsf{U}(t,t')$, we obtain:
\begin{equation}
|c_{0}|^{2} = \la \mathsf{U}(t,+\infty) \Omega_{0} | \mathsf{U}(t,-\infty) \Omega_{0} \ra
= \la  \Omega_{0} | \mathsf{U}(+\infty,-\infty) \Omega_{0} \ra 
= \la \Omega_{0} | \mathsf{S} \Omega_{0} \ra,
\end{equation}
where $\mathsf{S}:=\mathsf{U}(+\infty,-\infty)$ is the S-matrix in the interaction picture. For the next
step we are coerced to time-order the two field operators! Only once this is done can we piece together 
the S-matrix from the evolution operators in the last line of (\ref{GeMann}) to replace it by
the \emph{time-ordered product}, denoted by $\mathsf{T}\{ ... \}$, ie
\begin{equation}
\begin{split}
 \mathsf{U}(+\infty,t_{1}) \varphi_{0}(x_{1})  \mathsf{U}(t_{1},t_{2}) & \varphi_{0}(x_{2})  \mathsf{U}(t_{2},-\infty)   \\
 &= \mathsf{T}\{ \mathsf{U}(+\infty,t_{1})  \mathsf{U}(t_{1},t_{2}) \mathsf{U}(t_{2},-\infty)\varphi_{0}(x_{1}) \varphi_{0}(x_{2}) \}  \\
 &= \mathsf{T}\{ \mathsf{U}(+\infty,-\infty)\varphi_{0}(x_{1}) \varphi_{0}(x_{2}) \} 
 = \mathsf{T}\{ \mathsf{S} \ \varphi_{0}(x_{1}) \varphi_{0}(x_{2}) \} 
\end{split}
\end{equation}
and arrive at
\begin{equation}
 \la \Psi_{0} | \Ti \{ \varphi(x_{1})\varphi(x_{2}) \} \Psi_{0} \ra 
 = \frac{\la \Omega_{0} | \mathsf{T}\{ \mathsf{S} \ \varphi_{0}(x_{1}) \varphi_{0}(x_{2}) \}
 \Omega_{0} \ra}{\la \Omega_{0} | \mathsf{S}\Omega_{0} \ra}.
\end{equation}
For the $n$-point functions this is easily generalised to 
\begin{equation}\label{GeMaFo}
 \la \Psi_{0} | \mathsf{T}\{ \varphi(x_{1}) ...  \varphi(x_{n}) \} \Psi_{0} \ra 
 = \frac{\la \Omega_{0} |\mathsf{T}\{ \mathsf{S} \ \varphi_{0}(x_{1})  ... 
  \varphi_{0}(x_{n}) \} \Omega_{0} \ra}{\la \Omega_{0} | \mathsf{S}\Omega_{0} \ra} 
  \hs{0.4} (\mbox{Gell-Mann-Low formula})
\end{equation}
which finally is the Gell-Mann-Low formula.  

Haag's theorem directly controverts this formula or at least says $\mathsf{S}\Omega_{0} = \la \Omega_{0} | \mathsf{S} \Omega_{0} \ra \Omega_{0}$, ie 
that provided the above S-matrix really exists, then it must act trivially on the vacuum. 
As this is not acceptable, something must be wrong. In particular, the constant $c_{0}$ should vanish 
if the van Hove phenomenon occurs. Yet canonical perturbation theory depicts the probabilty of this 
'vacuum transition' as a divergent series of divergent integrals. The Feynman rules associate these integrals
with vacuum graphs such that  
\begin{equation}
|\la \Omega_{0} | \Psi_{0} \ra |^{2} = |c_{0}|^{2} = \la \Omega_{0} | \mathsf{S} \Omega_{0} \ra 
= \exp ( \mbox{ $\sum$ \texttt{vacuum graphs}}) .
\end{equation}
However, standard combinatorial arguments now claim that this problematic exponential is cancelled 
in (\ref{GeMaFo}) since it also appears and fortunately factors out in the numerator of the rhs. So no matter whether the van Hove
phenomenon occurs or not, it is irrelevant for the Gell-Mann-Low formula because 'van Hove cancels out'.

Notice that Haag's theorem does not know anything about which interacting Hamiltonian $H$ we choose and how its interaction part  
\begin{equation}
 H_{\mbox{\tiny int}} := H - H_{0},
\end{equation}
let alone its interaction picture representation $H_{\mbox{\tiny I}}(t) = e^{iH_{0}t} H_{\mbox{\tiny int}} e^{-iH_{0}t}$ 
is concretely constructed. As we have mentioned in two preceding chapters, Haag's theorem does not 
point out the ill-definedness of an interaction Hamiltonian like 
\begin{equation}\label{Hint}
 H_{I}(t) = \frac{g}{4!} \int d^{3}x \ \varphi_{0}(t,\mathbf{x})^{4}
 =: \int d^{3}x \ \Ha_{I}(t,\mathbf{x}) ,
\end{equation}
which is a monomial of the interaction picture field $\varphi_{0}(t,\mathbf{x})$.

Haag's theorem instead makes a very general statement, abstracting from the special form the Hamiltonian in a specific scalar theory.
All it says is this: any unitary transformation between a free and another 
sharp-time Wightman field must be such that all their vacuum expectation values agree. 

Now, the Gell-Mann-Low formula (\ref{GeMaFo}) asserts the contrary. The reason it does so is that it 
builds upon the wrong assumption that the interaction picture exists and that the interaction picture's
time evolution operator
\begin{equation}
\mathsf{U}(t,s) =  e^{iH_{0}t}  e^{-iH(t-s)} e^{-iH_{0}s} = \Ti e^{-i\int_{s}^{t}d\tau H_{I}(\tau)} 
\end{equation}
is well-defined.

\section{The CCR question}\label{sec:CCRQuest}                                                       
We will now briefly discuss the question whether the canonical formalism 
provides the tools to tackle the CCR question for an interacting field $\varphi$. Since the Gell-Mann-Low
formula (\ref{GeMaFo}) is designed to attain vacuum expectation values for the interacting Heisenberg picture
field $\varphi$ from those of the free interaction picture field $\varphi_{0}$, we may try and employ it.
So the question is: does the field $\varphi$ satisfy the CCR 
\begin{equation}
[\varphi(t,f),\varphi(t,g)]=0=[\dot{\varphi}(t,f),\dot{\varphi}(t,g)] , \hs{1} 
[\varphi(t,f),\dot{\varphi}(t,g)]=i(f,g)
\end{equation}
(in the spatially smoothed-out form) for all Schwartz functions $f,g$ in space at some time $t$? 
We have chosen $\pi(t,f)=\dot{\varphi}(t,f)$ for the conjugate momentum field which corresponds to  
Baumann's choice (see §\ref{sec:PowTHM}). 

Whatever the momentum field's form, we may assume that it involves the time derivative. This suffices to conlcude that
the Gell-Mann-Low formula is not apt to answer the CCR question. Nor can it be used to show that $\varphi$
is local. The reason is simply that \emph{time ordering} is indispensible for the Gell-Mann-Low
identity (\ref{GeMaFo}). 

First consider the case which can formally be treated, namely the first commutator of the CCR, 
\begin{equation}\label{GeMaFo1}
\begin{split}
\la \Psi_{0} | [\varphi(t,f),\varphi(t,g)] \Psi_{0} \ra 
&= \la \Psi_{0} | \Ti \{ [\varphi(t,f),\varphi(t,g)] \} \Psi_{0} \ra  \\
&= \frac{\la \Omega_{0} |\Ti \{ \mathsf{S} \ [\varphi_{0}(t,f),\varphi_{0}(t,g)] \}
 \Omega_{0} \ra}{\la \Omega_{0} | \mathsf{S}\Omega_{0} \ra} = 0,
 \end{split}
\end{equation}
because $[\varphi_{0}(t,f),\varphi_{0}(t,g)]=0$ for the free field $\varphi_{0}$ and time ordering
does not change anything in the first step. 
If $\Psi_{0}$ is cyclic for the field algebra of $\varphi$, one may argue that additionally inserting any number of 
already appropriately time-ordered field operators between the commutator and the two vacua on the lhs of 
(\ref{GeMaFo1}) does not change the fact that the corresponding rhs vanishes. This then entails $[\varphi(t,f),\varphi(t,g)]=0$, 
ie $\varphi$ exhibits a weak form of locality one might call \emph{time-slice locality}. 

To tackle the other commutators of the CCR, let us next consider
\begin{equation}\label{CCRGeMa}
 \frac{1}{\varepsilon}[\varphi(t,f),\varphi(t+\varepsilon,g)-\varphi(t,g)]
 = \frac{1}{\varepsilon}[\varphi(t,f),\varphi(t+\varepsilon,g)].
\end{equation}
Because this expression vanishes when time-ordered, whatever sign $\varepsilon \neq 0$ takes, we cannot apply
the Gell-Mann-Low formula as it relies on time ordering. In other words, even if we try weaker 
concepts of differentiation like left and right derivatives, ie taking the limits '$\varepsilon \uparrow 0$' 
or '$\varepsilon \downarrow 0$' instead of '$\varepsilon \rightarrow 0$', the time-ordering operator renders
all these attempts futile. Thus, the CCR question cannot be answered by the Gell-Mann-Low formula and consequently also not by perturbation theory
as we know it today.

\section{Divergencies of the interaction picture}\label{sec:DivInt}                                  
Because we encounter prolific divergences when the Gell-Mann-Low formula is expanded in perturbation theory, 
the contradiction between Haag's theorem and the Gell-Mann-Low formua is resolved. Either the interaction
picture is well-defined and trivial or must be ill-defined. In §\ref{sec:Fock} on Fock space it had already dawned on
us that the latter is the case, the divergences only confirm it. 

Before we review canonical renormalisation and see how it remedies the divergences in the next section, 
we remind ourselves in this section how they are incurred in the first place. 

\subsection{Divergences}\label{Diver} 
The problem of defining the interaction part of the Hamiltonian in (\ref{Hint}) appears on the agenda as 
soon as one attempts to put the Gell-Mann-Low identity (\ref{GeMaFo}) to use in perturbation theory. 
That is, when the perturbative expansion of the S-matrix, namely \emph{Dyson's series}
\begin{equation}\label{DysSeries}
 \mathsf{S} = 1 + \sum_{n\geq 1} \frac{(-i)^{n}}{n!} \int d^{4}x_{1} \ \dotsc \int d^{4}x_{n} \ 
 \mathsf{T}\{ \Ha_{I}(x_{1}) \ \dotsc \ \Ha_{I}(x_{n}) \}
\end{equation}
is employed (in four-dimensional spacetime).

We have already seen in the discussion of Theorem \ref{WightTHM}, ie Wightman's no-go theorem, that quantum fields are too singular to be defined at 
sharp spacetime points. Yet we have also seen in §\ref{sec:WightAx} that at least their $n$-point functions can 
be given a meaning in the sense of distribution theory (Wightman distributions). 
However, powers of a free field at one spacetime point are still ill-defined because their vacuum expectation values are divergent, eg
\begin{equation}\label{Power2}
\la \Omega_{0} | \varphi_{0}(x)^{2} \Omega_{0} \ra = \infty .
\end{equation}
The cure for this lies in defining so-called \emph{Wick powers}, given for a free field (!) recursively by
$:\varphi_{0}(x): = \varphi_{0}(x)$,
\begin{equation}\label{Wick1}
 :\varphi_{0}(x)^{2}: = \lim_{y \rightarrow x} \{ \varphi_{0}(x)\varphi_{0}(y) -
\la \Omega_{0} | \varphi_{0}(x)\varphi_{0}(y) \Omega_{0} \ra  \} 
\end{equation}
and 
\begin{equation}\label{Wick2}
 :\varphi_{0}(x)^{n}: = \lim_{y \rightarrow x} \{ :\varphi_{0}(x)^{n-1}: \varphi_{0}(y) -
 (n-1) \la \Omega_{0} | \varphi_{0}(x)\varphi_{0}(y) \Omega_{0} \ra  :\varphi_{0}(x)^{n-2}: \},   
\end{equation}
where the limit is to be understood in the weak sense, ie as a sesquilinear form on Hilbert space
\cite{Stro13}. In Euclidean field theories, Wick powers are defined mutatis mutandis in the obvious way, ie by the replacement 
$\la \Omega_{0} | \dotsc \Omega_{0} \ra \rightarrow \la \dotsc \ra_{0}$. 

However, for operator fields, Wick powers are equivalent to the well-known normal-ordered product in 
terms of annihilators and creators, namely, in terms of negative and positive frequency pieces,  
\begin{equation}
 : \varphi_{0}(x_{1}) \ \dotsc \ \varphi_{0}(x_{n}) : 
  = \sum_{J \subseteq \{1, \dotsc ,n\}} \prod_{j \in J} \varphi^{-}_{0}(x_{j}) 
 \prod_{i \notin J} \varphi^{+}_{0}(x_{i})  
\end{equation}
where the limit $x_{j} \rightarrow x$ is subsequently taken inside an expectation value.

Products of Wick-ordered monomials evaluate to a product of free two-point functions which are well-defined in the sense of distributions \cite{BruFK96}:
\begin{equation}\label{norord}
\la \Omega_{0} | :\varphi_{0}(x_{1})^{n_{1}}: \dotsc :\varphi_{0}(x_{k})^{n_{k}}: \Omega_{0} \ra 
= \sum_{G \in \G(n_{1}, \dotsc, n_{k})}  c(G) \prod_{l \in E(G)} \Delta_{+}(x_{s(l)}-x_{t(l)};m^{2}),
\end{equation}
in which the notation has the following meaning:
\begin{itemize}
 \item $\G(n_{1}, \dotsc, n_{k})$ is the set of all directed graphs without self-loops consisting of $k$ vertices with valencies $n_{1}, \dotsc, n_{k}$,
       respectively,
\end{itemize}
ie the $i$-th vertex, associated with the spacetime point $x_{i} \in \Mi$, has $n_{i}$ lines attached to it. 
\begin{itemize}
 \item $E(G)$ is the edge set of the graph $G$,
 \item $s(l)$ and $t(l)$ are source and target vertex of the line $l \in E(G)$.
\end{itemize}
The factor $c(G)$ is of purely combinatorial nature and is not of import to our discussion here (see \cite{BruFK96}).

Thus, one may alter the definition of the interaction picture Hamiltonian (\ref{Hint}) into the Wick-ordered form
\begin{equation}\label{Hint'}
  \Ha_{I}(x)=\frac{g}{4!} : \varphi_{0}(x)^{4}:
\end{equation}
to better the understanding of Dyson's series (\ref{DysSeries}). In terms of Feynman graphs, this means that self-loops evaluate to zero. 

Yet the Gell-Mann-Low formula (\ref{GeMaFo}) and Dyson's series (\ref{DysSeries}) require \emph{time-ordered 
vacuum expectation values}, ie we need the \emph{time-ordered} versions of the (\ref{norord}) which by 
virtue of Wick's theorem \cite{Wic50} evaluates formally to \cite{Fred10}
\begin{equation}\label{TimeNor}
\la \Omega_{0} | \mathsf{T}\{ :\varphi_{0}(x_{1})^{n_{1}}: \dotsc :\varphi_{0}(x_{k})^{n_{k}}: \} \Omega_{0} \ra 
= \sum_{G \in \G(n_{1}, \dotsc, n_{k})}  c(G) \prod_{l \in E(G)} i\Delta_{F}(x_{s(l)}-x_{t(l)};m^{2})
\end{equation}
and requires us to use time-ordered two-point functions, known as Feynman propagators (in position space):
\begin{equation}
i\Delta_{F}(x-y;m^{2}) := \la \Omega_{0} | \mathsf{T}\{ \varphi_{0}(x) \varphi_{0}(y) \} \Omega_{0} \ra ,
\end{equation}
given by the distribution  
\begin{equation}\label{FeynProp}
\Delta_{F}(x-y;m^{2}) = \lim_{\epsilon \downarrow 0}  
             \int \frac{d^{4}p}{(2\pi)^{4}} \frac{e^{-ip(x-y)}}{p^{2}-m^{2} + i \epsilon}.
\end{equation}
Note that the Feynman propagator has the property $\Delta_{F}(x-y;m^{2})=\Delta_{F}(y-x;m^{2})$, on 
account of the time-ordering. 
As is well known, products of these objects are in general ill-defined and are the origin of UV divergences
in perturbation theory \cite{He66}, in contrast to products of \emph{Wightman distributions} 
$\Delta_{+}(x-y;m^{2})= \la \Omega_{0} | \mathsf{T}\{ \varphi_{0}(x) \varphi_{0}(y) \} \Omega_{0} \ra$. 
Thus, the healing effect of Wick ordering has been reversed by the time-ordering. 

If we nevertheless insert Dyson's series (\ref{DysSeries}) into the Gell-Mann-Low formula (\ref{GeMaFo}) and 
use the interaction Hamiltonian (\ref{Hint'}), we get 
\begin{equation}\label{DysGeMa}
\begin{split}
& \la  \Psi_{0} |  \mathsf{T}\{ \varphi(x_{1}) \dotsc  \varphi(x_{n}) \} \Psi_{0} \ra  \\
&= \frac{1}{|c_{0}|^{2}} \sum_{l\geq 0} \frac{(-ig)^{l}}{(4!)^{l}l!} \int d^{4}y_{1} \ \dotsc \int d^{4}y_{l} \ 
 \la \Omega_{0} | \mathsf{T}\{ :\varphi_{0}(y_{1})^{4}: \ \dotsc :\varphi_{0}(y_{l})^{4}: 
\  \varphi_{0}(x_{1}) \dotsc  \varphi_{0}(x_{n})\} \Omega_{0} \ra 
 \end{split}
\end{equation}
which is a (formal) power series in the parameter $g$. It is ill-defined even if viewed as an asymptotic 
series: only its first few coefficients exist while the remainder consists of badly divergent integrals. 
In view of Haag's theorem, this is no surprise, though. We would, in fact, be confronted with a serious 
puzzle had we found a well-defined expression! Luckily, (\ref{DysGeMa}) is ill-defined.

\subsection{Regularisation}
Contrary to the commonly adopted view, the \emph{combinatorial approach} takes the following pragmatic stance. What (\ref{DysGeMa}) confronts 
us with, is an expression containing \emph{combinatorial data about a certain class of distributions} in the form of a 
formal power series. In this sense it is not meaningless. Let us simply write the series (\ref{DysGeMa}) as
\begin{equation}\label{formDysGe}
\la \Psi_{0} | \mathsf{T}\{ \varphi(x_{1}) \dotsc  \varphi(x_{n}) \} \Psi_{0} \ra 
= \sum_{G \in \G_{n}} g^{|V(G)|} \ \prod_{\ell \in L(G) } i\Delta_{F}(\ell) \prod_{\gamma \in C(G)} (\Mi_{\gamma},\nu_{\gamma}),
\end{equation}
where
\begin{itemize}
 \item $\G_{n}$ is the set of all scalar Feynman graphs, disconnected as well as connected, with $n$ external ends and vertices of the four-valent type, 
       ie '$\graph{0.1}{0}{0.1}{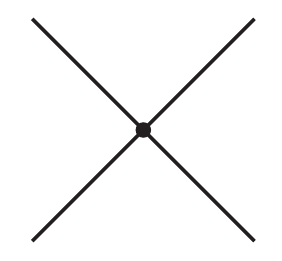}{0.1}$'. An example is the graph
\begin{equation}\label{G}
G = \graph{0.2}{-0.4}{0.3}{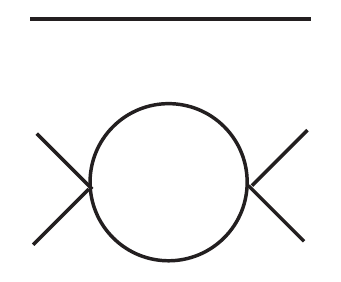}{0.5} 
\end{equation}
       which has $n=6$ external ends and $|V(G)|=2$ vertices.
 \item $L(G)$ is the set of connected pieces with no vertex, ie freely floating lines which connect two external points (the example graph $G$ has one 
       such line). 
 \item $C(G)$ is the set of all connected pieces contained in the graph $G$ with at least one vertex ($G$ has one such piece), 
 \item $V(G)$ is the vertex set of $G$ and $|V(G)|$ its cardinality.
\end{itemize} 
The symbol $i\Delta_{F}(\ell)$ is a shorthand for the Feynman propagator (\ref{FeynProp}) associated to the 
line $\ell \in L(G)$.  The pair $(\Mi_{\gamma},\nu_{\gamma})$, referred to as \emph{formal pair},
stands for the corresponding divergent integral as follows: the first component $\Mi_{\gamma} = \Mi^{|\gamma|}$ is the domain of 
integration\footnote{$|\gamma|$ denotes the loop number of the connected graph $\gamma \in C(G)$.}, 
while the second is the integrand written as a differential form. If the integral is convergent, we identify the formal pair
with the integral it represents and write
\begin{equation}\label{fop}
 (\Mi_{\gamma},\nu_{\gamma}) = \int_{\Mi_{\gamma}} \nu_{\gamma} =: \int \nu_{\gamma}. 
\end{equation}
For the sake of a neater and parsimonious notation, we suppress the dependence on the spacetime points $x_{1}, \dotsc, x_{n}$.  
In the case $n=2$ and $v=2$ we have, for example, the connected Feynman graph
\begin{equation}\label{sunset}
 \gamma = \graph{-0.1}{-0.55}{0.5}{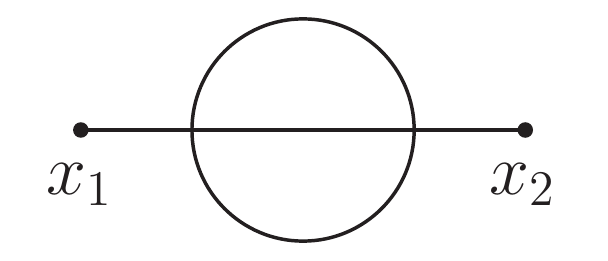}{0}
\end{equation}
with differential form
\begin{equation}\label{diform}
 \nu_{\gamma}(x_{1},x_{2},y_{1},y_{2}) = -\frac{1}{3!} 
 i\Delta_{F}(x_{1}-y_{1}) \left(i\Delta_{F}(y_{1}-y_{2})\right)^{3} i\Delta_{F}(y_{2}-x_{2}) 
 d^{4}y_{1} d^{4}y_{2} 
\end{equation}
and $\Mi_{\gamma}=\Mi^{2}$ for the two Minkowski integration variables $y_{1},y_{2} \in \Mi$. 
Clearly, the corresponding formal pair is not a convergent integral which cannot be given a meaning as
a distribution. 

In cases like this where a formal pair represents a divergent integral, one must \emph{regularise} it. 
This is done in various ways. All regularisation methods have in common that they alter the differential 
form\footnote{Dimensional regularisation is no exception. It fiddles with the dimension parameter in polar coordinates. }. 
Not all of them have a clear physical interpretation. If we take $h_{\varepsilon} \in \De(\Mi)$, ie a 
Schwartz function of compact support such that $h_{\varepsilon}(x)=0$ for all $x \in \Mi$ with Euclidean
length $||x|| < \varepsilon$, then 
\begin{equation}\label{regFP}
 \Delta^{\varepsilon}_{F}(x):= h_{\varepsilon}(x)\Delta_{F}(x)
\end{equation}
is a nicely behaving regularised Feynman propagator. Products of (\ref{regFP}) can be freely integrated.
This non-standard regularisation, which concerns us here only for the sake of the investigation, has 
eliminated two problems. First, by letting $h_{\varepsilon}$ have compact 
support, we stave off infrared divergences inflicted by the infinite volume of spacetime $\Mi$. Second, 
because $h_{\varepsilon}$ vanishes on a neighbourhood of the origin, we are also save from ultraviolet (UV)
singularities, ie short-distance singularities. 

The regularised version of the differential form in (\ref{diform}), with all Feynman propagators replaced by regularised ones, 
can now be construed as the distribution
\begin{equation}
 f \mapsto \int \nu^{\varepsilon}_{\gamma}(f):= \int d^{4}x_{1} \int d^{4}x_{2} \int 
 \nu^{\varepsilon}_{\gamma}(x_{1},x_{2},y_{1},y_{2})f(x_{1},x_{2})
\end{equation}
for $f \in \Sw(\Mi^{2})$. 

When all formal pairs in (\ref{formDysGe}) are regularised, now denoted by $(\Mi_{\gamma}, \nu^{\varepsilon}_{\gamma}) = \int \nu^{\varepsilon}_{\gamma}$, 
one obtains an asymptotic series in the coupling $g$ with coefficients representing distributions,  
\begin{equation}\label{formDysGeReg}
\la \Psi_{0} | \mathsf{T}\{ \varphi(x_{1}) \dotsc  \varphi(x_{n}) \} \Psi_{0} \ra_{\varepsilon} 
:= \sum_{G \in \G_{n}} g^{|V(G)|} \prod_{\ell \in L(G) } i\Delta^{\varepsilon}_{F}(\ell) \prod_{\gamma \in C(G)} \int \nu^{\varepsilon}_{\gamma},
\end{equation}
to be applied to a test function $f \in \Sw(\Mi^{n})$. Once that is done, one arrives at a formal power series with complex numbers as coefficients.  

Suppose we had a suitable resummation scheme for this series, then the result may enable us to define the rhs of (\ref{formDysGeReg}) as a distribution.
But its dependence on the regularisation function $h_{\varepsilon}$ is unacceptable, not least because Poincaré invariance is violated. 
To get rid of this dependence, the \emph{adiabatic limit} $\lim_{\varepsilon \rightarrow 0}h_{\varepsilon}=1$ is necessary, a condition that we additionally impose on $h_{\varepsilon}$.
Of course, since this restores the unfavourable original situation of divergent integrals, 
one has to modify the formal pairs in such a way that their limits lead to convergent and, moreover, 
Poincaré-invariant integrals. 

\subsection{Evade Haag's theorem by regularisation?} To summarise, we note that 
\begin{itemize}
 \item time-ordering, necessitated by the Gell-Mann-Low formula (\ref{GeMaFo}), leads
       inevitably to ill-defined products of Feynman propagators which then in turn bring about UV divergences;
 \item although regularisation helps, it is physically unacceptable.
\end{itemize}
Let us imagine for a moment we had chosen a Poincaré-invariant regularisation method and had found an explanation for 
why it is physically acceptable and satisfactory. Suppose further that the so-obtained two-point function
differs from the two-point function of the free field at spacelike distances. 
Then, contrary to what some might believe and wish for, by Haag's theorem (Theorem \ref{Hath}), we can be sure that the so-reconstructed theory is unitarily 
inequivalent to the free theory. Invoking the Stone-von Neumann theorem would be futile: something has to give, the provisos of both theorems
cannot form a coherent package!

\section{The renormalisation narrative}\label{sec:Ren}                                               
We shall review in this section the way renormalisation is nowadays canonically introduced and how it 
changes the Gell-Mann-Low perturbation expansion so drastically that the formal power series one obtains 
has finite coefficients. This outcome, however, brings back the conundrum posed by Haag's theorem because 
the same bold assertions about the interaction picture and the unitarity of its evolution operator are 
made yet again, albeit this time for the renormalised field. 

\subsection{Counterterms}
Now because one cannot accept the regularised theory as the final answer and removing the regulator brings back the
divergences, the canonical formalism backpedals at this point. To explain the necessary modifications, 
\emph{the story is changed} in a decisive way: the coupling $g$ is just the 'bare coupling', employed so 
far out of ignorance (in a sort of bare state of mind, one might say). 
The same holds for the \emph{bare mass} $m$ and the \emph{bare field} $\varphi$. These 'bare' quantities are deemed unphysical
because they have evidently led to divergences. 

Dyson explains this situation in \cite{Dys49b} by telling the amusing tale of an ideal observer
whose measuring apparatus, 'non-atomic' and therefore not comprised of atoms, is only limited by the 
fundamental constants $c$ and $\hbar$. Performing measurements at spacetime points which the fictitious 
observer is capable to determine with infinite precision, he finds infinite results. 

However, the physical ('renormalised') counterparts of the bare quantities
are constructed as follows. First, the bare field gets 'renormalised' by a factor $Z$:
\begin{equation}\label{wr}
\varphi_{r} := \frac{\varphi}{\sqrt{Z}},
\end{equation}
where the resulting field $\varphi_{r}$ is called \emph{renormalised field}, the new player that takes the place of the
'old', the bare field $\varphi$.
The so-called \emph{wavefunction renormalisation} or \emph{field-strength renormalisation constant} $Z$ is a function of several variables, 
in particular of the renormalised coupling $g_{r}$. Both renormalised coupling $g_{r}$ and mass $m_{r}$ are defined through
\begin{equation}\label{gmr}
g= g_{r} \frac{Z_{g}}{Z^{2}} \ , \hs{1}  m = m_{r} \sqrt{\frac{Z_{m}}{Z}} 
\end{equation}
in which two additional renormalisation constants are introduced: $Z_{g}$ is the 
\emph{coupling renormalisation constant} and $Z_{m}$ is the \emph{mass renormalisation constant}. 
Both are also functions of the renormalised coupling $g_{r}$. 

We shall now see that when the bare quantities are replaced by their physical, renormalised counterparts, 
\emph{the net effect is a modified interaction part} of the Hamiltonian and thus of the corresponding 
Lagrangian. Let us consider the original (bare) Lagrangian, given by
\begin{equation}
\La(\varphi) = \frac{1}{2}(\partial \varphi)^{2} - \frac{1}{2}m^{2}\varphi^{2} - \frac{g}{4!}\varphi^{4},
\end{equation}
formulated in terms of the bare quantities. In terms of the renormalised quantities, this same Lagrangian 
takes the form
\begin{equation}\label{Lren}
\La(\varphi) = \frac{1}{2}Z(\partial \varphi_{r})^{2} - \frac{1}{2}m_{r}^{2}Z_{m}\varphi_{r}^{2} - \frac{g_{r}}{4!}Z_{g}\varphi_{r}^{4} ,
\end{equation}
where only $\varphi$, $m$ and $g$ have been replaced in accordance with (\ref{wr}) and (\ref{gmr}). 
The next step is now to split this expression into two pieces: $\La_{r}$ and what is known as 
the \emph{counterterm} $\La_{ct}$, that is,
\begin{equation}\label{renLa}
\La = \La_{r} + \La_{ct}
\end{equation}
whose components are given by
$\La_{r} = \frac{1}{2}(\partial \varphi_{r})^{2} - \frac{1}{2}m_{r}^{2}\varphi_{r}^{2} 
- \frac{g_{r}}{4!}\varphi_{r}^{4}$ and the counterterm Lagrangian
\begin{equation}\label{ctLa}
\La_{ct} = \frac{1}{2}(Z-1)(\partial \varphi_{r})^{2} - \frac{1}{2}m_{r}^{2}(Z_{m}-1)\varphi_{r}^{2} 
- \frac{g_{r}}{4!}(Z_{g}-1)\varphi_{r}^{4}.
\end{equation}
The index '$r$' in $\La_{r}$ signifies that this part is composed of the renormalised quantities only. 
Itzykson and Zuber call $\La$ the 'renormalised Lagrangian' \cite{ItZu80}. As this bears some potential for confusion 
because $\La$ is equal to the original Lagrangian, Itzykson and Zuber admit that this is a very unfortunate 
denomination (ibidem, p.389). 

Unfortunately, this new splitting (\ref{renLa}) of the old Lagrangian marks a handwaving twist in the 
renormalisation narrative: $\varphi_{r}$ is now seen as the proper fully interacting Heisenberg picture 
field and is subsequently put through the same interaction picture transformation procedure as described in 
§\ref{sec:GeMaLo} for the bare field $\varphi$. The resulting interaction picture field $\varphi_{r,0}$ 
is again a free field, this time with mass $m_{r} \neq m$, but the obvious relation (\ref{wr}) to the old
free field brushed under the carpet. As already alluded to in §\ref{sec:whatodo}, differing 
masses imply unitary inequivalence. We shall come back to this point.

The new overhauled narrative starts out with the Lagrangian $\La=\La_{0,r}+\La_{int}$ 
consisting of new free and interacting parts 
\begin{equation}
 \La = \underbrace{\frac{1}{2}(\partial \varphi_{r})^{2} - \frac{1}{2}m_{r}^{2}\varphi_{r}^{2}}_{\La_{0,r}}
             \underbrace{- \frac{g_{r}}{4!}\varphi_{r}^{4} + \La_{ct}}_{\La_{int}} .
\end{equation}
The interaction term $\La_{int}$ is now subjected to the same interaction picture procedure as described 
in §\ref{sec:GeMaLo} and $m_{r}$ is the mass of the free interaction picture field. 
But the crucial difference is that the terms in $\La_{ct}$ have a nontrivial
coupling dependence: $Z,Z_{m}$ and $Z_{g}$ are themselves seen as functions of the new renormalised 
coupling $g_{r}$ and need to be expanded (in perturbation theory).

\subsection{Renormalised Gell-Mann-Low expansion}
As a result, the Gell-Mann-Low expansion is no more an exponential one in the renormalised coupling $g_{r}$. 
At every order of perturbation theory, the counterterm Lagrangian $\La_{ct}$ generates additional divergent integrals (or formal pairs) 
which, provided the coefficients of the renormalisation constants are chosen correctly, cancel the 
divergent integrals that the term
\begin{equation}
\La_{int,r} =  - \frac{g_{r}}{4!}\varphi_{r}^{4}
\end{equation}
produces (when transformed into the interaction picture representation).  The choice of $Z,Z_{m}$ and
$Z_{g}$ is not unique and needs physical conditions to be fixed\footnote{The underlying reason can be 
illustrated by the fact that a divergent integral $\int I$ is cancelled by $J_{C}:=C-\int I$ for any constant $C$. 
It is now a physical choice to specify what $C$ should sensibly be to make sense of $\int I + J_{C}$.}. 

Let the expansions of the renormalisation constants be given by
\begin{equation}\label{pertZ}
 Z(g_{r}) = 1 + \Sum_{j \geq 1} a_{j} g_{r}^{j} \ , \hs{1} 
 Z_{m}(g_{r}) = 1 + \Sum_{j \geq 1} b_{j} g_{r}^{j} \ , \hs{1}
 Z_{g}(g_{r}) = 1 + \Sum_{j \geq 1} c_{j} g_{r}^{j}.
\end{equation}
The condition $Z(0) = Z_{m}(0) = Z_{g}(0)=1$ must be imposed to make sure one obtains the free Lagrangian
when setting $g_{r}=0$. Let us briefly review the canonical 'song and dance' to see how the coefficients 
of these series are determined and how they may be interpreted. 

One starts by constructing the new renormalised interaction picture Hamiltonian for Dyson's matrix from 
$\La_{int}=\La_{int,r}+\La_{ct}$ and gets
\begin{equation}\label{renH}
\Ha^{r}_{I}(x) = \frac{g_{r}}{4!}Z_{g}(g_{r})\varphi_{r,0}(x)^{4} 
+ \frac{1}{2}[Z(g_{r})-1](\partial \varphi_{r,0}(x))^{2} 
+ \frac{1}{2}m_{r}^{2}[Z_{m}(g_{r})-1]\varphi_{r,0}(x)^{2}
\end{equation}
which is formulated in terms of free interaction picture fields $\varphi_{r,0}$. 
We may actually interprete this Hamiltonian physically. 

The first term describes interactions between particles and cures some of the divergences 
incurred by vertex corrections, whereas the additional two terms take into account that real-world physical
and relativistic interactions change the mass, ie the 'energy-momentum complex' of the system. 
Technically, their task is to cancel the remainder of the divergences that the first term generates. 
Yet they do not just cancel divergences, but \emph{bring about a coupling-dependent mass shift}. 
This almost surely destroys unitary equivalence between the fields $\varphi_{r}$ and $\varphi_{r,0}$,
as we shall discuss in §\ref{sec:RCircH}. 

The second step of the canonical procedure is to take the Gell-Mann-Low formula for the renormalised field, ie
\begin{equation}\label{GeMaFoR}
 \la \Psi_{0} | \mathsf{T}\{ \varphi_{r}(x_{1}) \dotsc  \varphi_{r}(x_{n}) \} \Psi_{0} \ra 
 = \frac{\la \Omega_{0} |\mathsf{T}\{ \mathsf{S}_{r} \ \varphi_{r,0}(x_{1})  \dotsc 
  \varphi_{r,0}(x_{n}) \} \Omega_{0} \ra}{\la \Omega_{0} | \mathsf{S}_{r} \Omega_{0} \ra},
\end{equation}
and then trade $\mathsf{S}_{r}$ for Dyson's series   
\begin{equation}\label{DysSeriesR}
 \mathsf{S}_{r} = 1 + \sum_{n\geq 1} \frac{(-i)^{n}}{n!} \int d^{4}x_{1} \ \dotsc \int d^{4}x_{n} 
 \mathsf{T}\{ \Ha^{r}_{I}(x_{1}) \ \dotsc \ \Ha^{r}_{I}(x_{n}) \}
\end{equation}
in the renormalised form.

Notice that this series is obtained iteratively in a way that is \emph{independent} of the coupling
parameter. Because the coupling dependence of the Hamiltonian has changed now so dramatically, we no
longer arrive directly at a perturbative expansion by inserting Dyson's series, as mentioned above:
to get there, we have to use the perturbation series of the $Z$ factors in (\ref{pertZ}). 

By taking all these perturbative series into account and regularising the resulting Feynman propagators, 
one arrives at the renormalised analogue of (\ref{formDysGeReg}). Its combinatorial content differs 
substantially, because the renormalised theory has two extra classes of vertex types, associated to the
counterterms of the Lagrangian, namely the counterterm vertices
\begin{equation}\label{ctvert}
 \graph{-0.1}{-0.5}{0.5}{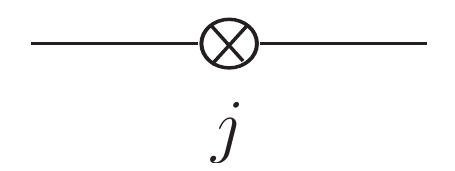}{0}  \hs{1} \mbox{and} \hs{1} \graph{-0.1}{-0.55}{0.5}{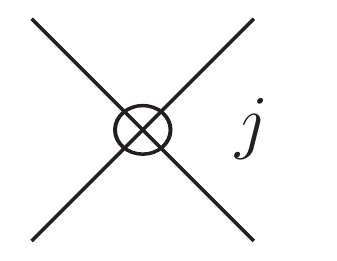}{2} j=1,2,3,\dotsc,
\end{equation}
ie an infinite number of different vertices! The number $j$ signifies the power of $g_{r}$ this graph is 
associated with. This means in particular that one counterterm vertex of order $j$ counts as if it was a 
graph with $j$ vertices. These again codify distributions built from integrals, 
\begin{equation}
 \graph{-0.1}{-0.55}{0.5}{2ver}{0} = \ i g_{r}^{j} \int \frac{d^{4}p}{(2\pi)^{4}} [ a_{j} p^{2} - m_{r}^{2} b_{j} ] e^{-ip \cdot x}
 \hs{1} \mbox{(configuration space)} 
\end{equation}
and factors 
\begin{equation}
 \graph{-0.1}{-0.55}{0.5}{R4ver}{0} = \ -i g_{r}^{j} c_{j}  \hs{1} \mbox{(configuration space)} .
\end{equation}
On the grounds that in momentum space, the connected pieces of graphs decompose into 1PI pieces, it is
now more convenient, however, to pass over to momentum space by taking the Fourier transform, ie
\begin{equation}
 \wt{G}^{(n)}_{r}(p_{1}, \dotsc, p_{n};\varepsilon) := \int dx_{1} \  e^{ip_{1} \cdot x_{1}} \dotsc \int dx_{n} \ e^{ip_{n} \cdot x_{n}} 
 \la \Psi_{0} | \mathsf{T}\{  \varphi_{r}(x_{1}) \dotsc  \varphi_{r}(x_{n}) \} \Psi_{0} \ra_{\varepsilon} 
\end{equation}
and one obtains
\begin{equation}\label{renG}
\wt{G}^{(n)}_{r}(p_{1}, \dotsc, p_{n};\varepsilon)
= (2\pi)^{4} \delta^{(4)}(p_{1}+ \dotsc + p_{n}) \sum_{G \in \G^{r}_{n}} g_{r}^{|V(G)|}  \prod_{\ell \in L(G) } i\wt{\Delta}^{\varepsilon}_{F}(\ell)
\prod_{\gamma \in P(G)} \int \omega^{\varepsilon}_{\gamma},
\end{equation}
in which 
\begin{itemize}
 \item $\G^{r}_{n}$ is the set of all graphs with four-valent vertices and $n$ external legs but this time
       of the \emph{renormalised} theory, ie including vertices of the class (\ref{ctvert}),
 \item $P(G)$ is the set of all 1PI pieces of the graph $G$ ('$P$' for proper graphs) 
\end{itemize}
and the set of freely floating lines $L(G)$ is now enriched by external leg (free) propagators.

The new product over all formal pairs $(\Mi_{\gamma},\omega^{\varepsilon}_{\gamma})= \int \omega^{\varepsilon}_{\gamma}$ 
contains the coefficients of the renormalisation $Z$ factors which can be adjusted in such a way 
that the adiabatic limit $\varepsilon \rightarrow 0$ can now be taken without harm in the sense that 
the result is a formal power series with finite (momentum dependent) coefficients. 

The rhs of (\ref{renG}) is a formal power series with coefficients in the set $\Sw(\Mi^{n})'$, the set of 
tempered distributions for which the adiabatic limit does no harm, diagrammatically, we write this as 
\begin{equation}\label{regLi}
\lim_{\varepsilon \rightarrow 0} \wt{G}^{(n)}_{r}(p_{1}, \dotsc, p_{n};\varepsilon) \hs{0.2} = \graph{0.1}{-0.65}{0.25}{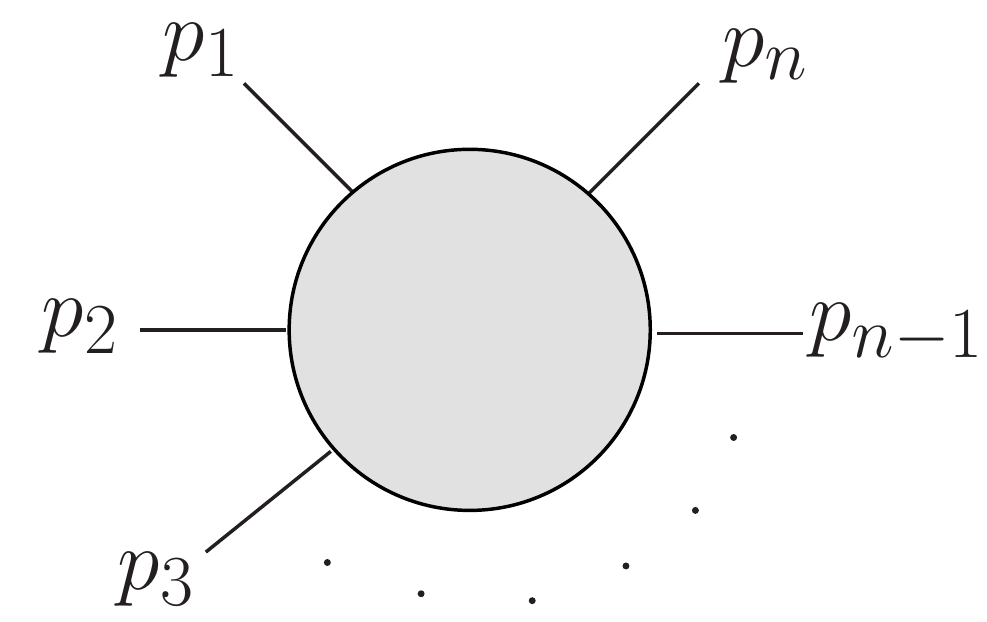}{0.3}  .
\end{equation}
Let us assume for simplicity that renormalised $\varphi^{4}$ theory is in some sense Borel-summable, ie that the distributions
\begin{equation}
\wt{G}^{(n)}_{r,v}(f;\varepsilon) := \sum_{G:|V(G)|=v} 
(\prod_{\ell \in L(G) } i\Delta^{\varepsilon}_{F}(\ell) \prod_{\gamma \in C(G)} \int \nu^{\varepsilon}_{\gamma})(f)
\end{equation}
of each vertex level $v$ collectively give rise to formal power series of the form  
\begin{equation}
\sum_{v \geq 0}  \lim_{\varepsilon \rightarrow 0} \wt{G}^{(n)}_{r,v}(f;\varepsilon) g_{r}^{v}
\end{equation}
whose Borel sum 
\begin{equation}\label{Bordist}
\wt{G}^{(n)}_{r}(f,g_{r}) = \mathcal{S}[ 
\mbox{$\sum_{v \geq 0}$}  \lim_{\varepsilon \rightarrow 0} \wt{G}^{(n)}_{r,v}(f;\varepsilon) g_{r}^{v} ] 
\end{equation}
with Borel summation operator $\mathcal{S}$ really represents a distribution defined for the test function $f$. Then, the time-ordered $n$-point 
function in configuration space  
\begin{equation}
 G^{(n)}_{r}(x_{1}, \dotsc, x_{n},g_{r}) := \la \Psi_{0} | \mathsf{T}\{ \varphi_{r}(x_{1}) \dotsc  \varphi_{r}(x_{n}) \} \Psi_{0} \ra 
\end{equation}
exists and is given by $G^{(n)}_{r}(f,g_{r}):=\wt{G}^{(n)}_{r}(\wt{f},g_{r})$ for a test function $f$ and its Fourier transform $\wt{f}$, both elements in 
$\Sw(\Mi^{n})$. 

Although this would certainly be a neat result, it brings back the inconvenient question raised by 
Haag's theorem. According to the canonical narrative, the renormalised free interaction picture field 
$\varphi_{r,0}$ is unitarily related to the fully interacting renormalised field $\varphi_{r}$ by
the intertwining relation 
\begin{equation}
 \varphi_{r}(t,\mathbf{x}) = \Ti \{ e^{i\int_{0}^{t}  \Ha^{r}_{I}} \} \varphi_{r,0}(t,\mathbf{x}) 
        \Ti \{ e^{-i\int_{0}^{t}  \Ha^{r}_{I}} \}
\end{equation}
in which 
$\int_{0}^{t} \Ha^{r}_{I} := \int_{0}^{t} dy_{0} \int d^{3}y \ \Ha^{r}_{I}(y)$. 
Because the above renormalisation procedure yields finite results, one may argue that this time, 
the interaction picture has done its job properly and the assertion about unitary equivalence is not 
contradicted by divergences because there are none.

Of course, there \emph{are} divergences. Notwithstanding that one may argue that the divergent terms 
cancel each other, Dyson's series (\ref{DysSeriesR}) is still not well-defined as the coefficients 
of the renormalisation factors are divergent integrals themselves. In other words, the formalism is 
sufficiently dubious such that a finite outcome neither proves nor disproves anything.  

\subsection{Counterterms describe interactions} 
The renormalised theory with interaction term 
\begin{equation}\label{LInt}
\La_{int} = - \frac{g_{r}}{4!}Z_{g}(g_{r})\varphi_{r}^{4} - \frac{1}{2}[Z(g_{r})-1](\partial \varphi_{r})^{2} 
- \frac{1}{2}m_{r}^{2}[Z_{m}(g_{r})-1]\varphi_{r}^{2}.
\end{equation}
can therefore not be the final answer. Even though, as explained, the additional counterterms do not seem 
entirely unphysical, what precludes this Lagrangian description from being a fully satisfactory theory 
is the fact that the coupling-dependent $Z$ factors in (\ref{LInt}) cannot be chosen finite. 

Unfortunately, the hackneyed phrase 'absorption of infinities into couplings and masses' is not just empty 
but explains nothing physically.
Notice that the two mass counterterms by themselves would not produce divergences if their coefficients were 
finite (see §\ref{sec:RCircH}). It is only when they are combined with the vertex interaction term 
that divergent graphs arise. 

As mentioned, these additional terms do more than merely 'counter' and thereby cure the divergences. They had better be seen as some kind of 
\emph{auxiliary interaction terms which partially capture relativistic quantum interactions and compensate for the ills incurred by the 'wrong choice'
of Lagrangian}. Our motivation for this interpretation is as follows.  

When relativistic quantum particles interact, they change their mode of existence such that during 
interactions, the particle concept breaks down completely. Because energy, mass and momentum are
intimately related and can only be disentangled for free particles, the initial unrenormalised guess 
\begin{equation}\label{unreguess}
\Ha_{I}(x) = \frac{g}{4!}\varphi_{0}(x)^{4} 
\end{equation}
did not capture the complexity of the relativistic situation. When new particles are created, the mass
of the system changes depending on the coupling strength. Consequently, (\ref{unreguess}) cannot be 
sufficient. This, we speculate, may actually be the physical reason behind why Fr\"ohlich has found 
Euclidean $(\varphi^{4})_{d}$-theories to be trivial for $d \geq 4 + \epsilon$ \cite{Fro82}. 

In other words, despite their auxiliary status, we contend that the two additional terms
\begin{equation}
 \frac{1}{2}[Z(g_{r})-1](\partial \varphi_{r,0}(x))^{2} 
 + \frac{1}{2}m_{r}^{2}[Z_{m}(g_{r})-1]\varphi_{r,0}(x)^{2}
\end{equation}
take into account that relativistic interactions create additional momentum and mass, one at the expensive
of the other in a way that depends on the coupling strength. 

In a sense, this is an interpretation of the \emph{self-energy} which describes  
a mass shift depending on the coupling and momenta.
Our interpretation of the counterterms are motivated by the way we think about this very mass shift:
to us, it captures relativistic interactions, what else could it possibly say?

Imagine for a moment we had found a Lagrangian not leading to divergences. We would still have to make sure that certain
conditions required by physical considerations are satisfied. In the presently known renormalised theories,
these take the form of renormalisation conditions.
Therefore, we would expect such terms even in a relativistic theory without divergences. 

If a mathematically sound Lagrangian quantum field theory is ever possible and one day we succeed in finding the proper 
interaction term that opens up a viable path to a well-defined perturbative expansion while circumventing
and staying clear of divergences all along, we can be sure that the resulting theory is \emph{not} unitary 
equivalent to a free theory, Haag's theorem is very clear about this. 

The underlying reason why the interaction term (\ref{LInt}) has been found is combinatorial in nature. 
In fact, \emph{there are sound mathematical structures behind renormalisation, namely those of a Hopf 
algebra}. This was discovered by Kreimer in the late 1990s 
and further developed in collaboration with Connes and Broadhurst \cite{Krei02,CoKrei98,BroK01}. 

The feasibilty of the above described renormalisation procedure, proved by Bogoliubov, Parasuik, Hepp and Zimmermann (see \cite{He66} and 
references therein) culminated eventually in what is known as \emph{Zimmermann's forest formula} which solves
\emph{Bogoliubov's recursion formula}. Connes and Kreimer later showed that the underlying combinatorics of this latter formula is of 
Hopf-algebraic nature \cite{CoKrei98,CoKrei00}.

\section{Renormalisation circumvents Haag's theorem}\label{sec:RCircH}                               
On the assumption that a scalar QFT's $n$-point distributions can be defined via the limit (\ref{Bordist}), one can construe the symbolic expression
\begin{equation}\label{fint}
 \varphi_{r}(t,\mathbf{x}) = \Ti \{ e^{i\int_{0}^{t} \Ha^{r}_{I}} \} 
                               \varphi_{r,0}(t,\mathbf{x}) 
        \Ti \{ e^{-i\int_{0}^{t} \Ha^{r}_{I}} \}
\end{equation}
as a way to denote the action of the field intertwiner 
$V_{r}(t)=\Ti \{ e^{-i\int_{0}^{t} \Ha^{r}_{I}} \}$, characterised by the schematic diagram
\begin{equation}\label{Proce}
 \varphi_{r,0}(x) \hs{0.5} \stackrel{1}{\longrightarrow} \hs{0.5} 
 \mathsf{S}_{r} = \mathsf{T} e^{-i \int \Ha^{r}_{I}} \hs{0.5} \stackrel{2}{\longrightarrow} \hs{0.5} 
 \wt{G}^{(n)}_{r} \stackrel{3}{\longrightarrow} \hs{0.5} \varphi_{r}(x),
\end{equation}
which is to be read as the following procedure:
\begin{enumerate}
 \item[Step 1:] formal construction of Dyson's series $\mathsf{S}_{r}$ from the renormalised 
           interaction picture Hamiltonian $\Ha_{I}^{r}(x)$ of the free field $\varphi_{r,0}(x)$; 
 \item[Step 2:] Gell-Mann-Low expansion, regulator limit at each order of perturbation theory and (some form of) Borel 
           summation which leads to the definition of $n$-point distributions;
 \item[Step 3:] reconstruction of the renormalised scalar field theory from the attained $n$-point distributions 
          by using Wightman's reconstruction theorem.
\end{enumerate}
Because the only provision of Haag's theorem we are not reluctant to give up is \emph{unitary equivalence}, 
we are inclined strongly to believe that the intertwiner $V_{r}(t)$ cannot be unitary and that this 
is precisely the reason why
\begin{equation*}
 \mbox{\textsc{renormalised theories are not affected by Haag's triviality theorem}}.
\end{equation*} 
The whole canonical narrative just created a misunderstanding by more or less naively nurturing the 
(certainly not entirely unfounded) belief that one could construct an interacting theory from a free field theory through a unitary 
intertwining operator $V_{r}(t)$. 

We shall in the following present an argument which makes this unprovable contention 
\emph{plausible beyond doubt}. The reason it cannot be proved lies in the mathematical ill-definedness 
of (\ref{fint}) and that we simply do not know whether the procedure (\ref{Proce}) is feasible. 

\subsection{Mass shift destroys unitary equivalence}\label{massdest} 
We consider a simple toy model given by the Lagrangian
\begin{equation}\label{TM}
 \La_{m}  = \frac{1}{2}(\partial \varphi)^{2} - \frac{1}{2}m_{0}^{2}\varphi^{2} 
        - \frac{1}{2} \delta m^{2} \varphi^{2} 
\end{equation}
with mass shift parameter $\delta m^{2} >0$. This is the example from Duncan's monograph \cite{Dunc21} we have mentioned in
§\ref{monog}. It is clear that the Lagrangian (\ref{TM}) describes a free field of mass $m>0$, with $m^{2}=m_{0}^{2} + \delta m^{2}$.

We demonstrate now that if we treat the last term of this Lagrangian as an interaction term and subject it 
to the canonical procedure of perturbation theory in the interaction picture, the 'mass 
interaction' term 
\begin{equation}\label{TMHam}
 \Ha^{m}_{I}(x) := \frac{1}{2} \delta m^{2} : \varphi_{0}(x)^{2} :,
\end{equation}
formally obtained by means of the intertwiner
\begin{equation}\label{ms}
V_{m}(t) = \Ti \{ e^{-i\int_{0}^{t}  \Ha^{m}_{I}} \} 
\end{equation}
will then enable us to compute the two-point function through the Gell-Mann-Low expansion
\begin{equation}\label{TMpert}
 \la \Omega |\mathsf{T} \{ \varphi(x) \varphi(y) \} \Omega \ra  = \frac{1}{|c_{0}|^{2}} 
 \sum_{n \geq 0} \frac{(-i)^{n}}{n!} \int_{\Mi^{n}} dz \ 
 \la \Omega_{0} | \mathsf{T}\{ \varphi_{0}(x) \varphi_{0}(y) \Ha^{m}_{I}(z_{1}) ... \Ha^{m}_{I}(z_{n}) 
 \Omega_{0} \ra  .
\end{equation}
We shall now see that the unitary-looking intertwining operator (\ref{ms}) is not unitary even though 
it does not at all lead to divergencies!  

This then makes it almost certain that renormalised $\varphi^{4}$ theory described by (\ref{Proce})
is unitarily inequivalent to its corresponding free interaction picture field $\varphi_{r,0}$. 
Since unitary equivalence is a key provision of Haag's theorem, the conclusion is that as a 
consequence, this theory is hence not afflicted by Haag's theorem. 

First, we consider a theorem which we dub 'Haag's theorem for free fields'. This theorem can be found
in \cite{ReSi75}, p.233 (Theorem X.46). We have simplified it a bit by omitting the conjugate momentum field from the 
description which makes for a somewhat more straightforward exposition.
Note that the dimension of spacetime is of no relevance. 

\begin{Theorem}[Haag's theorem for free fields]\label{freeH}
Let $\varphi$ and $\varphi_{0}$ be two free fields of masses $m$ and $m_{0}$, respectively. If at 
time $t$ there is a unitary intertwiner $V$ such that 
\begin{equation}
\varphi(t,\mathbf{x}) = V \varphi_{0}(t,\mathbf{x})V^{-1} ,
\end{equation}
then $m=m_{0}$, ie if $m \neq m_{0}$ then there exists no such unitary intertwiner.
\end{Theorem}
\begin{proof}
The proof is a simple version of the proof of Haag's theorem, Theorem \ref{Hath} (see the arguments there). 
The conclusion
\begin{equation}
 \Delta_{+}(0,\mathbf{x}-\mathbf{y};m^{2}) = \Delta_{+}(0,\mathbf{x}-\mathbf{y};m_{0}^{2})
\end{equation}
shows that the assertion is correct and needs no further justification because both fields are free. 
\end{proof}

We will now see that canonical perturbation theory enables us to compute the Feynman propagator of the field
$\varphi$ with mass $m$ from the Feynman propagator of $\varphi_{0}$ with mass $m_{0}$. 

\begin{Claim}
Let the symbol $V_{m}(t) = \Ti \{ e^{-i\int_{0}^{t}  \Ha^{m}_{I}} \}$ represent the intertwiner between
the two free fields $\varphi$ and $\varphi_{0}$ of masses $m$ and $m_{0}$, respectively, ie formally, 
\begin{equation}\label{formint}
 \varphi(t,\mathbf{x}) = \Ti \{ e^{i\int_{0}^{t}  \Ha^{m}_{I}} \} 
                               \varphi_{0}(t,\mathbf{x}) 
        \Ti \{ e^{-i\int_{0}^{t}  \Ha^{m}_{I}} \} .
\end{equation}
Then the Gell-Mann-Low expansion (\ref{TMpert}) yields 
\begin{equation}
  \la \Omega |\mathsf{T} \{ \varphi(x) \varphi(y) \} \Omega \ra  = 
  \frac{\la \Omega_{0} |\mathsf{T}\{ \mathsf{S}_{m}  \varphi_{0}(x) \varphi_{0}(y) \}
  \Omega_{0} \ra}{\la \Omega_{0} | \mathsf{S}_{m}\Omega_{0} \ra}
  = i\Delta_{F}(x-y;m^{2}),
\end{equation}
and the intertwiner represented by the symbol $V_{m}(t)$ is not unitary. 
\end{Claim}
\begin{proof}
We compute the rhs of (\ref{TMpert}) by using Wick's theorem and obtain
\begin{equation}\label{TMpert1}
 \la \Omega |\mathsf{T} \{ \varphi(x) \varphi(y) \} \Omega \ra = i \Delta_{F}(x-y; m_{0}^{2}) 
 +  i \sum_{n \geq 1} (\delta m^{2})^{n} \Delta_{F}^{*n+1}(x-y;m_{0}^{2}) 
\end{equation}
in which $\Delta_{F}^{*n+1}(x-y;m_{0}^{2})$ is the $(n+1)$-fold convolution of the Feynman propagator 
$\Delta_{F}$ given by
\begin{equation}
\begin{split}
\Delta_{F}^{*n+1}(x-y;m_{0}^{2}) = \int_{\Mi} dz_{1} \ ... \int_{\Mi} dz_{n} \ 
\Delta_{F}(x-z_{1}; & m_{0}^{2})  \Delta_{F}(z_{1}-z_{2};m_{0}^{2}) \\ 
& \ ... \ \Delta_{F}(z_{n-1}-z_{n};m_{0}^{2}) \Delta_{F}(z_{n}-y;m_{0}^{2}),
\end{split}
\end{equation}
for $n \geq 1$. The corresponding Feynman diagrams are all of the form
\begin{equation}\label{M}
G = \graph{0}{-0.15}{0.7}{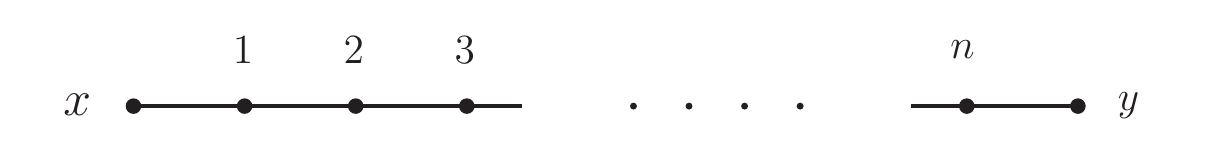}{0.5} .
\end{equation}
In momentum space, (\ref{TMpert1}) takes the simple form 
\begin{equation}\label{TMpert2}
\begin{split}
\int \frac{d^{4}p}{(2 \pi)^{4}} e^{-i p \cdot x} \la \Omega |\mathsf{T} \{ \varphi(x) \varphi(y) \} \Omega \ra
&= i D_{0}(p)^{-1}  +  i \sum_{n \geq 1} (\delta m^{2})^{n} D_{0}(p)^{-(n+1)} \\
&= \frac{i D_{0}(p)^{-1}}{ 1 - \delta m^{2} D_{0}(p)^{-1} } = \frac{i}{ D_{0}(p) - \delta m^{2}}
= \frac{i}{D(p)}
\end{split}
\end{equation}
in which $D_{0}(p) = p^{2} - m^{2}_{0} + i 0^{+}$ and $D(p) = p^{2} - m^{2}+ i 0^{+}$
are the inverse free propagators with masses $m_{0}$ and $m$, respectively. By Theorem \ref{freeH}, the
field interwtiner is not unitary. 
\end{proof}

We conclude that the intertwiner between the two fields, symbolically represented by (\ref{ms}), does 
indeed transform the field $\varphi_{0}$ into $\varphi$ but cannot be unitary on 
account of Theorem \ref{freeH} although the Hamiltonian (\ref{TMHam}) is actually a well-defined 
operator-valued distribution  when smeared in space and time.
In fact, it is self-adjoint because Wick powers of free fields are \cite{BruFK96}. 

However, the intertwiner (\ref{ms}) is far from being a simple exponentiation of the Hamiltonian $\Ha^{m}_{I}(x)$ 
because of the integration and time-ordering. It is therefore not surprising that $V_{m}(t)$ is not unitary. 

\subsection{Counterterms induce mass shift}
Because the intertwiner of renormalised $\varphi^{4}$ theory, symbolically given by
\begin{equation}
V_{r}(t)=\Ti \{ e^{-i\int_{0}^{t} \Ha^{r}_{I}} \} 
\end{equation}
with 
\begin{equation}\label{renH'}
\Ha^{r}_{I}(x) = \frac{g_{r}}{4!}Z_{g}(g_{r})\varphi_{r,0}(x)^{4} 
+ \frac{1}{2}[Z(g_{r})-1](\partial \varphi_{r,0}(x))^{2} 
+ \frac{1}{2}m_{r}^{2}[Z_{m}(g_{r})-1]\varphi_{r,0}(x)^{2}
\end{equation}
also contains a mass shift interaction piece given by the last (two) term(s), we expect it to be also 
non-unitary. 
This means that \emph{renormalised $\varphi^{4}$ theory is not affected by Haag's theorem} because the 
central proviso of unitary equivalence is not fulfilled. 
Therefore, \emph{renormalisation circumvents this theorem}. If renormalised $\varphi^{4}$-theory
exists (which we believe), then it is not difficult 
to see that the renormalised fully interacting field is certainly not unitary equivalent to a free field 
by simply taking a look at the mass dependence of its inverse propagator
\begin{equation}
 D_{r}(g_{r},p,m_{r}) = p^{2} - m_{r}^{2} - \Sigma_{r}(g_{r},p,m_{r}) + i0^{+}
\end{equation}
with self-energy $\Sigma_{r}(g_{r},p,m_{r})$. It is this very function which reflects the 
complexity of relativistic interactions, in some sense correctly captured by the counterterms 
in (\ref{renH'}). 

Given that even the slightest mass shift performed on a free field amounts to a non-unitary 
transformation to another free field, it is next to impossible for a theory with an interaction 
dependent (!) mass 
\begin{equation}
M^{2} = m_{r}^{2} + \Sigma_{r}(g_{r},p,m_{r})
\end{equation}
to be unitary equivalent to a free field in Fock space of mass $m_{r}$.

To assume that an interacting field theory is unitary equivalent is an erroneous idea suggested
by the form of the formal identity
\begin{equation}\label{HeiInt1}
 \varphi(t,\mathbf{x}) = e^{iHt} e^{-iH_{0}t}  \varphi_{0}(t,\mathbf{x}) e^{iH_{0}t} e^{-iHt} 
 = V(t)^{\dagger}\varphi_{0}(t,\mathbf{x}) V(t)
\end{equation}
for the intertwiner $V(t)$. The canonical formalism works with this assumption in both the renormalised
and the unrenormalised case. Only in the unrenormalised case did the fallacy become apparent by emerging divergences. 

The renormalised case is different because renormalisation effectively introduces new auxiliary interaction
terms so that the theory is changed drastically: 
by additional (auxiliary) interaction terms called counterterms, a renormalised theory is no longer the 
unrenormalised renormalisable theory is was prior to renormalisation.

\section{Conclusion: renormalisation staves off triviality}\label{sec:Conclus}                       

We have reviewed Haag's theorem and some pertinent triviality results, scrutinising in particular the details of the proof of Haag's theorem and its provisions.  
The most salient provision turned out to be unitary equivalence between free and interacting quantum fields. 

Because the theorem is independent of the 
dimension of Minkowski spacetime\footnote{Spacetime must have dimension $d \geq 2$, otherwise we would have either space or time, in particular would we 
have no boosts.}, it holds also for superrenormalisable quantum fields. 
A circumvention scheme in this context has been found by constructive field theorists: (super)renormalisation. And here, to the best of our knowledge,
these authors do not claim unitary equivalence.  

Although the case against quantum electrodynamics (and hence also quantum chromodynamics) is less clear due to a fundamental incompatibility with 
Wightman's framework, there is no doubt that the interaction picture can also not exist there.
We have argued that 
\begin{equation*}
\mbox{\textsc{renormalisation bypasses Haag's theorem}} 
\end{equation*}
in all cases by effectively rendering the field intertwiner non-unitary. This cannot be proved due to the mathematical elusiveness of the involved 
renormalisation $Z$ factors: per se only defined perturbatively to master their task of subtracting divergences, their nonperturbative status is totally 
obscure \cite{Ost86}. 

In particular, the wave-function renormalisation strikes us as a tenuously locked piece of the jigsaw puzzle that quantum field theory presents itself as.
We contend that this 'constant' cannot convincingly establish the link between the spectral representation and Haag-Ruelle (or LSZ) 
scattering theory\footnote{The LSZ reduction formula, being very instructive for amputating Green's functions in momentum space, makes total sense 
nonetheless. We do not take a nihilist stance here.} with standard perturbative quantum field theory, the biggest and best-understood jigsaw chunk.    
Interpreting this renormalisation constant as a probability amplitude, as standard textbooks do, is something we would like to ask posterity's opinion about.  

Another important theme are the canonical (anti)commutation relations for quantum fields. The germane triviality theorems we have discussed here suggest that these relations are 
incompatible with interactions in (flat) spacetimes of dimensions $d\geq 4$. With these results, we call into question the status of such relations in 
these spacetimes, in particular on the grounds that there is no reasonable analogue of the position operator in quantum field theory. 
This is intimately tied up with the Heisenberg uncertainty relations whose implementation is more or less obscure in QFT, a very interesting side issue 
we have only touched upon briefly and in passing.

\section*{Acknowledgements}

I am indebted to my supervisor Dirk Kreimer for supporting me throughout my PhD time and beyond. He suggested this interesting topic to me.

Furthermore, I thank Raimar Wulkenhaar and David Broadhurst for carefully reading the text, giving their positive feedback and encouraging me to 
publish it in the form of a review.  

\part{Appendices}

\begin{appendix}

\section{Baumann's theorem}\label{sec:AppBau} 
\begin{Theorem}[Baumann]
Let $n \geq 4$ be the space dimension and $\varphi(t,\cdot)$ a scalar field with 
conjugate momentum field $\pi(t,\cdot)=\partial_{t}\varphi(t,\cdot)$ such that the CCR (\ref{CCRSmear}) are 
obeyed and assume furthermore that $\dot{\pi}(t,\cdot):=\partial_{t}\pi(t,\cdot)$ exists. 
Then, if $\varphi(t,\cdot)$ has a 
vanishing vacuum expectation value and the provisions listed in the introduction of \cite{Bau87} are 
satified, one has
\begin{equation}
 \dot{\pi}(t,f) - \varphi(t,\Delta f) + m^{2} \varphi(t,f) = 0
\end{equation}
for all $f \in \De(\R^{n})$ and a parameter $m^{2}>0$. 
\end{Theorem}
\begin{proof}
 We only provide a sketch. To see the technical details, the reader is referred to Baumann's paper 
 \cite{Bau87}. First, take any $\Psi \in \D$ from the dense domain of definition. Then
 \begin{equation}\label{firstC}
  [\varphi(t,f),\dot{\pi}(t,g)]\Psi = 0
 \end{equation}
is an easily obtained consequence of $\partial_{t}[\varphi(t,f),\pi(t,g)] = i \partial_{t}(f,g) = 0$ and
causality of the momentum field, ie $[\pi(t,f),\pi(t,g)]=0$ for any test functions $f, g \in \De(\R^{n})$.
The Jacobi identity then entails 
\begin{equation}
  [\varphi(t,f),[\pi(t,h),\dot{\pi}(t,g)]]\Psi = 0,
 \end{equation}
which is fairly straightforward. The hard part is to prove that 
$[\pi(t,f),[\pi(t,h),\dot{\pi}(t,g)]]\Psi = 0$. Baumann uses a partition of unity whose constituting 
test functions have compact support in $\epsilon$-neighbourhoods that cover the support of the test
functions $f,g$ and $h$. Then he essentially shows that this double commutator vanishes with a power
law for $\epsilon \downarrow 0$ that depends on the space dimension $n>0$. If $n \geq 4$, then it 
vanishes. By virtue of the irreducibility of the field algebra $\{ \varphi(t,\cdot), \pi(t,\cdot) \}$,
one concludes from $[\varphi(t,f),[\pi(t,h),\dot{\pi}(t,g)]]\Psi = 0$ and 
$[\pi(t,f),[\pi(t,h),\dot{\pi}(t,g)]]\Psi = 0$ that 
\begin{equation}\label{cnum}
 [\pi(t,h),\dot{\pi}(t,g)] = \la \Omega | [\pi(t,h),\dot{\pi}(t,g)]\Omega \ra,
\end{equation}
ie that $[\pi(t,h),\dot{\pi}(t,g)]$ is a c-number. The next step is to write the rhs of (\ref{cnum}) in 
terms of the K\"allen-Lehmann representation, ie a positive superposition of free commutator functions. 
Recall that the commutator function of the free field $\varphi_{0}$ with mass $m>0$ is given by
\begin{equation}\label{KaLCf}
[\varphi_{0}(t,\mathbf{x}),\varphi_{0}(s,\mathbf{y})] 
= \int \frac{d^{4}q}{(2\pi)^{3}} \ \delta_{+}(q^{2}-m^{2}) 
[e^{-iq \cdot (x-y)} - e^{+iq \cdot (x-y)} ]=:D(t-s,\mathbf{x}-\mathbf{y};m^{2}),
\end{equation}
where $\delta_{+}(q^{2}-m^{2}) := \theta(q_{0}) \delta(q^{2}-m^{2})$ makes sure the particle is
real and on-shell.  
One finds the corresponding representation of the commutator (\ref{cnum}) by first differentiating the 
K\"allen-Lehmann representation
\begin{equation}\label{KaLC}
\la \Omega | [\varphi(t,\mathbf{x}),\varphi(s,\mathbf{y})]\Omega \ra 
= \int d\mu^{2} \ \rho(\mu^{2}) D(t-s,\mathbf{x}-\mathbf{y};\mu^{2}) .
\end{equation}
twice with respect to $s$ which gives 
\begin{equation}
\la \Omega | [\varphi(t,\mathbf{x}),\dot{\pi}(s,\mathbf{y})]\Omega \ra 
= \int d\mu^{2} \ \rho(\mu^{2}) [\Delta-\mu^{2}]D(t-s,\mathbf{x}-\mathbf{y};\mu^{2})
\end{equation}
because of $(\partial_{s}^{2} - \Delta + \mu^{2})D(t-s,\mathbf{x}-\mathbf{y};\mu^{2})=0$. Differentiating
with respect to $t$ and letting $t \rightarrow s$ yields the distribution
\begin{equation}
\la \Omega | [\pi(t,\mathbf{x}),\dot{\pi}(t,\mathbf{y})]\Omega \ra 
= -i \int d\mu^{2} \ \rho(\mu^{2}) [\Delta-\mu^{2}]\delta(\mathbf{x}-\mathbf{y}) .
\end{equation}
Applying this to two test functions $f,g \in \De(\R^{n})$ gives 
\begin{equation}\label{KL}
\la \Omega | [\pi(t,f),\dot{\pi}(t,g)]\Omega \ra 
= -i \int d\mu^{2} \ \rho(\mu^{2}) (f,[\Delta-\mu^{2}]g) = -i [(f,\Delta g) - m^{2} (f,g)]
\end{equation}
with $m^{2} = \int d\mu^{2} \ \rho(\mu^{2}) \mu^{2}$ and $\int d\mu^{2} \ \rho(\mu^{2}) = 1$. The latter is
implied by the CCR, obtained from (\ref{KaLC}). Note that $m^{2} < \infty$ is warranted by the existence
of the state $\dot{\pi}(g)\Omega$. Using the CCR, we can write the rhs of (\ref{KL}) as the commutator 
\begin{equation}
-i [(f,\Delta g) - m^{2} (f,g)] = [\pi(t,f),\varphi(t,\Delta g)-m^{2}\varphi(t,g)]
\end{equation}
and arrive at $[\pi(t,f),\dot{\pi}(t,g)-\varphi(t,\Delta g)+m^{2}\varphi(t,g)]=0$ because the commutator
in the vacuum expectation value of the lhs of (\ref{KL}) is a c-number. The CCR in combination with 
(\ref{firstC}) entail 
\begin{equation}
[\varphi(t,f),\dot{\pi}(t,g)-\varphi(t,\Delta g)+m^{2}\varphi(t,g)]=0.
\end{equation}
By the irreducibility of the field algebra, this means that 
$C:=\dot{\pi}(t,g)-\varphi(t,\Delta g)+m^{2}\varphi(t,g)$ is a c-number. The normalisation 
$\la \Omega | \varphi(t,f) \Omega \ra = 0$ then yields $C=0$.  
\end{proof}

\section{Wightman's reconstruction theorem}\label{sec:Recon} 
We denote the algebra of Schwartz functions on Minkowski space $\Mi$ by $\Sw(\Mi)$.

\begin{Theorem}[Reconstruction theorem]
 Let $\{ W_{n} \colon \Sw(\Mi)^{n} \rightarrow \C \}$ be a family of tempered distributions satisfying the following set of properties.
 \begin{enumerate}
 
 \item \textsc{Poincaré invariance}. 
 $W_{n}(f_{1}, ... , f_{n})=W_{n}(\{a,\Lambda\}f_{1}, ... , \{a,\Lambda\}f_{n})$ for all Poincaré
 transformations $(a,\Lambda) \in \Po$, where $(\{ a, \Lambda \}f)(x):=f(\Lambda^{-1}(x-a))$.
 
 \item \textsc{Spectral condition}. $W_{n}$ vanishes if one test 
 function's Fourier transform has its support outside the forward light cone, that is,
 \begin{equation}
  \wt{W}_{n}(\wt{f}_{1}, ... , \wt{f}_{n}) = W_{n}(f_{1}, ... , f_{n})= 0
 \end{equation}
 if there is a $j$ such that $\wt{f}_{j}(p)=0$ for all $p \in \overline{V}_{+}$. This means $\wt{W}_{n}$ has
 its support inside the forward light cone $(\overline{V}_{+})^{n}$. 
 
 \item \textsc{Hermiticity}. 
 $W_{n}(f_{1}, ... , f_{n})=W_{n}(f^{*}_{n}, ... , f^{*}_{1})^{*}$.
 
 \item \textsc{Causality}. If $f_{j}$ and $f_{j+1}$ have mutually spacelike separated support, then
 \begin{equation}
  W_{n}(f_{1}, ... , f_{j}, f_{j+1}, ... , f_{n}) = W_{n}(f_{1}, ... , f_{j+1}, f_{j}, ... , f_{n}).
 \end{equation}
 
 \item \textsc{Positivity}. For all $f_{j,l} \in \Sw(\Mi)$ one has
 \begin{equation}
  \sum_{n\geq 0} \sum_{j+k=n} W_{n}(f^{*}_{j,j}, ... , f^{*}_{j,1}, f_{k,1}, ... , f_{k,k}) \geq 0 , 
 \end{equation}
 where $W_{0}=|f_{0}|^{2}\geq 0$.
 
 \item \textsc{Cluster decomposition} Let $a\in \Mi$ be spacelike ($a^{2}<0$), then
 \begin{equation}
  \lim_{\lambda \rightarrow \infty} 
  W_{n}((f_{1}, ... , f_{j}, \{\lambda a,1\}f_{j+1}, ... , \{\lambda a,1\}f_{n})
  = W_{j}(f_{1}, ... ,f_{j})W_{n-j}(f_{j+1}, ... , f_{n}).
 \end{equation}
\end{enumerate}

 Then there is a scalar field theory fulfilling the Wightman axioms 0 to IV. Any other theory is unitarily 
 equivalent.
\end{Theorem}
\begin{proof} Axioms 0 \& I: we start by considering the vector space $\D$ given by terminating sequences
\begin{equation}
\Psi_{f} = (f_{0},f_{1},f_{2},...,f_{n}, 0, 0, ...) \ , \hs{2} n \in \N
\end{equation}
with elements $f_{k} \in \Sw(\Mi^{k})$ and hence $\D = \bigoplus_{n\geq 0}\Sw(\Mi^{n})$, where $\Sw(\Mi^{0}):=\C$.
We make use of the Wightman distributions to define an inner product on $\D$ by
\begin{equation}\label{preHilInn}
\la \Psi_{f} | \Psi_{g} \ra = \la (f_{0},f_{1},f_{2},...) | (g_{0},g_{1},g_{2},...) \ra 
                  := \sum_{n\geq 0} \sum_{j+k=n} W_{n}(f^{*}_{j} \te g_{k}),
\end{equation}
where for $n=0$ we set $W_{0}(f_{0},g_{0})=f_{0}g_{0}$, ie simply the product in $\C$. 
The sum in (\ref{preHilInn}) is finite because the sequences terminate. 
$\la \Psi_{f} | \Psi_{g} \ra = \la \Psi_{g} | \Psi_{f} \ra^{*}$ is guaranteed by property (3).  
A representation of the Poincaré group is established by introducing the linear map
\begin{equation}\label{Poinc}
 U(a,\Lambda)(f_{0},f_{1},f_{2},...) := (f_{0},\{a,\Lambda\}f_{1}, \{a,\Lambda \}f_{2}, ... ),
\end{equation}
where $(\{a,\Lambda\}f_{j})(x_{1},...,x_{n}):=f_{j}(\Lambda^{-1}(x_{1} - a),..., \Lambda^{-1}(x_{j} - a))$, 
ie all arguments are equally transformed. Poincaré invariance of the inner product (\ref{preHilInn}) is 
obvious from property (1). This implies in particular 
that the representation of the Poincaré group is unitary, ie
\begin{equation}
 \la U(a,\Lambda)\Psi | U(a,\Lambda)\Psi' \ra = \la \Psi | \Psi' \ra
\end{equation}
holds for any states $\Psi=(f_{0},f_{1},f_{2}, ...)$, $\Psi'=(f'_{0},f'_{1},f'_{2}, ...)$. 
Furthermore, we infer from (\ref{Poinc}) that the vacuum state is merely the vector
\begin{equation}
 \Psi_{0} = (1,0,0, ...).
\end{equation}
It is an easy exercise to show on states in $\D$ that that this representation obeys the group law
\begin{equation}
 U(a',\Lambda')U(a,\Lambda)=U(a'+\Lambda' a,\Lambda' \Lambda)
\end{equation}
and the property $||U(a,\Lambda)\Psi - \Psi|| \rightarrow 0$ as $(a,\Lambda) \rightarrow (0,1)$. This latter 
properties implies strong continuity of the representation by the Cauchy-Schwartz inequality. 
The cluster decomposition property (6) implies that the vacuum is unique: let there be another 
Poincaré-invariant state $\Psi_{0}' \in \D$. One may assume $\la \Psi_{0}' | \Psi_{0} \ra = 0$.  
Then, for spacelike $a \in \Mi$, we have 
\begin{equation}
\begin{split}
\la \Psi_{0}' | \Psi_{0}' \ra  &= 
\lim_{\lambda \rightarrow \infty} \la \Psi_{0}' | U(\lambda a,1) \Psi_{0}' \ra  
= \lim_{\lambda \rightarrow \infty} \sum_{n\geq 0}\sum_{j+k=n}W_{n}(f_{j}^{*} \te U(\lambda a,1)f_{k}) \\
& =  \sum_{n\geq 0}\sum_{j+k=n}W_{j}(f_{j}^{*}) W_{k}(f_{k}) 
 = \left(\sum_{n \geq 0}W_{j}(f_{j}^{*})\right) \left(\sum_{m \geq 0} W_{k}(f_{k})\right) 
 = \la \Psi_{0}' | \Psi_{0} \ra \la \Psi_{0} | \Psi_{0}' \ra .
\end{split}
\end{equation}
We refer the interested reader to \cite{StreatWi00} in which the completion of $\D$ with respect to 
Cauchy series, nullifying of zero norm states and the spectral property of $P^{\mu}$ are discussed,
the latter follows from property (2) .

Axiom II: a quantum field can now be defined through the assignment of $h \in \Sw(\Mi)$ to the operation
\begin{equation}
 \varphi(h) (f_{0},f_{1},f_{2},...) = (0,f_{0}h,h \te f_{1},h \te f_{2},...).
\end{equation}
All operators in the algebra $\Al(\Mi)$ are then of the form 
\begin{equation}
 A = h_{0} + \varphi(h_{1}) + \varphi(h_{2,1}) \varphi (h_{2,2}) +  ... 
+ \varphi(h_{n,1})  ... \varphi(h_{n,n}) \in \Al(\Mi).
\end{equation}
Applying this to the vacuum generates the state
$\Psi = (h_{0}, h_{1}, h_{2,1} \te h_{2,2}, ... , h_{n,1} \te ... \te h_{n,n}, 0, 0, ...)$ which is
an element in $\D$. It is obvious that this space is stable under the action of both the field
and the Poincaré representation:
\begin{equation}
 \varphi(\Sw(\Mi))\D \subset \D  \ , \hs{1}  U(\Po)\D \subset \D.
\end{equation} 
However, the algebra $\Al(\Mi)$ generates only the dense subspace
\begin{equation}
\D_{0}:= \Al(\Mi)\Psi_{0} = \bigoplus_{n \geq 0}\Sw(\Mi)^{\te n} \subset \D
\end{equation}
but not $\D$.  That the map 
$f \mapsto \la \Psi | \varphi(f) \Psi' \ra $ is a tempered distribution for all $\Psi,\Psi' \in \D$ 
is obvious from
\begin{equation}
 \la \Psi | \varphi(f) \Psi' \ra = \la (h_{0},h_{1},h_{2},...) | \varphi(f)(g_{0},g_{1},g_{2},...) \ra
 = \sum_{n\geq 0} \sum_{j+k=n} W_{n}(h_{j} \te f \te g_{k-1})
\end{equation}
because for all $n\in \N$ the assignment $f \mapsto W_{n}(h_{j} \te f \te g_{k-1})$ is a tempered 
distribution and the sum over all $n$ is finite (we set $g_{-1}=0$).
Poincaré covariance (\ref{smPoin}) on $\D$ and hence the validity of Axiom III is easy to see by
\begin{equation}
 \begin{split}
  U(a,\Lambda) \varphi(h) (f_{0},f_{1},f_{2},...) 
  &= U(a,\Lambda) (0,f_{0}h,h \te f_{1},h \te f_{2},...)  \\
  &=  (0,f_{0}\{ a,\Lambda\}h,\{ a,\Lambda\}h \te \{ a,\Lambda\}f_{1},
   \{ a,\Lambda\}h \te \{ a,\Lambda\}f_{2},...) \\
  &= \varphi(\{ a,\Lambda\}h) (f_{0},\{ a,\Lambda\}f_{1},\{ a,\Lambda\}f_{2},...) \\
  &= \varphi(\{ a,\Lambda\}h) U(a,\Lambda) (f_{0},f_{1},f_{2},...),
 \end{split}
\end{equation}
that is, $U(a,\Lambda) \varphi(h)U(a,\Lambda)^{\dagger}\D = U(a,\Lambda) \varphi(h)U(a,\Lambda)^{-1}\D 
= \varphi(\{a,\Lambda\}h)\D$. 

Local commutativity is guaranteed by property (4): let $A,B \in \Al(\Mi)$ generate the two states 
$\Psi=A\Psi_{0}$ and $\Phi=B\Psi_{0}$. If $f,g \in \Sw(\Mi)$ have mutually spacelike support, then
\begin{equation}
 \la \Psi | \varphi(f)\varphi(g) \Phi \ra = \la \Psi_{0} | B^{*}\varphi(f)\varphi(g) A \Psi_{0} \ra
  =  \la \Psi_{0} | B^{*}\varphi(g)\varphi(f) A \Psi_{0} \ra =  \la \Psi | \varphi(g)\varphi(f) \Phi \ra ,
\end{equation}
in which the second step makes use of property (4). Hence $[\varphi(f),\varphi(g)]=0$ on $\D_{0}$. Because $\D_{0} \subset \D$ is dense this also holds on $\D$.

To see that any other theory giving rise to the same Wightman distributions is unitarily equivalent to 
the one just constructed, let $\phi$ be the other field and $\Omega_{0}$ its vacuum state. We define a 
linear map by setting
\begin{equation}
 V \Psi_{0} := \Omega_{0} \ , \hs{1} 
 V \varphi(f_{1})...\varphi(f_{n})\Psi_{0} := \phi(f_{1})...\phi(f_{n})\Omega_{0}. 
\end{equation}
Then for a general state $\Psi_{f} = (f_{0}, f_{1} , ... , f_{n} ) \in \D_{0}$ with 
$f_{j}=f_{j,1} \te ... \te f_{j,j} \in \Sw(\Mi)^{\te j}$ we have 
\begin{equation}
 V \Psi_{f} = V [f_{0} \Psi_{0} + \sum_{n\geq 1} \varphi(f_{n,1}) ... \varphi(f_{n,n}) \Psi_{0} ]
 = f_{0} \Omega_{0} + \sum_{n\geq 1} \phi(f_{n,1}) ... \phi(f_{n,n}) \Omega_{0}. 
\end{equation}
It is unitary because $\la V \Psi_{f} | V \Psi_{g} \ra = \la \Psi_{f} |\Psi_{g} \ra$ follows from 
coinciding vacuum expectation values. Next, we consider
\begin{equation}
\begin{split}
 V \varphi(h) \Psi_{f} &= V [ f_{0} \varphi(h) \Psi_{0} 
  + \sum_{n\geq 1} \varphi(h) \varphi(f_{n,1}) ... \varphi(f_{n,n}) \Psi_{0} ] \\
  &= f_{0} \phi(h) \Omega_{0} + \sum_{n\geq 1} \phi(h) \phi(f_{n,1}) ... \phi(f_{n,n}) \Omega_{0} \\
  &= \phi(h) [ f_{0} \Omega_{0} + \sum_{n\geq 1} \phi(f_{n,1}) ... \phi(f_{n,n}) \Omega_{0} ]
   = \phi(h) V \Psi_{f}
\end{split}
\end{equation}
which means $V\varphi(h)=\phi(h)V$ on $\D_{0}$, ie $V\varphi(h)V^{-1}=\phi(h)$ on the dense subspace 
generated by the field $\phi$. 
\end{proof}

\section{Jost-Schroer theorem}\label{sec:AppJoSch} 
\begin{Theorem}[Jost-Schroer Theorem]
Let $\varphi$ be a scalar field whose two-point Wightman distribution coincides with that of a free
field with mass $m>0$, ie 
\begin{equation}
 \la \Psi_{0} |\varphi(f) \varphi(h) \Psi_{0} \ra
 =  \int \frac{d^{4}p}{(2\pi)^{4}} \  \wt{f}^{*}(p) 2 \pi \theta(p_{0})  \delta(p^{2}-m^{2}) \wt{h}(p) .
\end{equation}
Then $\varphi$ is itself a free field of the same mass. 
\end{Theorem}
\begin{proof}
We follow \cite{StreatWi00}. First note that because $ \la \Psi_{0} | \varphi([\Box + m^{2}]f) \varphi([\Box + m^{2}]h) \Psi_{0} \ra = 0$,
where $\Box + m^{2}$ is the Klein-Gordon operator, represented by the multiplication 
$\wt{f}(p) \mapsto (-p^{2}+m^{2})\wt{f}(p)$ in momentum space for a Schwartz function $f$. This means in particular 
\begin{equation}
|| j(f)\Psi_{0} ||^{2} =  \la \Psi_{0} | \varphi([\Box + m^{2}]f^{*}) \varphi([\Box + m^{2}]f) \Psi_{0} \ra = 0
\end{equation}
for the operator valued distribution $j(f):=\varphi([\Box + m^{2}]f)$ and thus $j(f) \Psi_{0} = 0$. 
By the Reeh-Schlieder theorem (Theorem \ref{Comu}) we have $j(f)= 0$ for all $f \in \Sw(\Mi)$. 
This means that $\varphi$ obeys the free Klein-Gordon equation. We define the Fourier transform of the 
operator field by
\begin{equation}
 \wt{\varphi}(\wt{f}):=\varphi(f),
\end{equation}
where $\wt{f} \in \Sw(\Mi)$ is the Fourier transform of $f \in \Sw(\Mi)$ (and vice versa). 
The d'Alembertian takes the form of a multiplication operator 
\begin{equation}
 [\wt{\Box}+m^{2}]\wt{f}(p)=[-p^{2}+m^{2}]\wt{f}(p).
\end{equation}
The Klein-Gordon equation for the Fourier transform of the field reads 
$\wt{\varphi}([\wt{\Box}+m^{2}]\wt{f})=0$ for all $\wt{f} \in \Sw(\Mi)$. This means that $\wt{\varphi}$ 
vanishes unless the test function has some of its support inside the two hyperbolae 
\begin{equation}
H_{m}^{\pm} = \{ \  p \in M \ | \ p^{2} = m^{2} , p_{0} \gtrless 0 \ \}.
\end{equation}
The spectrum therefore is contained in the pair of mass hyperbolae, ie $\sigma (\wt{\varphi}) \subset  H_{m} := H_{m}^{+} \cup H_{m}^{-}$. 
We introduce two functions $\chi_{\pm} \in \Sw(\Mi)$ with $\chi_{\pm}(p)=1$ 
on the hyperbola $H_{m}^{\pm}$ and vanishing outside some neighbourhood of $H_{m}^{\pm}$. This enables us to split
\begin{equation}
 \wt{\varphi}(\wt{f}) =  \wt{\varphi}(\chi_{-}\wt{f}) + \wt{\varphi}(\chi_{+}\wt{f})  
 = \wt{\varphi}_{+}(\wt{f}) + \wt{\varphi}_{-}(\wt{f}), 
\end{equation}
where $\wt{\varphi}_{\pm}(\wt{f}):=\wt{\varphi}(\chi_{\mp}\wt{f})$, admittedly a very confusing convention.
However, their Fourier transforms are 
\begin{equation}
 \varphi_{\pm}(f):= \wt{\varphi}_{\pm}(\wt{f}),
\end{equation}
that is, the positive (negative) and negative (positive) frequency (energy) parts of the field. 
Because there are no negative energy states, one has 
\begin{equation}\label{ann}
\varphi_{+}(f)\Psi_{0} = 0.
\end{equation}
Because the state $\varphi_{+}(f)\varphi_{-}(h)\Psi_{0}$ has a forward and a backward timelike momentum, its
sum, the total momentum, is spacelike. One cannot create such a state. Either this state vanishes or it 
is a multiple of the vacuum. The only acceptable answer is that the latter is the case, 
since otherwise $\varphi_{-}(f)$ would annihilate the vacuum or be trivial, ie a multiple of the identity,
neither of which makes sense. By (\ref{ann}) we see that
\begin{equation}
 \la \Psi_{0} | \varphi_{+}(f)\varphi_{-}(h) \Psi_{0} \ra 
 = \la \Psi_{0} | \varphi(f)\varphi(h) \Psi_{0} \ra = W(f,h)
\end{equation}
and hence $[\varphi_{+}(f),\varphi_{-}(h)] \Psi_{0} = \varphi_{+}(f)\varphi_{-}(h) \Psi_{0} 
= W(f,h) \Psi_{0}$ because the state $\varphi_{+}(f)\varphi_{-}(h) \Psi_{0}$ is a multiple of the vacuum. 
The commutator then acts on the vacuum as
\begin{equation}
 [\varphi(f),\varphi(h)]\Psi_{0} = \{ W(f,h)-W(h,f) \} \Psi_{0} + [\varphi_{-}(f),\varphi_{-}(h)] \Psi_{0} .
\end{equation}
The first thing we notice about the last term is 
$\la \Psi_{0} |[\varphi_{-}(f),\varphi_{-}(h)] \Psi_{0} \ra =0$. The next aspect is that the distribution
\begin{equation}
 F(f,h):= \la \Psi |[\varphi_{-}(f),\varphi_{-}(h)] \Psi_{0} \ra
\end{equation}
for any $\Psi \in \D$ vanishes when $f$ and $h$ have mutually spacelike support. Consequently, the 
analytic continuation $F(z_{1},z_{2})$ into the forward tube $\Tu_{2} = \Mi^{2} + i (V_{+})^{2}$ 
vanishes on the open subset of real spacelike $z_{1}-z_{2}$. Since these vectors comprise an open 
set $E \subset \Mi^{2}$, this distribution vanishes by the edge of the wedge theorem (Theorem \ref{edwed}). 
Therefore, we have
\begin{equation}
 ([\varphi(f),\varphi(h)] -\{ W(f,h)-W(h,f) \}) \Psi_{0} = 0,
\end{equation}
ie the operator $T=[\varphi(f),\varphi(h)] -\{ W(f,h)-W(h,f) \}$ annihilates the vacuum. Since by Theorem
\ref{Comu} no annihilator other than the null operator can be constructed from a locally generated
polynomial algebra of fields, one has $T=0$ and thus
\begin{equation}
[\varphi(f),\varphi(h)] = W(f,h)-W(h,f) \ , \hs{1} \varphi([\Box + m^{2}]f)=0
\end{equation}
for all $f,h \Sw(\Mi)$. This entails $W_{n}=0$ for $n$ odd and 
\begin{equation}
 W_{n}(f_{1}, ... , f_{n} ) = \sum_{(i,j)} W(f_{i_{1}},f_{j_{1}}) ... W(f_{i_{m}},f_{j_{m}})
\end{equation}
for even $n=2m$, where the sum is over all permutations of $1, 2, ... , n$ written as 
$i_{1},j_{1}, ... , i_{m},j_{m}$ such that 
$i_{1} < i_{2} < ... < i_{m} < j_{1} < ... < j_{m}$ (Wick contractions). 
\end{proof} 
 
\end{appendix}


\begin{thebibliography}{9}




\bibitem[AlHoMa08]{AlHoMa08} S.A. Albeverio, R.J. H\o egh-Krohn, S. Mazzucchi : \emph{Mathematical Theory of Feynman Path Integrals. An introduction.}
Springer (2008)

\bibitem[AMY12]{AMY12} K.V. Antipin, M.N. Mnatsakanova, Yu.S. Vernov: \emph{Haag's theorem in noncommutative
quantum field theory}. arXiv: math-ph/1202.0995v1 

\bibitem[Bar63]{Bar63} G. Barton: \emph{Introduction to Advanced Field Theory}. Wiley Interscience (1963)

\bibitem[Bau87]{Bau87} K. Baumann: \emph{On relativistic irreducible quantum fields fulfilling CCR}. 
J. Math. Phys. 28, 697-704 (1987)

\bibitem[Bau88]{Bau88} K. Baumann: \emph{On canonical irreducible quantum field theories describing 
bosons and fermions}.  J. Math. Phys. 29, 1225-1230 (1988) 

\bibitem[Blee81]{Blee81} D. Bleecker: \emph{Gauge Theory and Variational Principles}. Mineola, New York (1981)

\bibitem[Bleu50]{Bleu50} K. Bleuler : \emph{Eine neue Methode zur Behandlung der longitudinalen und 
skalaren Photonen}. Helv. Phys. Act. 23, 567-586 (1950)

\bibitem[BjoDre65]{BjoDre65} J.D. Bjorken, S.D. Drell : \emph{Relativistic Quantum Fields}. McGraw-Hill 
(1965)

\bibitem[Bogo75]{Bogo75} N.N. Bogoliubov, A.A. Logunov, I.T. Todorov: \emph{Introduction to Axiomatic 
Quantum Field Theory}. W.A. Benjamin, Inc. (1975)  

\bibitem[Bor62]{Bor62} H.-J. Borchers: \emph{On Structure of the algebra of field operators}. Nuovo Cimento
24, 214- 236 (1962)

\bibitem[BoHaSch63]{BoHaSch63} H.J. Borchers, R. Haag, B. Schroer : \emph{The Vacuum State in Quantum 
Field Theory}. Nuovo Cim. 29, 148 (1963)

\bibitem[BoRo33]{BoRo33} N. Bohr, L. Rosenfeld : \emph{Zur Frage der Messbarkeit der elektromagnetischen
Feldgr\"ossen}. Danske Videnskabernes Selskab Matematisk-Fysiske Meddelelser 12, 1-65 (1933)

\bibitem[BoRo50]{BoRo50} N. Bohr, L. Rosenfeld : \emph{Field and Charge Measurementa in Quantum Electrodynamics}. Phys. Rev. 78, 794-798 (1950)

\bibitem[BroK01]{BroK01} D.J. Broadhurst, D. Kreimer: \emph{Exact solution of Dyson-Schwinger equations for iterated one-loop integrals
and propagator coupling duality}, Nucl. Phys. B 600 (2001), 403 - 422, arXiv: hep-th/0012146 

\bibitem[BruFK96]{BruFK96} R. Brunetti, K. Fredenhagen, M. K\"ohler: \emph{The microlocal spectrum 
condition and the Wick polynomials of free fields}. Comm. Math. Phys. 180, 633 - 652 (1996) 

\bibitem[Bu75]{Bu75} D. Buchholz: \emph{Collision theory for massless Fermions}. Comm. Math. Phys. 42, 
269-279 (1975) 

\bibitem[Bu77]{Bu77} D. Buchholz: \emph{Collision theory for massless Bosons}. Comm. Math. Phys. 52, 
147-173 (1977) 

\bibitem[Bu00]{Bu00} D. Buchholz: \emph{Algebraic Quantum Field Theory: A Status Report}. 
arXiv: math-ph/0011044v1 (2000)

\bibitem[BuHa00]{BuHa00} D. Buchholz, R. Haag: \emph{The Quest for Understanding in Relativistic Quantum Physics}. J. Math. Phys. 41, 3674-3697 (2000),
arXiv: hep-th/9910243

\bibitem[Chai68]{Chai68} J.M. Chaiken : \emph{Number Operators for Representations of the Canonical 
Commutation Relations}. Comm. Math. Phys. 8, 164-184 (1968)

\bibitem[Co53]{Co53} J.M. Cook: \emph{The Mathematics of Second Quantization}. Trans. Amer. Math. Soc. 74, 222-245 (1953)

\bibitem[CoKrei98]{CoKrei98} A. Connes, D. Kreimer : \emph{Hopf algebras, renormalisation and noncommutative geometry}. Comm. Math. Phys. 199, 203-242 
(1998), arXiv: hep-th/9808042

\bibitem[CoKrei00]{CoKrei00} A. Connes, D. Kreimer: \emph{Renormalisation in quantum field theory and the
Riemann-Hilbert problem I: The Hopf algebra structure of graphs and the main theorem.} Comm. Math. Phys.
210, 249-273 (2000)

\bibitem[DeDoRu66]{DeDoRu66} G.-F. Dell'Antonio, S. Doplicher, D. Ruelle : \emph{A Theorem on Canonical
Commutation Relations}. Comm. Math. Phys. 2, 223 -230 (1966)

\bibitem[DeDo67]{DeDo67} G.-F. Dell'Antonio, S. Doplicher : \emph{Total Number of Particles and Fock
representation}. J. Math. Phys. 8, 663-666 (1967) 

\bibitem[Di73]{Di73} J. Dimock: \emph{Asymptotic Perturbation Expansion in the $P(\phi)_{2}$ Quantum 
Field Theory}. Comm. Math. Phys. 35, 347-356 (1974) 

\bibitem[Di11]{Di11} J. Dimock : \emph{Quantum Mechanics and Quantum Field Theory. A Mathematical Primer}.
Cambridge (2011)

\bibitem[Dir31]{Dir31} P. Dirac: \emph{Quantised Singularities in the Electromagnetic Field}. Proc. R. Soc. Lond. A 133, 60-72 (1931)

\bibitem[Dix58]{Dix58} J. Dixmier: \emph{Sur la relation $i(PQ-QP)=1$}. Comp. Math. 13, 263-270 (1958)

\bibitem[DresKa62]{DresKa62} M. Dresden, P.B. Kahn: \emph{Field theories with persistent one particle states. I. General formalism}. Rev. Mod. Phys. 34, 401 (1962)

\bibitem[Dunc21]{Dunc21} A. Duncan: \emph{The Conceptual Framework of Quantum Field Theory}. Oxford University Press (2012)


\bibitem[Dys49a]{Dys49a} F.J. Dyson: \emph{The Radiation Theories of Tomonaga, Schwinger, and Feynman}. 
Phys. Rev. 75, 486 - 502 (1949) 

\bibitem[Dys49b]{Dys49b} F.J. Dyson: \emph{The S Matrix in Quantum Electrodynamics}. Phys. Rev. 75, 
1736 -1755 (1949) 
 
\bibitem[Dys51]{Dys51} F.J. Dyson : \emph{Divergence of Perturbation Theory in Quantum Electrodynamics}, Phys. Rev. 85 (1951), 631



\bibitem[EaFra06]{EaFra06} J. Earman, D. Fraser: \emph{Haag's theorem and its implication for the foundations of quantum field theory.} 
Erkenntnis, 64, 305-344 (2006) 


\bibitem[Em09]{Em09} G.G. Emch: \emph{Algebraic Methods in Statistical Mechanics and Quantum Field Theory}. 
Dover Publications, Inc. (2009)

\bibitem[FePStro74]{FePStro74} R. Ferrari, L. Picasso, F. Strocchi : \emph{Some Remarks on Local Operators in Quantum Electrodynamics}.
Comm. Math. Phys. 35, 25-38 (1974)

\bibitem[FeJo60]{FeJo60} P.G. Federbush, K.A. Johnson: \emph{Uniqueness Properties of the Twofold 
Vacuum Expectation Value}. Phys. Rev. 120, 1926 (1960) 

\bibitem[Feyn49]{Feyn49} R.P. Feynman: \emph{Space-Time Approach to Quantum Electrodynamics}. Phys. Rev. 49,
769 -789 (1949)

\bibitem[Feyn06]{Feyn06} R.P. Feynman : \emph{QED. The strange theory of light and matter}. Princeton University Press (2006)

\bibitem[Fo32]{Fo32} V. Fock: \emph{Konfigurationsraum und zweite Quantelung}. Zeitschr. f. Phys. 75, 
622-647 (1932)

\bibitem[Fra06]{Fra06} D.L. Fraser: \emph{Haag's Theorem and the interpretation of quantum field theories
with interactions}. Thesis, University of Pittsburgh (2006)


\bibitem[Fred10]{Fred10} K. Fredenhagen: \emph{Quantum field theory}. Lecture notes, Universit\"at Hamburg 
(2010), available online at \texttt{unith.desy.de/research/aqft/lecture\_notes}

\bibitem[Fried53]{Fried53} K.O.Friedrichs: \emph{Mathematical Aspects of the Quantum Theory of Fields}. 
Interscience Publishers, Inc., New York (1953)

\bibitem[Fro82]{Fro82} J. Fr\"ohlich: \emph{On the triviality of $\lambda \varphi_{d}^{4}$ theories and the 
approach to the critical point in $d\geq 4$ dimensions}. Nucl. Phys. B 200, 281-296 (1982)

\bibitem[GaWi54]{GaWi54} L. G\aa rding, A. Wightman: \emph{Representations of the Commutation Relations}. Proc. Nat. Acad. Sci. 40, 622-626 (1954)

\bibitem[GeMLo51]{GeMLo51} M. Gell-Mann, F. Low: \emph{Bound States in Quantum Field Theory}. Phys. Rev. 84, 350-354 (1951)

\bibitem[GeMLo54]{GeMLo54} M. Gell-Mann, F. Low: \emph{Quantum Electrodynamics at Small Distances States}. Phys. Rev. 95, 1300 (1954)

\bibitem[GliJaf68]{GliJaf68} J. Glimm, A. Jaffe: \emph{A $\lambda \phi^{4}$ quantum theory without cutoffs I}.
Phys. Rev. 176, 1945-1951 (1968) 

\bibitem[GliJaf70]{GliJaf70} J. Glimm, A. Jaffe: \emph{A $\lambda \phi^{4}$ quantum theory without cutoffs II}.
Ann. Math. 91, 362-401 (1970) 

\bibitem[GliJaf81]{GliJaf81} J. Glimm, A. Jaffe: \emph{Quantum Physics. A Functional Integral Point of View}.
Springer (1981)
Ann. Math. 91, 362-401 (1970) 

\bibitem[GliJaSp74]{GliJaSp74} J. Glimm, A. Jaffe, T. Spencer: \emph{The Wightman axioms and particle 
structure in the $P(\phi)_{2}$ quantum field model}. Ann. Math. Phys. 100, 585-632 (1974)

\bibitem[Gre59]{Gre59} O.W. Greenberg: \emph{Haag's Theorem and Clothed Operators}. Phys. Rev. 115, 706-710
(1959)

\bibitem[Gre61]{Gre61} O.W. Greenberg: \emph{Generalized Free Field and Models of Local Field Theory}.
Ann. Phys. 16, 158-176 (1961)

\bibitem[Gu50]{Gu50} S. N. Gupta :\emph{Theory of Longitudinal Photons in Quantum Electrodynamics}. 
Proc. Phys. Soc. 63, 681-691 (1950)

\bibitem[Gue66]{Gue66} M. Guenin : \emph{On the Interaction Picture}. Comm. math. phys. 3, 120 - 132 (1966)

\bibitem[GueRoSi72]{GueRoSi72} F. Guerra, L. Rosen, B. Simon: \emph{Nelson's Symmetry and the infinite
Volume Behaviour of the vacuum in $P(\phi)_{2}$}. Comm. Math. Phys. 27, 10-22 (1972)

\bibitem[Ha55]{Ha55} R. Haag: \emph{On Quantum Field Theories}. Danske Videnskabernes Selskab 
Matematisk-Fysiske Meddelelser 29, 1-37 (1955)

\bibitem[Ha58]{Ha58} R. Haag: \emph{Quantum Field Theories with Composite Particles and Asymptotic Conditions}. Phys. Rev. 112, 669-673 (1958)

\bibitem[Ha96]{Ha96} R. Haag: \emph{Local Quantum Physics}. Springer (1996)

\bibitem[HaKa64]{HaKa64} R.Haag, D.Kasterl: \emph{An Algebraic Approach to Quantum Field Theory}. 
J. Math. Phys. 5, 848-861 (1964)

\bibitem[HeiPau29]{HeiPau29} W.Heisenberg, W.Pauli: \emph{Zur Quantendynamik der Wellenfelder}. Zeitschr.
f. Phys. 56, 1-61 (1929) 

\bibitem[He63]{He63} K. Hepp: \emph{Lorentz-kovariante analytische Funktionen}. Helv. Phys. Acta 36, 355 (1963)

\bibitem[He66]{He66} K. Hepp: \emph{Proof of the Bogoliubov-Parasiuk Theorem on Renormalisation}. Comm. Math. Phys. 2, 301-326 (1966)

\bibitem[ItZu80]{ItZu80} C. Itzykson, J.-B. Zuber: \emph{Quantum Field Theory}. Dover Publications, Inc.
(1980)

\bibitem[Jaff69]{Jaff69} A. Jaffe: \emph{Whither axiomatic quantum field theory?} Rev. Mod. Pyhs. 41, 
576 - 580 (1969) 

\bibitem[Jo61]{Jo61} R. Jost: \emph{Properties of Wightman Functions}. In: Lectures on Field Theory and 
the Many-Body Problem, edited by E.R. Caianiello, Academic Press (1961)

\bibitem[Jo65]{Jo65} R. Jost: \emph{The General Theory of Quantized Fields}. Amercian Mathematical Society (1965)

\bibitem[Kla15]{Kla15} L. Klaczynski: \emph{Haag's Theorem in Renormalisable Quantum Field Theories}. PhD thesis, Humboldt University, 
supervisor: Dirk Kreimer (2015)

\bibitem[Krei02]{Krei02} D. Kreimer: \emph{Combinatorics of (perturbative) quantum field theory}. Phys. Rep. 363, 387-424 (2002)

\bibitem[LeBlo67]{LeBlo67} J.-M. Lévy-Leblond: \emph{Galilean Quantum Field Theories and a Ghostless Lee
Model}. Comm. Math. Phys. 4, 157-176 (1967)

\bibitem[Lo61]{Lo61} J. Lopuszanski : \emph{A Criterion for the Free Field Character of Fields}. 
J. Math. Phys. 2, 743 - 747 (1961)

\bibitem[LSZ55]{LSZ55} H.Lehmann, K.Symanzik, W.Zimmermann: \emph{Zur Formulierung quantisierter Feldtheorien}. Il Nuovo Cimento 1, 205-225 (1955)

\bibitem[LSZ57]{LSZ57} H.Lehmann, K.Symanzik, W.Zimmermann: \emph{On the Formulation of Quantized Field Theories II}. Il Nuovo Cimento 6, 319-333 (1957)

\bibitem[Lu05]{Lu05} T. Lupher: \emph{Who proved Haag's Theorem?} International Journal of Theoretical
Physics 44, 1995 - 2005 (2005) 

\bibitem[OS76]{OS76} K. Osterwalder, R. Sénéor: \emph{The scattering matrix is nontrivial for weakly
coupled $P(\varphi)_{2}$ models}. Helv. Phys. Act. 49, 525-535 (1976)

\bibitem[OSchra73]{OSchra73} K.Osterwalder, R.Schrader: \emph{Axioms for Euclidean Green's Functions}.
Comm. Math. Phys. 31, 83-112 (1973)

\bibitem[Ost86]{Ost86} K. Osterwalder: \emph{Constructive quantum field theory: goals, methods, results}. Helv. Phys. Acta 59, 220-228 (1986)

\bibitem[PeSch95]{PeSch95} M.E. Peskin, D.V. Schroeder: \emph{An Introduction to Quantum Field Theory}.
Westview Press (1995)

\bibitem[Po69]{Po69} K. Pohlmeyer : \emph{The Jost-Schroer Theorem for Zero-Mass Fields}. 
Comm. Math. Phys. 12, 204-211 (1969)

\bibitem[Pow67]{Pow67} R.T. Powers : \emph{Absence of Interactions as a Consequence of Good 
Ultraviolet Behaviour in the Case of a Local Fermi Field}. Comm. Math. Phys. 4, 145-156 (1967)

\bibitem[Ree88]{Ree88} H. Reeh: \emph{A remark concerning commutation relations}. J.Math.Phys. 29 (7), 
1535 (1988)  

\bibitem[ReSi75]{ReSi75} M. Reed, B. Simon : \emph{Methods of Modern Mathematical Physics}. Vol. 2: 
Fourier Analysis, Self-Adjointness, Academic Press (1975)

\bibitem[Ri14]{Ri14} D. Rickles: \emph{A Brief History of String Theory}. Springer (2014) 

\bibitem[Ro69]{Ro69} P. Roman: \emph{Introduction to Quantum Field Theory}. John Wiley and Sons, Inc. 
(1969) 

\bibitem[Rue11]{Rue11} L. Ruetsche: \emph{Interpretating Quantum Theories}. Oxford University Press (2011), 
p.18

\bibitem[Rue62]{Rue62} D. Ruelle: \emph{On the Asymptotic Condition in Quantum Field Theory}. Helv. Phys. 
Acta, 35, 147 (1962)

\bibitem[Scha14]{Scha14} G. Scharf: \emph{Finite Quantum ELectrodynamics. The Causal Approach}. Dover Publications, Ic. (2014)

\bibitem[Scho08]{Scho08} M. Schottenloher: \emph{A Mathematical Introduction to Conformal Field Theory}. Springer, 2008

\bibitem[Schra74]{Schra74} R. Schrader: \emph{On the Euclidean Version of Haag's Theorem in $P(\varphi)_{2}$
field theories}. Comm. Math. Phys. 36, 133-136 (1974) 

\bibitem[Schra76]{Schra76} R. Schrader: \emph{A possible Constructive Approach to $(\phi^{4})_{4}$}. 
Comm. Math. Phys. 49, 131-153 (1976) 

\bibitem[Schw94]{Schw94} S. Schweber: \emph{QED and the men who made it}.  Princeton University Press (1994)

\bibitem[Seg67]{Seg67} I. Segal: \emph{Notes toward the construction of non-linear relativistic quantum 
fields I. The Hamiltonian in two space-time dimensions as the generator of a C*-automorphism group}. 
Proc. Nat. Am. Soc. 57, 1178 - 1183 (1967)

\bibitem[SinEm69]{SinEm69} K. Sinha, G.G. Emch: \emph{Adaptation of Powers' No-Interaction Theorem to 
Bose Fields}. Bull. Am. Phys. Soc. 14, 86 (1969)

\bibitem[Sok79]{Sok79} A.D. Sokal: \emph{An improvement of Watson's theorem on Borel summability.} J. Math. Phys. 21, 261-263 (1979) 

\bibitem[Sta62]{Sta62} H.P. Stapp: \emph{Axiomatic S-Matrix Theory}. Rev. Mod. Phys. 34, 390-394 (1962)

\bibitem[Stei00]{Stei00} O. Steinmann: \emph{Perturbative Quantum Electrodynamics and Axiomatic Field
Theory}. Springer (2000)

\bibitem[Ster93]{Ster93} G. Sterman : \emph{An introduction to Quantum Field Theory.} Cambridge University Press (1993)

\bibitem[Strau01]{Strau01} N. Straumann: \emph{Schr\"odingers Entdeckung der Wellenmechanik}. 
arXiv: quant-ph/0110097v1 (2001)

\bibitem[Strau13]{Strau13} N. Straumann: \emph{Quantenmechanik}. Springer Verlag (2013)  

\bibitem[Streat75]{Streat75} R.F. Streater: \emph{Outline of axiomatic quantum field theory}. Rep. Prog. Phys.
38, 771-846 (1975)

\bibitem[StreatWi00]{StreatWi00} R.F. Streater, A.S. Wightman: \emph{PCT, Spin and Statistics, and all that}.
Princeton University Press (2000)

\bibitem[Strei68]{Strei68} L. Streit : \emph{A Generalization of Haag's Theorem}. Il Nuovo Cimento 72,
673-680 (1969)

\bibitem[Stro67]{Stro67} F. Strocchi : \emph{Gauge Problem in Quantum Field Theory}. 
Phys. Rev. 162, 1429-1438 (1967)

\bibitem[Stro70]{Stro70} F. Strocchi : \emph{Gauge Problem in Quantum Field Theory III. Quantization of
Maxwell Equations and Weak Local Commutativity}.  
Phys. Rev. 162, 2334-2340 (1970)

\bibitem[Stro93]{Stro93} F. Strocchi: \emph{Selected topics on the general properties of quantum field
theory}. World Scientific (1993)

\bibitem[Stro13]{Stro13} F. Strocchi: \emph{An Introduction to Non-Perturbative Foundations of 
Quantum Field Theory}. Oxford (2013)

\bibitem[StroWi74]{StroWi74} F. Strocchi, A.S. Wightman: \emph{Proof of the charge superselection rule
in local relativistic quantum field theory}. J. Math. Phys. 15, 2198-2224 (1974)

\bibitem[Stue50]{Stue50} E.C.G. Stueckelberg: \emph{Relativistic Quantum Theory for Finite Time Intervals}.
Phys. Rev. 81, 130-133 (1950)

\bibitem[Su01]{Su01} S.J. Summers: \emph{On the Stone-von Neumann Theorem and its Ramifications}, in: 
\emph{John von Neumann and the Foundations of Quantum Physics}, edited by M.Rédei, M.St\"oltzner, Kluwer
Academic Publishers (2001)

\bibitem[Su12]{Su12} S.J. Summers: \emph{A Perspective on Constructive Quantum Field Theory}. 
arXiv: math-ph/1203.3991v1 (2012)

\bibitem[Ti99]{Ti99} Ticciati, R. : \emph{Quantum Field Theory for Mathematicians}. Cambridge University Press (1999)

\bibitem[Wei11]{Wei11} M. Weier : \emph{An Algebraic Version of Haag's Theorem}. Comm. Math. Phys. 
305, 469-485 (2011)

\bibitem[vHo52]{vHo52} L. van Hove: \emph{Les difficultés de divergences pour un modelle particulier de 
champ quantifié}. Physica 18, 145-159 (1952)

\bibitem[vNeu31]{vNeu31} J. von Neumann: \emph{Zur Eindeutigkeit der Schr\"odingerschen Operatoren}. Ann. 
Math. 104, 570-578 (1931)

\bibitem[Wein95]{Wein95} S. Weinberg: \emph{The Quantum Theory of Fields}. Vols.I-III,   
Cambridge University Press (1995)

\bibitem[Wic50]{Wic50} G.C. Wick: \emph{The Evaluation of the Collision Matrix}. Phys. Rev. 80, 268 (1950) 

\bibitem[Wi56]{Wi56} A. Wightman : \emph{Quantum Field Theory in Terms of Vacuum Expectation Values}.
Phys. Rev. 56, 860-866 (1956)

\bibitem[Wi64]{Wi64} A. Wightman: \emph{La théorie quantique locale et la théorie quantique des champs}.
Ann. Inst. H. Poincaré I, 403-420 (1964)

\bibitem[Wi67]{Wi67} A. Wightman : \emph{Introduction to some aspects of the relativistic dynamics of
quantised fields}. In: M. Lévy: \emph{Cargése Lectures in Theoretical Physics: High Energy Interactions
and Field Theory}. Gordon and Breach, pp.171-291 (1967) 

\bibitem[WiHa57]{WiHa57} D. Hall, A. Wightman: \emph{A theorem on invariant analytic functions with 
applications to relativistic Quantum Field Theory}. Danske Videnskabernes Selskab 
Matematisk-Fysiske Meddelelser 31, 1-41 (1957)

\bibitem[WiGa64]{WiGa64} A.S. Wightman, L. G\aa rding : \emph{Fields as operator-valued distributions in 
relativistic quantum theory}. Arkiv f\"or Fysik, 28, 129 - 184 (1964)

\bibitem[YaFe50]{YaFe50} C.N. Yang, D. Feldman: \emph{The S-Matrix in the Heisenberg Representation}. 
Phys. Rev. 79, 972 (1950) 

\bibitem[Yng77]{Yng77} J. Yngvason : \emph{Remarks on the reconstruction theorem for field theories with
indefinite metric}. Rep. Math. Phys. 12, 57 -64 (1977)

\bibitem[Zeid06]{Zeid06} E. Zeidler: \emph{Quantum Field Theory I: Basics in Mathematics and Physics}.
Springer (2006)


\end{thebibliography}
\end{document}